\documentclass[11pt,a4paper]{article}
\pdfoutput=1

\makeatletter
\providecommand{\mdseries@tt}{}
\makeatother

\usepackage[ansinew]{inputenc}
\usepackage{hyperref}
\usepackage{amssymb,amsmath,amsthm}
\usepackage{thmtools, thm-restate}
\usepackage{graphicx}
\usepackage{url}
\usepackage{tabularx}
\usepackage{calc}
\usepackage{paralist} 
\usepackage{enumitem}
\usepackage{caption}
\usepackage{subcaption}
\usepackage{lmodern}
\usepackage{tikz}
\usetikzlibrary{calc}
\usetikzlibrary{decorations.markings}
\usepackage[sort&compress,numbers]{natbib}
\usepackage[linesnumbered,vlined]{algorithm2e}

\usepackage{xspace}
\usepackage{fullpage}

\usepackage{mathabx}

\usepackage{environ}

\newtheorem{theorem}{Theorem}
\newtheorem{proposition}[theorem]{Proposition}

\newtheorem{definition}[theorem]{Definition}
\newtheorem{corollary}[theorem]{Corollary}
\newtheorem{lemma}[theorem]{Lemma}
\newtheorem{claim}{Claim}

\newenvironment{claimproof}{%
  \begin{proof}[Proof of Claim~\theclaim]%
}{%
	\end{proof}%
}

\newcommand{\cf}{\mathsf{cf}}
\newcommand{\Vout}{V_\textsf{out}}
\newcommand{\lmax}{\ell_{\max}}
\newcommand{\dist}{\operatorname{dist}}
\newcommand{\poly}[1]{\operatorname{poly}(#1)}
\newcommand{\polyn}{n^{\mathcal O(1)}}
\newcommand{\Along}{A^\textsf{long}_t}
\newcommand{\Gcirc}{G^\circ_t}
\newcommand{\Cstar}{C^\star_t}
\newcommand{\Gstar}{G^\star_t}
\newcommand{\Voutt}{V^\textsf{out}_t}
\newcommand{\Gtorso}{G^\textsf{torso}_t}
\newcommand{\Ulong}{U^\textsf{long}_t}

\newcommand{\overallruntime}{$2^{\mathcal{O}(\ell k^3\log k + k^5\log k\log \ell)} \cdot \polyn$}

\newcommand{\DLCHS}{\textsc{Directed Long Cycle Hitting Set}}

\newcommand{\algorithmILCHSI}{A_\textsf{ii}}
\newcommand{\setILCHSI}{\mathcal{S}_\textsf{ii}}
\newcommand{\functionILCHSI}{f_\textsf{ii}}

\newcommand{\setkcrit}{\mathcal{S}_\textsf{crit}}
\newcommand{\setdisj}{\mathcal{S}_\textsf{disj}}

\newcommand{\algorithmHS}{A_\textsf{hs}}
\newcommand{\setHS}{\mathcal{S}_\textsf{hs}}
\newcommand{\functionHS}{f_\textsf{hs}}

\newcommand{\setManyClusters}{\mathcal{S}_\textsf{mc}}
\newcommand{\setSeparateClusters}{\mathcal{S}_\textsf{sc}}
\newcommand{\setMultiwayCut}{\mathcal{S}_\textsf{mw cut}}
\newcommand{\setImportantSeparator}{\mathcal{S}_\textsf{imp sep}}

\newcommand{\algorithmCS}{A_\textsf{cs}}
\newcommand{\setCS}{\mathcal{S}_\textsf{cluster}}
\newcommand{\functionCS}{f_\textsf{cs}}

\title{Hitting Long Directed Cycles is Fixed-Parameter Tractable}

\author{Alexander G{\"o}ke\thanks{TU Hamburg, Hamburg, Germany. \texttt{alexander.goeke@tuhh.de}. Supported by DFG grant MN 59/1-1.}
  \and D\'{a}niel Marx\thanks{Max-Planck-Institut f{\"u}r Informatik, Saarbr{\"u}cken, Germany. \texttt{dmarx@mpi-inf.mpg.de}}
  \and Matthias Mnich\thanks{TU Hamburg, Hamburg, Germany. \texttt{matthias.mnich@tuhh.de}. Supported by DFG grant MN 59/4-1.}}

\date{}

\begin{document}

\maketitle

\begin{abstract}
%
%
%
%
    In the {\sc Directed Long Cycle Hitting Set} problem we are given a directed graph $G$, and the task is to find a set $S$ of at most $k$ vertices/arcs such that $G-S$ has no cycle of length longer than $\ell$.
    We show that the problem can be solved in time  \overallruntime{}, that is, it is fixed-parameter tractable (FPT) parameterized by $k$ and $\ell$.
    This algorithm can be seen as a far-reaching generalization of the fixed-parameter tractability of {\sc Mixed Graph Feedback Vertex Set} [Bonsma and Lokshtanov WADS 2011], which is already a common generalization of the fixed-parameter tractability of (undirected) {\sc Feedback Vertex Set} and the {\sc Directed Feedback Vertex Set} problems, two classic results in parameterized algorithms.
    The algorithm requires significant insights into the structure of graphs without directed cycles length longer than $\ell$ and can be seen as an exact version of the approximation algorithm following from the Erd{\H{o}}s-P{\'o}sa property for long cycles in directed graphs proved by Kreutzer and Kawarabayashi [STOC 2015].
\end{abstract}


\thispagestyle{empty}

\clearpage
\pagebreak

\setcounter{page}{1}

\section{Introduction}
\label{sec:introduction}
{\sc Feedback Vertex Set} (FVS) and its directed variant {\sc Directed Feedback Vertex Set} (DFVS) are among the most classical problems in algorithmic graph theory: given a (directed) graph $G$ the task is to find a minimum-size set~$S \subseteq V(G)$ of vertices such that $G - S$ contains no (directed) cycles.
Interestingly, the directed version is {\em not} a generalization of the undirected one. There is no obvious reduction from FVS to DFVS (replacing each undirected edge with two arcs of opposite directions does not work, as this would create directed cycles of length 2).

Both problems received significant amount of attention from the perspective of parameterized complexity.
The main parameter of interest there is the optimal solution size $k = |S|$.
Both problems can easily be solved in time $n^{\mathcal{O}(k)}$ by enumerating all size-$k$ vertex subsets $S\subseteq V(G)$ and then checking whether $G - S$ is acyclic.
The interesting question is thus whether the problems are \emph{fixed-parameter tractable} with respect to $k$, i.e. whether there is an algorithm with run time $f(k)\cdot n^{\mathcal{O}(1)}$ for some computable function $f$ depending only on $k$.
FVS is one of the most studied problems in parameterized complexity: starting in the early 1990's, a long series of improved fixed-parameter algorithms~\cite{CaoEtAl2015,CyganEtAl2011,FominEtAl2008,GuoEtAl2006,KociumakaPilipczuk2014,RamanEtAl2006} lead to the currently fastest (randomized) algorithm from 2020 with run time $2.7^k\cdot n^{\mathcal O(1)}$~\cite{LiNederlof2020}.
The DFVS problem has also received a significant amount of attention from the perspective of parameterized complexity.
It was a long-standing open problem whether DFVS admits such an algorithm; the question was finally resolved by Chen et al. who gave a $4^k k!k^4\cdot \mathcal{O}(nm)$-time algorithm for graphs with $n$ vertices and~$m$ edges.
Recently, an algorithm for DFVS with run time $4^kk! k^5\cdot \mathcal{O}(n+m)$ was given by Lokshtanov et al.~\cite{LokshtanovEtAl2018}.
A fruitful research direction is trying to extend the algorithm to more general problems than DFVS.
On the one hand, Chitnis et al.~\cite{ChitnisEtAl2015} generalized the result by giving a fixed-parameter algorithm for \textsc{Directed Subset FVS}:  here we are given a subset $U$ of arcs and only require the $k$-vertex set $S$ to hit every cycle that contains an arc of $U$. 
On the other hand, Lokshtanov et al.~\cite{LokshtanovEtAl2020} showed that the \textsc{Directed Odd Cycle Transversal} problem, where only the directed cycles of odd length needed to be hit, is $\mathsf{W}[1]$-hard parameterized by solution size.

It is worth noting that very different algorithmic tools form the basis of the fixed-parameter tractability of FVS and DFVS: the undirected version behaves more like a hitting set-type problem, whereas the directed version has a more cut-like flavor.
These differences motivated Bonsma and Lokshtanov~\cite{BonsmaLokshtanov2011} to consider {\sc Mixed FVS}, the common generalization of FVS and DFVS where the input graph contains both directed and undirected edges.
In such \emph{mixed graphs}, cycles can contain directed arcs and undirected edges, but in particular the walk visiting an undirected edge twice is not a cycle.
They obtained an algorithm for {\sc Mixed FVS} with run time $2^{\mathcal O(k\log k)}\cdot n^{\mathcal O(1)}$ for $k$ the size of the smallest feedback vertex set.


In this paper we study the following generalization of DFVS:
We want to find a minimum size vertex set $S$ such that all cycles of $G - S$ to have length at most $\ell$.
For $\ell = 1$ this is DFVS in loopless graphs.
For $\ell = 2$ this is {\sc Mixed FVS} in mixed graphs.
The length of a longest cycle in a (directed) graph is also known as (directed) circumference of a graph.
The parameterized version of our problem thus reads:
\begin{center}
  \framebox[\textwidth]{
    \begin{tabular}{rl}
      \multicolumn{2}{l}{{\sc Directed Long Cycle Hitting Set} \hfill \textit{Parameter:} $k + \ell$.}\\
      \textit{Input:}      & A directed multigraph $G$ and integers $k,\ell\in\mathbb N$.\\
      \textit{Task:}   & Find a set $S$ of at most $k$ vertices such that $G - S$ has circumference at most $\ell$.\\
  \end{tabular}}
\end{center}
Note that \DLCHS{} for $\ell = 2$ generalizes 
{\sc Mixed FVS} (and hence both FVS and DFS): to see this, subdivide anti-parallel arcs to make all cycles have length at least three and then replace undirected edges by anti-parallel arcs.

In contrast to FVS and DFVS, even checking feasibility of a given solution is a non-trivial task.
It amounts to checking, for a digraph $G$ and integer $\ell$, whether $G$ contains a cycle of length more than $\ell$. 
This is also known as the {\sc Long Directed Cycle} problem, which is obviously $\mathsf{NP}$-hard since it contains the {\sc Directed Hamiltonian Cycle} problem for $\ell = |V(G)| - 1$.
However, {\sc Long Directed Cycle} is fixed-parameter tractable parameterized by $\ell$ \cite{Zehavi2016}, hence \emph{verifying} the solution of {\sc Directed Long Cycle Hitting Set} is fixed-parameter tractable in~$\ell$.

\subsection{Our contributions}
\label{sec:ourcontributions}

Our main result is a fixed-parameter algorithm for \DLCHS{}.
\begin{theorem}
\label{thm:boundedcyclelengthdeletion_main}
  There is an algorithm that solves {\sc Directed Long Cycle Hitting Set} in time \overallruntime{} for $n$-vertex directed graphs $G$ and parameters $k,\ell\in\mathbb N$.
\end{theorem}
The result also extends to the arc deletion variant of the problem, as we show both of them to be equivalent in a parameter-preserving way.

The run time in \autoref{thm:boundedcyclelengthdeletion_main} depends on two parameters, $k$ and $\ell$.
This is necessary for the following reason.
For $\ell = 1$, \DLCHS{} corresponds to the DFVS problem, which is $\mathsf{NP}$-hard.
Moreover, the problem is also $\mathsf{NP}$-hard for $k = 0$, as it contains the {\sc Directed Hamiltonian Cycle} problem.
This also shows that the run time cannot be polynomial in $k$ or $\ell$ (unless $\mathsf{P} = \mathsf{NP}$).
Assuming ETH, it is even necessary that the run time depends \emph{exponentially} on both $k$ and $\ell$.
Our algorithm achieves a run time that is single-exponential in both parameters $k$ and $\ell$.
It is, in this sense, optimal.



\subsection{Our methodology}
\label{sec:ourmethodology}
Our algorithm witnessing \autoref{thm:boundedcyclelengthdeletion_main} is based on an elaborate combination algorithmic techniques, some of them used previously, some of them new.
\begin{itemize}
  \item We use the standard opening step of \textbf{iterative compression}, which allows us to assume that every directed cycle of length longer than $\ell$ goes through a small number of exceptional vertices.
  \item We do not want to deal with the situation when there are two exceptional vertices $x$ and~$y$ that are in the same strong component of the directed graph $G-S$ that results from deleting the vertices of a solution $S$.
    If we guess that this happens in the solution $S$, then a way to avoid this problem is to guess a directed cycle $C$ containing both $x$ and $y$, and to contract this cycle.
    In order to guess this cycle, we essentially need a \textbf{representative set} of $x\to y$-paths, that is, a collection of paths such that if an (unknown) set $S$ of at most $k$ vertices does not disconnect $y$ from $x$, then there is at least one $x\to y$-path disjoint from~$S$ in our collection.
    As an interesting self-contained result, we construct such a collection of size $\ell^{\mathcal{O}(k^2\log k)}\cdot \log n$ on directed graphs without cycles of length greater than $\ell$.
  \item If we can assume that the exceptional vertices are in different strong components of the solution, then this defines a separation problem on the exceptional vertices and makes the \textbf{directed shadow removal} technique of Chitnis et al.~\cite{ChitnisEtAl2015} relevant to simplify the structure of the instance.
    In particular, a major structural goal that we want to achieve is to ensure that every arc of the input digraph lies in a directed cycle of length at most $\ell$.
  \item Removing the exceptional vertices breaks the digraph into some number of strong components with no cycle of length longer than $\ell$ in any of them.
    We call \textbf{portal vertices} the endpoints of the arcs connecting these strong components with each other and with the exceptional vertices.
    We show that the portal vertices can be partitioned into \textbf{clusters}: portals in each cluster are close to each other, while the distance between any two clusters is large.
  Furthermore, every solution has to separate the clusters from each other, defining another directed multiway cut problem.
  \item In the final step of the algorithm, we would like to use the technique of \textbf{important separators} to solve the directed multiway cut problem defined above: these are separators that are maximally ``pushed'' towards the target of the separations.
    However, the exact notion of importance is difficult to define due to the additional constraints of the problem being solved.
    To this end, we perform a detailed analysis of the structure of the instance to identify \textbf{outlet vertices} that allows us to represent these additional constraints as separation and to formally reduce the problem to branching on the choice of an important separator.
\end{itemize}

Let us remark that the algorithm can be somewhat simplified in the case of highly connected digraphs, namely, in case of directed graphs which are $(k+1)$-strong.
A major challenge is to get the arguments right when this is not the case, and some vertices are connected only by few vertex-disjoint paths in either direction.
For this situation, no general algorithmic tools are available for such directed graphs.
So again, the situation is more complicated than that in undirected graphs, for which it is known how to reduce large classes of problems to solving them on highly-connected graphs~\cite{RamanujanEtAl2018}.
In response, we provide a fine-grained analysis of the combinatorics of how long cycles interact with small cut sets, to let the branching process make progress on the instance.

\subsection{Related work}
\label{sec:relatedwork}
The structure of long cycles in digraphs has been of interest for long time.
For instance, Lewin~\cite{Lewin1975} analyzed the density of such graphs, and Kintali~\cite{Kintali2017} analyzes the directed treewidth of such directed graphs.
Algorithmically, though, it was only recently shown by Kawarabayashi and Kreutzer~\cite{KawarabayashiKreutzer2015} that the vertex version of the Erd{\H{o}}s-Posa property holds for long directed cycles: namely, they show that any directed graph~$G$ either contains a set of $k+1$ vertex-disjoint directed cycles of length at least~$\ell$ or some set~$S$ of at most $f(k,\ell)$ vertices that intersects all directed cycles of~$G$ with length at least~$\ell$.
The corresponding questions for directed cycles without length restrictions have also been well-investigated~\cite{AmiriEtAl2016,ReedEtAl1996}.
%
%
%
%

Note that an algorithmic proof of the Erd{\H{o}}s-Posa property can be a useful opening step for a fixed-parameter algorithm: we either find a set of $k+1$ arc- or vertex-disjoint cycles of length at least $\ell$ (and thus can reject the instance $(G,k,\ell)$ as ``no''-instance) or obtain a set $S$ which can serve as a feasible approximate solution.
Such an opening step was also discussed in the well-known fixed-parameter algorithm for DFVS by Chen et al.~\cite[Remark 5.3]{ChenEtAl2008}, where the function $f(k,1)$ is known to be near-linear.
In our case though, the function $f(k,\ell)$ from the Kawarabayashi-Kreutzer result is way too large for us to obtain an algorithm for \DLCHS{} with run time $2^{\poly{k, \ell}}\cdot \polyn$.
%

We further point out that it is even $\mathsf{NP}$-hard to \emph{verify} that $G - S$ has the desired property of having its directed circumference bounded by at most $\ell$; this is another difference compared to the DFVS problem.
For checking if the digraph $G - S$ does not have cycles of length \emph{exactly} $\ell+1$, the well-known color-coding technique of Alon et al.~\cite{AlonEtAl1995} can be employed to give the correct answer in time $2^{\mathcal O(\ell)}\cdot n^{\mathcal O(1)}$.
But as we also need to refute the existence of cycles with lengths $\ell + 2,\ell + 3,\hdots$ and in general of cycles whose length is not bounded by $\ell$, more sophisticated techniques are needed; Zehavi~\cite{Zehavi2016} provides a deterministic algorithm for this purpose with run time $2^{\mathcal O(\ell)}\cdot n^{\mathcal O(1)}$.

Finally, directed circumference can be seen as an intermediate step towards a general algorithmic framework for graph optimization problems related to \emph{directed treewidth}.
In undirected graphs, treewidth as a graph measure has enjoyed unprecedented success as a tool towards efficient approximation algorithms and fixed-parameter algorithms.
For instance, as part of their Graph Minors series, Robertson and Seymour~\cite{RobertsonSeymour1995} showed that the $k$-linkage problem is fixed-parameter tractable, heavily relying on the reduction of the problem to graphs of bounded treewidth.
While a famous algorithm by Bodlaender~\cite{Bodlaender1996} shows that graphs of bounded treewidth can be recognized in linear time, it was only shown recently, by Fomin et al.~\cite{FominEtAl2012}, how to recognize graphs of nearly-bounded treewidth, i.e. graphs that have bounded treewidth after deleting at most $k$ arcs or vertices from it.
Yet in directed graphs, the situation is again much more complicated: Johnson et al.~\cite{JohnsonEtAl2001} introduced the notion of directed treewidth for digraphs.
Yet, for digraphs the $k$-linkage problem is $\mathsf{NP}$-hard already for $k = 2$, and no fixed-parameter algorithm is known which recognizes digraphs of nearly-bounded directed treewidth.
On the positive side, though, digraphs of bounded directed circumference are nicely squeezed between acyclic digraphs and digraphs of bounded directed treewidth~\cite{Kintali2017}.
Moreover, the arc version of the $k$-linkage problem is fixed-parameter tractable on digraphs of directed circumference~2~\cite{BangJensenLarsen2016}; the question remains open for digraphs of arbitrary directed circumference.


\medskip

\noindent
\textbf{Structure of the paper.}
We define the basic terms and symbols in \autoref{sec:notionsandnotations}.
The combinatorial properties of digraphs with bounded circumference, which are necessary for our algorithm, we collect in~\autoref{sec:technical_tools}.
Then we give the fixed-parameter algorithm for the vertex-deletion version of the problem, in \autoref{sec:bcl_thealgorithm}.
In that section we start with an overview of the algorithm, and then work in a completely modular way: at the end of each subsection, we summarize the state of the algorithm in a concise statement which then forms the starting point of the next subsection.
In \autoref{sec:technical_tools_proofs}, we provide proofs for the theorems of \autoref{sec:technical_tools}.
Those combinatorial insights can thus be read independently of the algorithm, and potentially be used for solving other algorithmic problems on digraphs of bounded circumference.
In \autoref{sec:boundedcyclelength}, we reduce the arc deletion version to the vertex deletion version, as well as {\sc Mixed FVS}.
Finally, we conclude in \autoref{sec:discussion}.

\section{Notions and Notations}
\label{sec:notionsandnotations}
In this paper, we mainly consider finite loop-less directed graphs (or digraphs) $G$ with vertex set~$V(G)$ and (directed) arc set $A(G)$.
We allow multiple arcs and arcs in both directions between the same pairs of vertices.
A \emph{walk} is a sequence of vertices $(v_1,\hdots,v_\ell)$ with corresponding arcs $(v_i,v_{i+1})$ for $i = 1,\hdots,\ell-1$ which forms a subgraph of $G$; the \emph{length} of a walk is its number of arcs.
A walk is \emph{closed} if $v_1 = v_\ell$; otherwise, it is \emph{open}.
A \emph{path} in $G$ is an open walk where all vertices are visited at most once.
A \emph{cycle} in $G$ is a closed walk in which every vertex is visited at most once, except for $x_1 = x_\ell$ which is visited twice.
(Throughout this entire paper, by ``cycle'' we always mean directed cycle.)
We call $G$ \emph{acyclic} if $G$ does not contain any cycle.
For two vertices $x_i, x_j$ of a walk $W$ with $i  \leq j$ we denote by $W[x_i, x_j]$ the subwalk of $W$ starting at~$x_i$ and ending in $x_j$.
For a walk $W$ ending in a vertex~$x$ and a second walk $R$ starting in $x$, $W \circ R$ is the walk resulting when concatenating~$W$ and~$R$.

We say that $y$ is \emph{reachable from $x$} in $G$ if there is a directed path from $x$ to $y$ in $G$.
The \emph{distance} $\mathsf{dist}_G(x,y)$ between any two vertices $x,y\in V(G)$ in~$G$ is the minimum length of a directed path from~$x$ to $y$ in $G$.
We say that a set $S$ \emph{separates $y$ from $x$} in $G$ if $y$ is not reachable from $x$ in $G - S$.

For each vertex $v\in V(G)$, its \emph{out-degree} in $G$ is the number $d^+_G(v)$ of arcs of the form $(v,w)$ for some $w\in V(G)\setminus\{v\}$, and its \emph{in-degree} in $G$ is the number $d^-_G(v)$ of arcs of the form $(w,v)$ for some $w\in V(G)\setminus\{v\}$.
For each subset $V'\subseteq V(G)$, the subgraph induced by~$V'$ is the graph~$G[V']$ with vertex set $V'$ and arc set $\{(u,v)\in A(G)~|~u,v\in V'\}$.
For a set $X$ of vertices or arcs, let $G - X$ denote the subgraph of $G$ obtained by deleting the elements in $X$ from~$G$.
For a subgraph~$G'$ and an integer $d$ we denote by $R^+_{G', d}(X)$ the set of vertices that are reachable from $X$ in $G'$ by a path of length at most $d$.
We omit $G'$ if it is clear from the context, and omit~$d$ if $d \geq |V(G)|$ (when the path length is unrestricted).

A digraph $G$ is called \emph{strong} if either $G$ consists of a single vertex (then $G$ is called \emph{trivial}), or for any distinct $u,v\in V(G)$ there is a (directed) path from $u$ to $v$.
A \emph{strong component} of~$G$ is an inclusion-wise maximal induced subgraph of $G$ that is strong.
The \emph{(directed) circumference} of a digraph $G$ is the length $\mathsf{cf}(G)$ of a longest cycle of $G$; if $G$ is acyclic, then define $\mathsf{cf}(G) = 0$.

\section{Technical Tools}
\label{sec:technical_tools}
This section is a collection of important structural properties of separators and bounded circumference graphs which we use in our algorithm.
  The statements themselves are presented here, whereas the proofs can be found in \autoref{sec:technical_tools_proofs}.

\subsection{Important Separators and Consequences}
An important tool in the design of parameterized graph modification algorithms are important separators.

\begin{definition}[directed separator]
  Let $G$ be a digraph.
  For disjoint non-empty sets $X,Y\subseteq V(G)$ a set $S$ is an \emph{$X - Y$-separator} if $S$ is disjoint from $X\cup Y$ and there is no path from~$X$ to~$Y$ in $G - S$.
  An $X - Y$-separator $S$ is \emph{minimal} if no proper subset of $S$ is an $X - Y$-separator.
  An $X - Y$-separator~$S$ is \emph{important} if there is no $X - Y$-separator $S'$ with $|S'|\leq |S|$ and $R^+_{G - S}(X) \subsetneq R^+_{G - S'}(X)$.
\end{definition}
Notice that $S$ can be either a vertex set or an arc set.
A standard result on important separators of size $\leq k$ is that there cannot exist to many of them.
More precisely:

\begin{proposition}[\cite{ChitnisEtAl2012}]
\label{thm:fptenumerationofimportantseperators}
  Let $G$ be a digraph and let $X,Y\subseteq V(G)$ be disjoint non-empty vertex sets.
  For every $p\geq 0$ there are at most $4^p$ important $X - Y$-separators of size at most~$p$, and all these separators can be enumerated in time $\mathcal{O}(4^p\cdot p (n + m))$.
\end{proposition}

Using this, we can establish results bounding the number of vertices defining a separator.

\begin{restatable}{lemma}{restatesmallpathwitness}
\label{thm:bcl_smallpathwitness}
  Let $G$ be a digraph, let $x\in V(G)$, let $Y\subseteq V(G)$, and let $k\in\mathbb N$.
  Then in time $2^{\mathcal{O}(k)}\cdot n^{\mathcal{O}(1)}$ we can identify a set $Y'\subseteq Y$ of size at most $(k+1)4^{k+1}$ with the following property:
  \begin{equation}
  \label{eqn:bcl_impsepproperty}
  \tag{$\dagger$}
    \parbox{0.9\textwidth}{
    \emph{if $S\subseteq V(G)$ is a set of at most $k$ vertices such that there is an $x\rightarrow Y$-path in $G - S$, then there is also a $x\rightarrow Y'$-path in $G - S$.}
    }
  \end{equation}
\end{restatable}

\begin{restatable}{lemma}{restatesmallpathwitnessset}
\label{thm:bcl_smallpathwitnessset}
  Let $G$ be a digraph, let $X,Y\subseteq V(G)$ be sets of vertices, and let $k\in\mathbb N$.
  Then in time $2^{\mathcal{O}(k)}\cdot n^{\mathcal{O}(1)}$ we can identify sets $X'\subseteq X,Y'\subseteq Y$ each of size at most $(k+1)4^{k+1}$ such that the following holds:
  If $S\subseteq V(G)$ is a set of at most $k$ vertices such that there is an $X \rightarrow Y$-path in $G - S$, then there is also an $X' \rightarrow Y'$-path in $G - S$.
\end{restatable}

\begin{restatable}{lemma}{restatesmallpathwitnesstsets}
\label{thm:bcl_smallpathwitnesstsets}
  Let $G$ be a digraph, let $X_1,\hdots,X_t \subseteq V(G)$ be sets of vertices, and $k\in\mathbb{N}$.
  Then in time $t^22^{\mathcal{O}(k)}\cdot n^{\mathcal{O}(1)}$ we can identify sets $X'_i \subseteq X_i$ of size at most $2(t-1)(k+1)4^{k+1}$ for every $i \in \{1, \hdots, t\}$, such that the following holds:
  If $S\subseteq V(G)$ is a set of at most $k$ vertices such that there is an $X_i \rightarrow X_j$-path in $G - S$ for some $i \not = j$, then there is also an $X'_i \rightarrow X'_j$-path in $G - S$.
\end{restatable}

\subsection{Properties of Directed Graphs with Bounded Circumference}
As we are interested in a vertex set whose deletion leads to a digraph of bounded circumference, it is useful to study the properties of this class of graphs.
One of the main observations is that paths going in both directions between two vertices cannot differ in their length by more than a factor of $\cf(G) - 1$.

\begin{restatable}{lemma}{restatedistboundedcircum}
\label{thm:bcl_distboundedcircum}
  Let $G$ be a digraph and let $x,y\in V(G)$.
  If $P_1$ is an $x\rightarrow y$-path and $P_2$ is a $y\rightarrow x$-path, then $|P_1|\leq (\mathsf{cf}(G)-1)|P_2|$.
  Consequently, we have $\mathsf{dist}_G(x,y)\leq (\mathsf{cf}(G)-1)\mathsf{dist}_G(y,x)$.
\end{restatable}

By using that there is always a backward path in strong digraphs, applying above result twice yields:

\begin{restatable}{lemma}{restatedistboundedcircumsquared}
\label{thm:bcl_distboundedcircumsquared}
  Let $G$ be a strong digraph and $x,y\in V(G)$.
  Then $|P_1|\leq (\mathsf{cf}(G)-1)^2\mathsf{dist}_G(x,y)$ for every $x\rightarrow y$-path~$P_1$.
\end{restatable}

We now give versions of the above lemmas, that consider the case where start- and endpoints of two paths are not identically but at close distance.

\begin{restatable}{lemma}{restatedistsecondpath}
\label{thm:bcl_distsecondpath}
  Let $G$ be a strong digraph, $x,y\in V(G)$ two vertices, and $P_1,P_2$ be two $x\rightarrow y$-paths.
  For every vertex $v$ of $P_1$, we have $\mathsf{dist}_G(P_2,v)\leq 2(\mathsf{cf}(G) - 2)$ and $\mathsf{dist}_G(v,P_2)\leq 2(\mathsf{cf}(G) - 2)$.
\end{restatable}

\begin{restatable}{lemma}{restatepathdistance}
\label{thm:bcl_pathdistance}
  Let $G$ be a strong digraph, let $P_1$ be an $x_1\rightarrow y_1$-path and $P_2$ be an $x_2\rightarrow y_2$-path such that $\mathsf{dist}_G(x_1,x_2)\leq t$ and $\mathsf{dist}_G(y_1,y_2)\leq t$ for some integer $t$.
  Then every vertex of $P_1$ is at distance at most $(\mathsf{cf}(G) - 1)t + 2(\mathsf{cf}(G) - 2)$ from $P_2$.
\end{restatable}

\subsection{Bypassing}

The properties of strong digraphs with bounded circumference also lead to a result on how to switch between two close paths while avoiding a deletion set of bounded size.

\begin{restatable}{lemma}{restatebypassingthroughnearbypaths}
\label{thm:bypassing_through_nearby_paths}
  Let $G$ be a strong digraph, let $S\subseteq V(G)$ be a set of at most $k$ vertices, let $P_1$ be an $x_1\rightarrow y_1$-path and $P_2$ be an $x_2\rightarrow y_2$-path such that $\mathsf{dist}_G(x_1,x_2)\leq t$ and $\mathsf{dist}_G(y_1,y_2)\leq t$ for some integer $t$.
  Let $P_1[a, b]$ be an subpath of $P_1$ of length at least $\cf(G)^5 \cdot (t + 2)k$ that is disjoint from $S$.
  If $P_2$ is disjoint from $S$, then there is an $x_2 \to b$-path in $G - S$.
\end{restatable}

\subsection{Representative Sets of Paths}
Last but not least we were able to obtain a nice self-contained result on so called representative sets of path.

\begin{restatable}{definition}{DefinitionRepresentativeSetOfPaths}
\label{def:k_representative_set_of_paths}
  Let $G$ be a digraph, $x, y \in V(G)$ and $k \in \mathbb{Z}_{\geq 0}$.
  A set $\mathcal{P}$ of $x \to y$-paths is a \emph{$k$-representative set of $s \to t$-paths}, if for every set $S \subseteq V(G)$ of size at most $k$ the following holds:
  If there is an $x \to y$-path in $G - S$ there is an $x \to y$-path $P \in \mathcal{P}$ that is disjoint from $S$.
\end{restatable}

On can think of a $k$-representative set of $x \to y$-paths as certificate whether or not a vertex set of size at most $k$ separates $y$ from $x$.
The goal is to find such a set with small size.
For strongly connected digraphs of bounded circumference we were able to obtain such a result.

\begin{restatable}{lemma}{restatecomputerepresentativesetofpathsforboundedcircum}
\label{thm:bcl_computerepresenetativeforboundedcf}
  Let $G$ be a strong digraph, let $x,y\in V(G)$, and let $k\in\mathbb N$.
  In time $\mathsf{cf}(G)^{\mathcal{O}(k^2\log k)}\cdot n^{\mathcal{O}(1)}$, we can compute a $k$-representative set $\mathcal P_{x,y,k}$ of $x\rightarrow y$-paths of size $\mathsf{cf}(G)^{\mathcal{O}(k^2\log k)}\cdot \log n$.
\end{restatable}

The above lemma is stated in a self-contained way, as we think this tool may be of independent interest.
We also derive a version fine-tuned to our needs:

\begin{restatable}{lemma}{restatestongplusW}
\label{lem:stongplusW}
  Let $G$ be a digraph, let $W\subseteq V(G)$ be a set for which $\cf(G-W)\le \ell$, and let $k\in\mathbb N$.
  Then in time $2^{\mathcal O(k\ell + k^2\log k)}\cdot n^{\mathcal O(1)}$, we can compute a collection $\mathcal{Q}$ of $|W|^22^{\mathcal O(k\ell + k^2\log k)}\log^2 n$ closed walks in $G$, each containing at least two members of $W$, such that the following holds: if $S\subseteq V(G)$ is a set of at most $k$ vertices in $G$ such that $\cf(G-S)\le \ell$ and $G-S$ has a strong component containing at least two vertices of $W$, then either there is a simple cycle of length at most $\ell$ containing at least two vertices of~$W$ or a closed walk in $\mathcal{Q}$  disjoint from $S$.
\end{restatable}

\section{Directed Long Cycle Hitting Set Algorithm}
\label{sec:bcl_thealgorithm}
In this section we will present our fixed-parameter algorithm for hitting all long cycles of the input digraph.
Formally, we wish to solve the following problem:
\begin{center}
  \framebox[\textwidth]{
    \begin{tabular}{rl}
      \multicolumn{2}{l}{\DLCHS} \hfill\textit{Parameter:} $k + \ell$.~\\
      \textit{Input:}      & A directed multigraph $G$ and integers $k,\ell\in\mathbb N$.\\
      \textit{Task:}   & Find a set $S$ of at most $k$ arcs/vertices such that $G - S$ has circumference at most $\ell$.\\
  \end{tabular}}
\end{center}

\subsection{Algorithm Outline}
Our algorithm will only solve the vertex variant of \DLCHS{} problem.
  This will suffice, as the arc deletion version can be reduced to the vertex deletion version in a parameter-preserving way, as we will show in \autoref{thm:bcl_reductionarctovertex} (c.f. \autoref{sec:boundedcyclelength}).

The algorithm performs a sequence of reductions to special cases of the original {\sc Bounded Cycle Length Vertex Deletion}.
All these sections are modular and just need the problem formulation and theorem at the end of the previous section.

The overall algorithm can be described as follows:
In the first section, \autoref{sec:compressionandcontraction}, we apply the standard technique of iterative compression to our problem to solve the easier problem where we are already given a solution $T$ and search for a smaller solution $S$.
This is further refined by a contraction argument such that~$T$ has at most one vertex in every strong component of~$G - S$.

In the following section, \autoref{sec:reductiontostrongsubgraph}, we use this to find a strong subgraph $G^\star$ of $G$ which contains exactly one vertex $t$ of $T$.
We then reduce our compression problem to finding a set~$S$ in $G^\star$ which intersects all long cycles (of length more than $\ell$) in $G^\star$ and is additionally some $t \to V_\text{out}$-separator for some appropriate vertex set $V_\text{out}$.
In \autoref{sec:findinghittingseparators} we then reduce this problem to finding all important $t \to V_\text{out}$-``cluster separators''.
The concept of cluster separators (which we introduce here)  allows us describe the structure of $S$ in every strong component of $G - t$.
In the reduction, we make use of an algorithm for the {\sc Directed Multiway Cut} problem on a certain digraph with specific terminal sets.
The {\sc Directed Multiway Cut} problem can be solved in time $2^{\mathcal O(p^2)}\cdot n^{\mathcal O(1)}$ to find solutions of size at most $p$.
On instances which we cannot phrase as a \textsc{Directed Multiway Cut} problem, we finally find our cluster separators in \autoref{sec:finding_important_cluster_separators} with the help of important $t \to V_\text{out}\cup V_\Omega$-separators, for some further vertex set~$V_\Omega$.

Some of our reductions construct several instances of the reduced problem such that the reduction holds for at least one of them.
These reductions allow us to make assumption about a hypothetical solution $S$.

The procedure will give us a solution candidate for every instance.
To check whether a candidate~$S$ is indeed a solution, we have to test if $|S| \leq k$ and if $G - S$ contains cycles of length more than~$\ell$.
This can be done by an algorithm of Zehavi~\cite{Zehavi2016}:
\begin{theorem}[\cite{Zehavi2016}]
\label{thm:checkforsolution}
  There is an algorithm that decides in time $2^{\mathcal{O}(\ell)} \cdot n^{\mathcal{O}(1)}$  whether a digraph~$G$ contains a cycle of length more than $\ell$.
\end{theorem}

\subsection{Compression and Contraction}
\label{sec:compressionandcontraction}
The goal of this subsection is to get an existing solution $T$ for which we have to find a disjoint solution~$S$ of size less than $|T|$.
For this we use the standard techniques of iterative compression and disjoint solution.
Additionally we will contract some closed walks in $G$ such that every strong component of $G - S$ contains at most one vertex of $T$.

%
%

We first apply iterative compression to our instance:
For this, we label the vertices of the input digraph~$G$ arbitrarily by $v_1,\hdots,v_n$, and set $G_i = G[\{v_1,\hdots,v_i\}]$.
To find a solution of size at most~$k$, we start with the digraph $G_1$ and the solution $S_1 = \{v_1\}$.
For $i \geq 2$, as long as $|S_{i-1}| < k$, we can set $S_i = S_{i-1} \cup \{v_i\}$ as a solution (of size at most $k$) for $G_i$ and continue.
If $|S_{i-1}| = k$, set $T_i = S_{i-1} \cup \{v_i\}$ is a solution for $G_i$ of size $k+1$.
Given such a pair $(G_i,T_i)$, we then wish to solve the following ``compression variant'' of our problem:
%
\emph{Given a digraph $G$ and a solution $T$ of size $k+1$, find a solution $S$ of size at most $k$ or prove that none exists.}

If there is an algorithm $\mathcal A'$ to solve this problem, we can call it on $(G_i, T_i)$ to obtain a solution~$S_i$ of size at most $k$ or find that $G_i$ does not have a solution of size $k$, but then neither has $G = G_n$.

\begin{lemma}[safety of iterative compression]
\label{thm:safetyofiterativecompression}
  Any instance of {\sc Directed Long Cycle Hitting Set} of size $n$ can be solved by $n$ calls to an algorithm for the problem's compression variant.
\end{lemma}

%

The next step is to ask for a solution $S_i$ for $G_i$ of size $k$ that is disjoint from $T_i$.
This assumption can be made by guessing the intersection $S_i\cap T_i$, and fix those vertices for any solution of $G_i$ (by deleting these vertices from~$G_i$ and decreasing the budget $k$ by the number of deleted vertices).
For each guess we create a new instance where we assume that $S_i$ and~$T_i$ are disjoint.
Since $T$ has $k+1$ elements, we produce at most~$2^{k+1}$ instances.

The disjoint compression variant of the problem is the same as the problem's compression variant, except that the sought solution should be disjoint from $T$.

\begin{lemma}[adding disjointness]
\label{thm:addingdisjointness}
  Instances of the compression variant of {\sc Directed Long Cycle Hitting Set} can be solved by $\mathcal{O}(2^{|T|})$ calls to an algorithm for the problem's disjoint compression variant.
\end{lemma}

So henceforth we will consider the following problem:
Given a digraph $G$, integers $k, \ell \in\mathbb N$ and a set $T \subset V(G)$ of size $k+1$ such that $\cf(G - T) \leq \ell$, the task is to find a set $S$ disjoint from $T$ of size $|S| < |T|$ for which $\cf(G - S) \leq \ell$.

The last reduction in this section is to give $T$ a particular nice structure with respect to~$S$.
Namely, we want to achieve that every strong component contains at most one vertex of~$T$.

\begin{definition}
  Let $G$ be a digraph, $\ell \in \mathbb{N}$ and $T \subset V(G)$ be some subset of vertices.
  A vertex set $U \subseteq V(G) \setminus T$ is called \emph{isolating long cycle hitting set} (with respect to $T$) if $\cf(G - U) \leq \ell$ and every strong component contains at most one vertex of $T$.
\end{definition}

For transforming our solutions to isolating long cycle hitting sets we try to contract closed walks containing several vertices of $T$.
To ensure that this does not harm the circumference or our solution we use the following lemma:

\begin{lemma}
\label{thm:bcl_smallorlargecycle}
  Let $G$ be a digraph and let $X\subseteq V(G)$ be such that $G[X]$ is strong and \mbox{$\mathsf{cf}(G[X])\leq\ell$}.
  Suppose that the following two properties hold:
  \begin{compactenum}
    \item[(P1)] Every cycle of $G$ has length at most $\ell$ or length at least $\ell^2$.
    \item[(P2)] For any $a,b\in X$ there can not be both an $a\rightarrow b$-path $P_{ab}$ of length at least $\ell$ in $G[X]$ and a $b\rightarrow a$-path~$P_{ba}$ of length at most $\ell$ in $G - (X \setminus \{a,b\})$.
  \end{compactenum}
  Let $G/X$ be the digraph obtained by contracting $X$ to a single vertex $x$.
  \begin{compactenum}
    \item If $\mathsf{cf}(G - S) \leq \ell$ for some $S\subseteq V(G)\setminus X$, then $\mathsf{cf}(G/X - S)\leq \ell$.
    \item If $\mathsf{cf}(G/X - S')\leq \ell$ for some $S'\subseteq V(G/X)\setminus\{x\}$, then $\mathsf{cf}(G - S')\leq\ell$.
  \end{compactenum}
\end{lemma}
\begin{proof}
  For Statement 1, suppose that $G/X - S$ has a cycle $C$ of length more than $\ell$.
  If $C$ does not go through~$x$, then $C$ is a cycle of $G$ disjoint from $S$.
  Otherwise, if $C$ goes through~$x$, then the arcs of~$C$ correspond to a walk of $G$ going from some vertex $x_1\in X$ to a vertex $x_2\in X$ and having length more than~$\ell$.
  If $x_1 = x_2$ then this walk is a cycle of length more than $\ell$ in $G - S$, a contradiction.
  Suppose therefore that $x_1\not=x_2$, in which case this walk is a simple path $P$.
  As $G[X]$ is strong, there is an $x_2\rightarrow x_1$-path $Q$ in $G[X]$.
  The paths $P$ and~$Q$ create a cycle in~$G$ that is disjoint from $S$ and has length more than $\ell$, a contradiction.
  
  For Statement 2, suppose that $G - S'$ has a cycle $C$ of length greater than $\ell$.
  Let us choose~$C$ such that it has the minimum number of vertices outside $X$.
  By assumption, cycle $C$ cannot be fully contained in $X$.
  If $C$ is disjoint from $X$, then $C$ is a cycle of $G/X$ disjoint from $S'$, a contradiction.
  If~$C$ contains exactly one vertex of $X$, then there is a corresponding cycle~$C'$ in the contracted digraph with the same length and disjoint from $S'$, a contradiction.
  Assume therefore that~$C$ contains more than one vertex of $X$; let $P$ be an $x_1\rightarrow x_2$-subpath of $C$ with both endpoints in~$X$ and no internal vertex in $X$.
  If $P$ has length more than~$\ell$, then there is a corresponding cycle~$C'$ in the contracted digraph with the same length and disjoint from~$S'$, a contradiction.
  Let~$v$ be an arbitrary internal vertex of $P$ and let $G^\star = G[X\cup V(C)\setminus\{v\})]$.
  By the minimality of the choice of~$C$, we have $\mathsf{cf}(G^\star)\leq \ell$.
  As $G^\star[X] = G[X]$ is strong, it contains an $x_2\rightarrow x_1$-path~$P_1$.
  Also, the subpath $P_2$ of $C$ from $x_1$ to $x_2$ is in $G^\star$.
  By property~(P1), the length of $C$ is at least $\ell^2$, hence $|P_2| = |C| - |P|\geq \ell^2 - \ell$.
  Thus, \autoref{thm:bcl_distboundedcircum} implies that $|P_1|\geq |P_2|/(\ell - 1)\geq \ell$.
  However, this means that~$P$ and $P_1$ contradict property (P2).
\end{proof}

To get candidates for our strong subgraph $G[X]$ we make use of several techniques including important separators and representative sets of paths.
The technical details can be found in \autoref{sec:bcl_representativesetsofpaths}.
The result can be summarized in the following lemma:

\restatestongplusW*

We now combine \autoref{lem:stongplusW} and \autoref{thm:bcl_smallorlargecycle}.
Ideally, we would like to compute $\mathcal{Q}$ as in \autoref{lem:stongplusW} for our digraph $G$ and vertex set $T$ and use the walks inside $\mathcal{Q}$.
Alas, our set $\mathcal{Q}$ only contains a closed walk connecting two vertices of $T$ in $G - S$ (given such a walk exists) if we guarantee that there is no cycle of length at most $\ell$ containing at least two vertices of $T$ which is disjoint from $S$.
Therefore, we have to check for such cycles beforehand.
Also, we cannot use \autoref{thm:bcl_smallorlargecycle} directly, as the second condition may not be fulfilled.
We handle both issues via the following lemma:

\begin{lemma}
\label{thm:bcl_unifyalg}
  Let $(G, k, \ell, T)$ be an instance of the disjoint compression variant of {\sc Directed Long Cycle Hitting Set}.
  There is an algorithm that in time $2^{\mathcal{O}(\ell k^2\log k)} \cdot \polyn$ branches in $|T|^2 2^{\mathcal{O}(\ell k^2\log k)}\log^2 n$ directions.
  If $(G, k, \ell, T)$ did not already contain an isolating long cycle hitting set, one of the branches reduces either the parameter $k$ (by identifying some vertex in~$S$) or the number of vertices in~$T$.
\end{lemma}
\begin{proof}
  Assume $(G, k, \ell, T)$ does not have an isolating long cycle hitting set, i.e. for every solution~$S$ one of the components of $G - S$ contains to vertices of $T$.
  First, we check whether there is a  cycle of length at most $\ell$ visiting at least two vertices of $T$.
  This can be done by standard color-coding techniques in time $2^{\mathcal{O}(\ell)} \cdot n^{\mathcal{O}(1)}$.
  If we found such a cycle $C$ it either intersects $S$ or is disjoint from it.
  We branch in $|C| + 1$ directions.
  In the first $|C|$ branches we choose a vertex of $C$ and delete it, as we guess it to be in $S$.
  For these branches we can decrease $k$ by one and are done.
  In the remaining branch we may assume that~$C$ is disjoint from $S$.
  We set $\mathcal{Q} = \{C\}$ and continue with our algorithm.
  If there is no such cycle, we compute~$Q$ with the algorithm from \autoref{lem:stongplusW}.

  Now we know that there is a closed walk $C \in \mathcal{Q}$ that contains a least two vertices of $T$ and is disjoint from $S$.
  By definition, this walk lies inside one strong component of $G - S$.
  We branch for every $C \in \mathcal{Q}$ and assume in the following we have picked the right walk $C$.

  Now we make sure the conditions of \autoref{thm:bcl_smallorlargecycle} are fulfilled.
  By the same color-coding techniques as above we can detect cycles $C$ of length $\ell'$ with $\ell < \ell' \leq \ell^2$.
  We know that one of the vertices of~$C$ has to be in $S$ and we branch on deleting one of them.

  Our algorithm then checks for every $a, b \in V(C)$ whether an $a \rightarrow b$-path $P_{ab}$ in $G[V(C)]$ of length at least $\ell$ exists and whether a $b \rightarrow a$-path $P_{ba}$ in $G - (V(C) \setminus \{a,b\})$ of length at most $\ell$ exists.
	
  If for some $a, b \in V(C)$ both paths exists, the closed walk $W = P_{ab} \circ P_{ba}$ is in fact a cycle, as the paths only intersect in $a$ and $b$.
  As $P_{ab}$ has length at least $\ell$, the cycle $W$ has to be intersected by $S$.
  As $V(P_{ab}) \subset V(C) \subset V(G) \setminus S$, we know that $S$ has to intersect $V(P_{ba})$, which has size at most $\ell$.
  We branch on deleting a vertex of $V(P_{ba})$ and reducing $k$ by one (as we have guessed a vertex of $S$).
	
  Otherwise, we have no such paths for any $a, b \in V(C)$.
  But then we can apply \autoref{thm:bcl_smallorlargecycle} to contract~$C$ to a single vertex $v_C$.
  By setting $G' = G/V(C)$ and $T' = (T \setminus V(C)) \cup \{v_C\}$ we obtain a new instance $(F', T', k, \ell)$ of our disjoint reduction problem which is guaranteed to have the same solution as the original instance.
  For the correct branch $C$, we know that $|T \cap V(C)| > 1$ and these vertices were in the same strong component of $G - S$.
  Therefore, we reduced the size of $T$ by at least one.
	
  The number of branches is dominated by the number of walks in $\mathcal{Q}$ times $\ell$, as we might have to branch on deleting vertices from $V(P_{ba})$.
\end{proof}

\begin{corollary}
\label{cor:isolating_t}
  Let $(G, k, \ell, T)$ be an instance of the disjoint compression variant of {\sc Directed Long Cycle Hitting Set}.
  There is an algorithm that in time $2^{\mathcal{O}((k + |T|) \ell k^2\log k)}\cdot \polyn$ generates $2^{\mathcal{O}((k + |T|) \ell k^2\log k)}\cdot\log^{2|T|} n$ instances $(G_i, k_i, \ell, T_i)$ with $|V(G_i)| \leq |V(G)|, k_i \leq k$ and\linebreak $|T_i| \leq T$ such that if $(G, k, \ell, T)$ has a solution, one of the instances $(G_i, k_i, \ell, T_i)$ has an isolating hitting set of size at most $k_i$.
\end{corollary}
\begin{proof}
	Call \autoref{thm:bcl_unifyalg} at most $k + |T| -1$ levels deep, and for every call keep one instance unchanged.
	Fix a solution $S$ of $(G, k, \ell, T)$.
	In every level, one of the branches
	\begin{compactitem}
		\item either has an isolating long cycle hitting set,
		\item or $k$ is reduced by one (by finding a vertex of $S$),
		\item or $|T|$ is reduced by one.
	\end{compactitem}
	As we can find at most $k$ vertices of $S$ and delete at most $|T|$ times vertices of $T$, one of the branches in the lowest level must contain an instance that has an isolating long cycle hitting set.
\end{proof}
%
%

Note that the previous lemma may reduce the size of $T$ without affecting the size of $S$.
Therefore, it might occur that $T$ is smaller than $S$ but as we search for a disjoint solution we may not use the smaller $T$ as solution.


Next, we handle the deletion of ``medium-length cycles'':
\begin{definition}
  For an integer $\ell \in \mathbb{N}$ and a digraph $G$, a cycle $C$ in $G$ is called \emph{medium-length cycle} if the length of $C$ fulfills $\ell < |\ell(C)| < 2\ell^6$.
\end{definition}

Once we have removed all medium-length cycles, we will be left with the following problem:
\begin{center}
  \framebox[1.0\textwidth]{
    \begin{tabularx}{0.98\textwidth}{rX}
      \multicolumn{2}{l}{{\sc\centering Isolating Long Cycle Hitting Set Intersection} \hfill \textit{Parameter:} $k + \ell+|T|$.}\\[1em]
      \textit{Input:}      & A directed multigraph $G$, integers $k,\ell \in\mathbb N$ and a set $T \subseteq V(G)$\\
      \textit{Properties:} 	& $G$ has no medium-length cycles, $\cf(G - T) \leq \ell$\\
      \textit{Task:}	&  Find a set $\mathcal{S}$ intersecting some isolating long cycle hitting set $S$\\
      				& of size at most $k$ with respect to $T$ if such a set exists.
  \end{tabularx}}
\end{center}

\begin{lemma}
\label{lem:hitting_medium_length_cycles}
  There is an algorithm that solves the {\sc Isolating Long Cycle Hitting Set} problem with medium-length cycles in time $2^{\mathcal{O}(k\ell)} \cdot \polyn$ by making $2^{\mathcal{O}(k\log \ell)}$ calls to an algorithm solving {\sc Isolating Long Cycle Hitting Set} in digraphs without medium-length cycles.
\end{lemma}
\begin{proof}
  Start with an instance $(G, k, \ell)$ of {\sc Isolating Long Cycle Hitting Set} problem with medium length cycles.
  Set $S = \emptyset$.
  If $k=|S|$ we can solve the instance $(G - S, k - |S|, \ell)$ by \autoref{thm:checkforsolution}.
  Otherwise, by standard color-coding techniques we can check in time $2^{\mathcal{O}(\ell)} \cdot \polyn$ whether $G - S$ contains a cycle $C$ of length $\ell'$ with $\ell < \ell' < 2\ell^6$.
  If it does $C$ has to be intersected by any solution and we branch on each vertex $v \in V(C)$ whether to include it into the solution or not, and add the selected vertices to $S$ (reducing $k$ appropriately).
  We then proceed for each branch with the instance $(G - S, k - |S|, \ell)$ as above.

  If $(G - S, k - |S|, \ell)$ contains no medium-length cycle, we use the oracle to obtain a solution~$S'$.
  Then $G - (S \cup S')$ is a solution for $G$ and $|S \cup S'| = |S| + |S'| \leq |S| + k - |S| = k$.
  So if one of the branches has a solution we found a solution for our original solution.
  Likewise, if the original instance has a solution, we find one by choosing the correct branches.

	The run time and number oracle calls follow from $|C| < 2\ell^6$ and the fact that we branch at most $k$ steps deep.
\end{proof}

We can now summarize the results of this section in the following theorem:

\begin{theorem}
\label{thm:original_to_isolating_lchs}
  Instances $(G, k, \ell)$ of {\sc Directed Long Cycle Hitting Set} can be solved in time $2^{\mathcal{O}(\ell k^3\log k)} \cdot \left(\functionILCHSI(k, \ell)\right)^k \cdot \polyn$ by at most $2^{\mathcal{O}(\ell k^3\log k)} \cdot \left(\functionILCHSI(k, \ell)\right)^k \cdot n^2\log^{2k + 2}(n)$ calls to an algorithm~$\algorithmILCHSI$ solving the {\sc Isolating Long Cycle Hitting Set Intersection} problem, where~$\functionILCHSI(k, \ell)$ is a size bound on the set produced by $\algorithmILCHSI$.
\end{theorem}
\begin{proof}
  We apply in order: \autoref{thm:safetyofiterativecompression}, \autoref{thm:addingdisjointness}, \autoref{cor:isolating_t} and \autoref{lem:hitting_medium_length_cycles}.
  This produces in time
  \begin{displaymath}
    n \cdot 2^{\mathcal{O}(k)} \cdot 2^{\mathcal{O}((k + |T|) \ell k^2\log k)}\cdot \polyn \cdot 2^{\mathcal{O}(k\ell)} \cdot \polyn = 2^{\mathcal{O}(\ell k^3\log k)} \cdot \polyn
  \end{displaymath}
  at most
  \begin{displaymath}
    n \cdot 2^{\mathcal{O}(k)} \cdot 2^{\mathcal{O}((k + |T|) \ell k^2\log k)}\log^{2|T|} n \cdot 2^{\mathcal{O}(k\ell)} = 2^{\mathcal{O}(\ell k^3\log k)} \cdot n^2 \log^{2k + 2}(n)
  \end{displaymath}
  instances $(G_i, k_i, \ell, T_i)$ of \textsc{Isolating Long Cycle Hitting Set}.
  If $(G, k, \ell)$ has a solution then one of these instances has an isolating long cycle hitting set $S_i$ of size at most $k$.
  Furthermore, if we have an isolating long cycle hitting set $S_i$ in one of these instances, we can complete it to a solution of $(G, k, \ell)$.

  On each of the instances $(G_i, k_i, \ell, T_i)$, we start with $S_i = \emptyset$.
  We then call $\algorithmILCHSI$ on\linebreak $(G_i - S_i, k - |S_i|, \ell, T_i)$.
  If there is an isolating hitting set of size at most $k$ the set $\setILCHSI$ returned by our algorithm intersects at least one of them.
  We branch on $v \in \setILCHSI$ and add $v$ to $S_i$.
  Then we continue as above until $S_i = k_i$.

  If there is an isolating hitting set of size at most $k$, this branching procedure will find it.

  Now it just remains to analyze the running time of this branching procedure.
  For each instance, we branch in each step into $|\setILCHSI| \leq \functionILCHSI(k, \ell)(k, \ell)$ branches and do this at most $k$ levels deep.
  This yields the claimed run time.
\end{proof}

\subsection{Reduction to Important Hitting Separator}
\label{sec:reductiontostrongsubgraph}
In the previous section we reduced the \DLCHS{} problem to the \textsc{Isolating Long Cycle Hitting Set Intersection} problem, a variant where we are already given a solution $T$ of size at most $k + 1$ and search for a solution~$S$ disjoint from $T$ of size at most~$k$.
Additionally, we know that $T$ has at most one vertex in each strong component of $G - S$.
For the remainder of this subsection, assume that there is a solution $S$ of size at most~$k$.

\begin{definition}
  Let $(G, k, \ell, T)$ be an instance of {\sc Isolating Long Cycle Hitting Set Intersection} and let $S$ be an isolating long cycle hitting set in it.
  A vertex $t \in T$ is called \emph{last vertex of~$T$} (with respect to~$S$) if there is a topological ordering of strong components of $G - S$ such that~$t$ appears in the last component that contains some vertex of $T$.
	
  Fix $t \in T$ and let $G_t$ be the graph $G - (T \setminus \lbrace t \rbrace)$.
  Let $\Along = \lbrace (u, v) \in A(G_t) | \dist_{G_t}(v, u) \geq \ell \rbrace$ be the set of arcs of $G_t$ that only lie on long cycles in $G_t$.
  Further, let $\Gcirc = G_t - \Along$ and let~$\Cstar$ be the strong component of $\Gcirc$ containing $t$.
  Finally, set $\Gstar = \Gcirc[\Cstar]$.
\end{definition}

Note that a last vertex of some solution $S$ may not be unique (as there may be different topological orderings).
Yet, no last vertex of a solution may reach another vertex of $T$ in $G - S$.

\begin{lemma}
\label{lem:last_vertex_cant_reach_other_tvertices}
  Let $(G, k, \ell, T)$ be an instance of {\sc Isolating Long Cycle Hitting Set Intersection} and let $S$ be an isolating long cycle hitting set in it.
  Let $t \in T$ be a last vertex of $T$ with respect to $S$.
  Then there is no $t \to T \setminus \lbrace t \rbrace$-path in $G - S$.
\end{lemma}
\begin{proof}
  Suppose, for sake of contradiction, that there exists such a path $P$ and it ends in some~$t'$.
  Then~$P$ is a certificate that $t'$ must be in no earlier strong component than $t$ in every topological ordering of strong components of $G - S$.
  As $t$ was in the last component containing a vertex $t$ in such an ordering, they have to be in the same component.
  This is a contradiction to $S$ being an isolating long cycle hitting set.
\end{proof}

We now make some observations about the strong component in $G - S$ containing some $t \in T$.

\begin{lemma}
\label{lem:ct_contained_in_gstar}
  Let $(G, k, \ell, T)$ be an instance of {\sc Isolating Long Cycle Hitting Set Intersection} and let $S$ be an isolating long cycle hitting set in it.
  For some $t \in T$ let $C_t$ be the strong component of $G - S$ containing $t$.
  Then $C_t$ is contained in $\Gstar$.
\end{lemma}
\begin{proof}
  Consider a closed walk in $G - S$ containing $t$ and some vertex $v$.
  This closed walk cannot go through any vertex of $T\setminus\{t\}$, as this would imply that some vertex of $T\setminus\{t\}$ is in the strong component $C_t$ of $G - S$ containing $t$, contradicting that $S$ is an isolating long cycle hitting set.
  Every arc that is in a closed walk is also in a cycle and as the closed walk is in $G - S$, this cycle has to be of length at most $\ell$.
  Thus, every arc of the closed walk is in $\Gcirc$, which means that the closed walk is fully contained in $\Gstar$.
\end{proof}

\begin{lemma}
\label{lem:cstar_has_only_long_outgoing_arcs}
  Let $(G, k, \ell, T)$ be an instance of \textsc{Isolating Long Cycle Hitting Set Intersection}.
  For $t \in T$, each arc of $\delta^+_G(\Cstar)$ lies either in $A_\textsf{long}$ or is an incoming arc of some $t' \in T \setminus \{t\}$.
\end{lemma}
\begin{proof}
  Suppose, for sake of contradiction, that there is an arc $(u,v) \in \delta^+_G(\Cstar)$ that is not in~$\Along$ and not an incoming arc of some $t' \in T \setminus \{t\}$.
  Thus we have $v \notin T \setminus \{t\}$ and also $u \notin T \setminus \{t\}$ by definition of $\Cstar$, so $(u,v)$ must exist in $G_t$.
  As $(u, v) \notin \Along$, $(u, v)$ must exist in $G_t - \Along = \Gcirc$.
  This guarantees us the existence of a cycle $O$ in $G_t$ of length at most~$\ell$, such that $(u, v) \in O$.
  But then every arc of $O$ lies in $\Gcirc$.
  Therefore, $u$ lies in the same strong component of $G_\circ$ as $v$.
  This, however, is a contradiction to the choice of $u \in C^\star$ and $v \not \in C^\star$.
\end{proof}

To simplify our instance we use the shadow-covering technique established by Chitnis et al.~\cite{ChitnisEtAl2015}.
Let us formally define what the shadow of a solution is:
\begin{definition}[shadow]
\label{def:shadow}
  Let $G$ be a digraph and let $T,S\subseteq V(G)$.
  A vertex $v\in V(G)$ is \emph{in the forward shadow $f_{G,T}(S)$ of $S$ (with respect to $T$)} if $S$ is a $(T,\{v\})$-separator in $G$, and $v$ is \emph{in the reverse shadow $r_{G,T}(S)$ of $S$ (with respect to $T$)} if $S$ is a $(\{v\},T)$-separator in $G$.
  All vertices of~$G$ which are either in the forward shadow or in the reverse shadow of $S$ (with respect to $T$) are said to be \emph{in the shadow of $S$ (with respect to $T$)}.
\end{definition}

Note that $S$ itself is not in the shadow of $S$, by definition of separators.
Intuitively we now have that the endpoints of the unwanted outgoing arcs from above should lie in the shadow of~$S$.
After finding a set which covers the shadow (is a superset of it), we give a method to remove these vertices.
The method requires the notions of $T$-connected and $\mathcal F$-transversals.

\begin{definition}[$T$-connected and $\mathcal F$-transversal]
  Let $G$ be a digraph, let $T\subseteq V(G)$ and let $\mathcal F$ be a set of subgraphs of $G$.
  Say that $\mathcal F$ is \emph{$T$-connected} if for every $F\in\mathcal F$, each vertex of $F$ can reach some and is reachable by some (maybe different) vertex of $T$ by a walk completely contained in $F$.

  For a set $\mathcal F$ of subgraphs of $G$, an \emph{$\mathcal F$-transversal} is a set of vertices that intersects the vertex set of every subgraph in $\mathcal F$.
\end{definition}

Chitnis et al.~\cite{ChitnisEtAl2015} gave a deterministic algorithms for covering the shadow of some $\mathcal F$-transversal.

\begin{proposition}[deterministic covering of the shadow, \cite{ChitnisEtAl2015}]
\label{thm:deterministiccoveringoftheshadow}
  Let $T\subseteq V(G)$.
  One can construct, in time $2^{\mathcal{O}(k^2)}\cdot n^{\mathcal{O}(1)}$, sets $Z_1,\hdots,Z_p$ with $p \leq 2^{\mathcal{O}(k^2)}\log^2n$ such that for any set of subgraphs $\mathcal F$ which is $T$-connected, if there exists an $\mathcal F$-transversal of size at most $k$ then there is an $\mathcal F$-transversal $S$ of size at most $k$ that is disjoint from $Z_i$ and such that $Z_i$ covers the shadow of $S$, for some $i \leq p$.
\end{proposition}

Note that $\mathcal F$ is not an input of the algorithm described by \autoref{thm:deterministiccoveringoftheshadow}.
Hence, issues related to the representation of $\mathcal F$ (which could be exponential in the size of the graph) do not arise.

\begin{corollary}
\label{cor:covering_of_shadows_lchs}
  Let $(G, k, \ell, T)$ be an instance of {\sc Isolating Long Cycle Hitting Set Intersection}.
  One can construct, in time $2^{\mathcal{O}(k^2)}\cdot n^{\mathcal{O}(1)}$, sets $Z_1,\hdots,Z_p$ with $p \leq 2^{\mathcal{O}(k^2)}\log^2n$ such that if there is an isolating long cycle hitting set of size at most $k$, there is isolating long cycle hitting set $S$ an $i \leq t$ of size at most $k$ such that $Z_i \cap S = \emptyset$ and $Z_i$ covers the shadow of $S$ w.r.t. $T$.
  Furthermore, we can assume that $Z_i \cap T = \emptyset$.
\end{corollary}
\begin{proof}
  As the set $\mathcal F$ of forbidden subgraphs we will use all subgraphs of $G$ that are a cycle of length greater than $\ell$.
  Then $\mathcal F$ is clearly $T$-connected and every isolating long cycle hitting set~$S$ is a $\mathcal F$-transversal.
  Applying \autoref{thm:deterministiccoveringoftheshadow} gives us the desired sets, except for the disjointness from~$T$.
  As vertices of $T$ are never in the shadow, we can remove $T$ from all $Z_i$ and get the desired result.
\end{proof}

\noindent
We define a ``torso operation'' to reduce to a digraph $G_\textsf{torso}$ on $V(G) \setminus Z$ preserving connectivity:

\begin{definition}[torso]
  Let $(G, k, \ell, T)$ be an instance of {\sc Isolating Long Cycle Hitting Set Intersection}, let $t \in T$ and $Z \subseteq V(G)$.
  Then the \emph{torso} $\Gtorso$ of $G$ (with respect to $t$ and~$Z$) with a special set $\Ulong \subseteq A(\Gtorso)$ of long arcs is the graph constructed as follows:
	
  Let $V(\Gtorso) = V(G) \setminus Z$ be the vertex set of the torso.
  The arc set $A(\Gtorso)$ contains an arc $(u,v)$ if there is a $u \rightarrow v$-path~$P_{(u,v)}$ in $G$ whose internal vertices are contained in $Z$.
  If there is such a path $P_{(u,v)}$ intersecting $A_\textsf{long}$, add the arc $(u, v)$ to the set $U_\textsf{long}$.
\end{definition}

Note that the path $P_{(u,v)}$ can potentially consist of only a single arc.
The purpose of $U_\textsf{long}$ is to identify arcs in $G_\textsf{torso}$ leaving $C^\star$.

Now, we use the tool of ``critical vertices'' to ensure that $U_\textsf{long}$ not reachable from $t$:
We need the following technical tool, introduced by Chitnis et al.~\cite{ChitnisEtAl2015}.

\begin{definition}[critical vertex]
  Let $G$ be a digraph, $t \in V(G)$ a vertex, $k \in \mathbb{N}$ an integer and $U \subseteq A(G)$ some subset of arcs.
  For some subset of vertices $S \subseteq V(G) \setminus \lbrace t  \rbrace$ an edge $(v, w) \in U$  is called \emph{traversable from~$t$} if there is a $t \to v$-path in $G - S$ and $w \notin S$.
  A vertex $w \in V(G) \setminus \lbrace t  \rbrace$ is called \emph{$k$-critical} (with respect to $t$) if there exists an arc $(v, w) \in U$ and a set $S \subseteq V(G) \setminus \lbrace t  \rbrace$ of size at most $k$ such that there is a $t \to v$-path in $G - S$ but no arc of $U$ is traversable in $G - S$.
\end{definition}

\begin{proposition}[\cite{ChitnisEtAl2015}]
\label{thm:bcl_findsupersetofcriticalvertices}
  Given a digraph $G$, a subset $U$ of its arcs, and some $t \in V(G)$, we can find in time $2^{\mathcal{O}(k)}\cdot n^{\mathcal{O}(1)}$ a set $F_\textsf{critical}$ of $2^{\mathcal{O}(k)}$ vertices that is a superset of all $k$-critical vertices.
\end{proposition}

\begin{lemma}
\label{thm:bcl_nouniversearctraversable}
  Let $(G, k, \ell, T)$ be an instance of {\sc Isolating Long Cycle Hitting Set Intersection} and let $S$ be an isolating long cycle hitting set in it.
  Let $t \in T$ be a last vertex of $T$ with respect to $S$ and $Z \subset V(G) \setminus S$ covering the shadow of $S$ with respect to $T$.
  Then no arc of $\delta^+_{\Gtorso}(\Cstar \setminus Z) \cup U_\textsf{long}$ is traversable from $t$ in $\Gtorso - S$.
\end{lemma}
\begin{proof}
  Let $P$ be a path of $\Gtorso$ starting at $t$, disjoint from $S$, and containing at least one arc~$(v,w) \in \delta^+_{\Gtorso}(\Cstar \setminus Z) \cup \Ulong$.
  By the definition of $\Gtorso$ we can replace the all arcs of $P$ with paths whose internal vertices are in $Z$ (hence disjoint from~$S$ by choice of $Z$) to obtain a $t \to w$-walk~$P'$ in $G - S$.
  As $w \notin Z$ and $Z$ covers the shadow of $S$ with respect to $T$, there is a $w \to T$-path $R$ in $G - S$.
  Thus, $W = P' \circ R$ is a $t \to T$-walk in $G - S$.
  By \autoref{lem:last_vertex_cant_reach_other_tvertices},~$W$ must be disjoint from $T\setminus\{t\}$ i.e. a closed walk.
  Therefore, every arc of $W$ is in a cycle, and as $W$ is disjoint from $S$, these cycles have length at most~$\ell$.
  Hence, every arc of $W$ exists in $\Gcirc$.
	
  If $(v,w) \in \Ulong$, then, by definition of $\Ulong$, $P'$ and therefore $W$ must contain an arc of $\Along$---contradiction to the fact that all arcs of $W$ are in $\Gcirc$.
  Otherwise, we have that $(v,w) \in \delta^+_{\Gtorso}(\Cstar \setminus Z)$, implying $w \notin \Cstar$.
  But then the closed walk $W$ contains $w$ and proves that $w$ is in the same connected component of $\Gcirc$ as $t$, namely $\Cstar$ --- a contradiction.
\end{proof}

\begin{definition}
  Let $(G, k, \ell, T)$ be an instance of {\sc Isolating Long Cycle Hitting Set Intersection}, $t \in T$ and $Z \subseteq V(G)$.
  Let $\Gtorso$ the torso of $G$ with respect to $t$ and $Z$ with $\Ulong$ as set of long arcs.
  We define $\Voutt$ to be the vertex set $\lbrace v \in \Cstar | (v,w) \in \delta^+_{\Gtorso}(\Cstar \setminus Z) \cup \Ulong \rbrace$.
\end{definition}

\begin{lemma}
\label{lem:intersecting_t_to_vout_paths}
  Let $(G, k, \ell, T)$ be an instance of {\sc Isolating Long Cycle Hitting Set Intersection}, let $t \in T$ and $Z \subset V(G) \setminus T$.
  Then in time $2^{\mathcal{O}(k)}\cdot n^{\mathcal{O}(1)}$ we can find a set $\setkcrit$ of $2^{\mathcal{O}(k)}$ vertices, such that for every isolating long cycle hitting set $S$ of size at most $k$
  \begin{compactitem}
    \item for which $t$ is a last vertex of $T$ with respect to $S$,
    \item that is disjoint from $Z$ and for which $Z$ covers the shadow of $S$ with respect to $T$,
    \item and that is disjoint from $\setkcrit$,
  \end{compactitem}
  there is no $t \to \Voutt$-path in $G - S$.
\end{lemma}
\begin{proof}
  Let $\Gtorso$ be the torso of $G$ with respect to $t$ and $Z$ with $\Ulong$ as set of long arcs.
  Use \autoref{thm:bcl_findsupersetofcriticalvertices} on the graph $\Gtorso$ with arc subset $\delta^+_{\Gtorso}(\Cstar \setminus Z) \cup\Ulong$ and vertex $t$ to obtain the set $\setkcrit$.
  The run-time and size bounds follow directly, we just have to argue about the correctness.

  Suppose, for sake of contradiction, that there was a $t \rightarrow v$-path $P$ in $\Gtorso - S$ with $(v,w) \in \delta^+_{\Gtorso}(\Cstar \setminus Z) \cup \Ulong$ for an $S$ as in the statement of the lemma.
  By definition of $v$ and~$w$, we have $v,w \in V(\Gtorso)$.
  Therefore,~$P$ implies a $t \rightarrow v$-path $P'$ in $\Gtorso - S$.
  But no arc of $\delta^+_{\Gtorso}(\Cstar \setminus Z) \cup\Ulong$ is traversable in $G - S$ by \autoref{thm:bcl_nouniversearctraversable}, proving that $w$ is $k$-critical with respect to $t$ and~$S$ in~$\Gtorso$.
  Therefore, $w \in \setkcrit$---which yields a contradiction to the choice of~$S$ being disjoint from~$\setkcrit$.
\end{proof}

\begin{lemma}
\label{lem:force_isolating_lchs_to_t_vout_separators}
  Let $(G, k, \ell, T)$ be an instance of {\sc Isolating Long Cycle Hitting Set Intersection}, let $t \in T$ and $Z \subset V(G) \setminus T$.
  Let $\setkcrit$ as in \autoref{lem:intersecting_t_to_vout_paths}.
  Then in time $2^{\mathcal{O}(k)} \cdot \polyn$ we can find a set $\setdisj$ of at most $(k +1)4^{k+1}$ vertices, such that every isolating long cycle hitting set $S$
  \begin{compactitem}
    \item of size at most $k$,
    \item for which $t$ is a last vertex of $T$ w.r.t. $S$,
    \item that is disjoint from $Z$ and for which $Z$ covers the shadow of $S$ w.r.t. $T$,
    \item and that is disjoint from $\setdisj{} \cup \setkcrit$,
  \end{compactitem}
  contains a $t \to \Voutt$-separator in $G$.
\end{lemma}
\begin{proof}
  Use \autoref{thm:bcl_smallpathwitness} on $t$ and $\Voutt$ to obtain the set $V' \subseteq \Voutt$.
  We return $V' \cup \lbrace t \rbrace$ as $\setdisj$.
  The run-time and size bounds follow directly, we just have to argue about the correctness.
	
  Suppose, for sake of contradiction, that some $S$ as in the lemma does not contain a $t \to \Voutt$-separator.
  In particular, $S' = S \setminus (\Voutt \cup \lbrace t\rbrace)$ is not a $t \to \Voutt$-separator.
  So there is a $t \to \Vout$-path in $G - S'$ and $|S'| \leq |S| \leq k$.
  Thus there is a $t \to V'$-path in $G - S'$ by \autoref{thm:bcl_smallpathwitness}.
  By \autoref{lem:intersecting_t_to_vout_paths}, this path does not exist in $G - S$ (recall that $V' \subseteq \Voutt$).
  Hence the set $S \setminus S' = V' \cup \lbrace t \rbrace$ is not empty.
  By definition of isolating long cycle hitting sets, $t \notin S$ holds.
  Therefore, $S$ intersects $V'$ in contradiction to the choice of $S$ and $V' =\setdisj$.
\end{proof}


All isolating long cycle hitting sets $S$ not already covered by \autoref{lem:force_isolating_lchs_to_t_vout_separators},  are $t \rightarrow \Voutt$-separator.
Also, these $S$ intersects all cycles of length more than~$\ell$ in $G$.
Combining these two properties we introduce the notion of ``hitting separators'':

\begin{definition}
  Let $G$ be a digraph, let $X, Y \subseteq V(G)$ be two vertex sets and let $\ell \in \mathbb{N}$ .
  We call an $X \to Y$-separator $U$ an \emph{hitting $X \to Y$-separator} if $\cf(G - U) \leq \ell$.
  For a set $Z \subseteq V(G)$ an $X \to Y$-separator $U$ is \emph{shadowless} if it is disjoint from $Z$ and every vertex in $V(G) \setminus (U \cup Z)$ can reach a vertex $v \in X \cup Y$ and is reachable from some $u \in X \cup Y$ in $G - U$.
  
  A (shadowless) hitting $X \to Y$-separator $U$ is \emph{important} if there is no (shadowless) hitting $X \to Y$-separator $U'$ with $|U'| \leq |U|$ and $R^-_{G - U'}(Y) \subsetneq R^-_{G - U}(Y)$.
  
  We call two important $X \to Y$-separators $U, U'$ \emph{range equivalent} if $R^-_{G - U}(Y)= R^-_{G - U'}(Y)$.
  This forms an equivalence relation among the important $X \to Y$-separators and we call the equivalence classes \emph{range equivalent classes}.
\end{definition}

Note that in the definition of important hitting separators, instead of maximizing the forward range we minimize the backward range.

As already stated, all isolating long cycle hitting sets $S$ not already covered by \autoref{lem:force_isolating_lchs_to_t_vout_separators} are $t \rightarrow \Voutt$-separators and fulfill $\cf(G - S) \leq \ell$.
Therefore, these $S$ are hitting $t \rightarrow \Voutt$-separators for~$G$.
The subgraph~$\Gstar$ inherits these properties.
Our goal is to replace $S \cap V(\Gstar)$ by an important hitting $t \rightarrow \Voutt$-separator with the help of the following lemma:

\begin{lemma}
\label{thm:bcl_pushinglemmanotcompletlyincstar}
  Let $(G, k, \ell, T)$ be an instance of {\sc Isolating Long Cycle Hitting Set Intersection} and let $S$ be an isolating long cycle hitting set in it.
  Let $t \in T$ be a last vertex of $T$ with respect to $S$, and let $Z \subset V(G) \setminus (S \cup T)$ be a set covering the shadow of $S$ with respect to $T$.
  If a hitting $t \rightarrow \Voutt$-separator $D$ in~$\Gstar$ with $R^-_{\Gstar - D}(\Voutt) \subseteq R^-_{\Gstar - S}(\Voutt)$ exists, then $S' = (S \setminus V(\Gstar)) \cup D$ is an isolating long cycle hitting set.
\end{lemma}
\begin{proof}
  We first show that $\cf(G - S') \leq \ell$.
  Suppose, for sake of contradiction, that $G - S'$ contains a cycle $O$ of length more than~$\ell$.
  As~$O$ is intersected by $S$ but not by $S'$, it is intersected by some $v \in S \cap V(\Gstar)$ but not by $D$.
  By $S$ and $Z$ being disjoint, $v$ exists in $\Gtorso$ and especially in $\Gtorso[\Cstar \setminus Z]$.
  Moreover, $\cf(G - T) \leq \ell$ and therefore $O$ has to be intersected by $T$.
	
  \begin{claim}
    $O$ does not contain $t$.
  \end{claim}
  \begin{claimproof}
    Suppose, for sake of contradiction, that $O$ is intersected by~$t$.
    Then $O$ cannot be completely contained in $\Gstar$, as $D$ is a hitting $t \rightarrow \Voutt$-separator and would intersect it.
    So it has to leave $\Gstar$ by either visiting a $t' \in T \setminus  \lbrace t \rbrace$ or using an arc of $\Along$.
	
    If it visits a $t'$ we have that $t' \notin \Cstar$ (by definition of $\Cstar$ and $t' \notin Z$ (by $Z \cap T = \emptyset$).
    Therefore,~$O$ contains a $t \to t'$-path $P$ in $G - D$ that induces an $t \to t'$-path $P'$ in $\Gtorso - D$.
    Then $P'$ has to use an arc $(x, y) \in \delta_{\Gtorso}(\Cstar \setminus Z)$ which implies $x \in \Voutt$.
    This is a contradiction to $D$ being a $t \rightarrow \Voutt$-separator, as $P[t, x]$ is a $t \rightarrow \Voutt$-path in $\Gstar - D$.
		
    Otherwise, $O$ uses an arc in $\Along$.
    The induced cycle $O'$ in $\Gtorso$ contains at least the vertex $t$ by $T \cap Z = \emptyset$.
    Therefore, $O'$ has at least one vertex and one arc.
    One arc of $O'$ was induced by an path containing an arc in $\Along$.
    This arc is in $\Ulong$ and therefore $O$ contains an vertex $y \in \Voutt$.
    Then the existence of the path $O[t, y]$ is a contradiction to $D$ being a $t \rightarrow \Voutt$-separator.
  \end{claimproof}
	
  As $O$ does not contain $t$, it has to contain some vertex in $t' \in T \setminus \lbrace t \rbrace$.
  So $O[v, t']$ is an $S \cap \Cstar \to T \setminus \lbrace t \rbrace$-path in $G - S'$.
  We get a contradiction by the following claim:
  
  \begin{claim}
  \label{claim:no_cstar_s_path_leaving_cstar}
    There is no $S \cap \Cstar \to T \setminus \lbrace t \rbrace$-path in $G - S'$.
  \end{claim}
  \begin{claimproof}
    Suppose, for sake of contradiction, that such a path $R$ exists.
    Let $R$ be an $s \to t'$-path.
    As $s, t' \notin Z$, path $R$ induces an $s \to t'$-path $R'$ in $\Gtorso - S'$.
    Furthermore, $t' \notin \Cstar$ and therefore $R'$ has to leave $\Cstar$ at least once.
    Let $(x, y)$ be the first arc on $R'$ such that $(x, y) \in \delta_{\Gtorso}(\Cstar \setminus Z) \cup \Ulong$.
    Then $x \in \Voutt$.
    Furthermore, all arcs before $(x, y)$ on $R'$ are induced only by paths that do not contain $\Along$.
    Also $R'[s, x]$ contains no vertex in $T \setminus \lbrace t \rbrace$.
    So $R[s, x]$ lies in $\Gstar - D$.
    This shows that $s$ lies in $R^-_{\Gstar - D}(\Voutt)$ but not in $R^-_{\Gstar - S}(\Voutt)$---a contradiction.
  \end{claimproof}
	
  We just proved that $S'$ is a long cycle hitting set.
  Now we have to show that $S'$ is indeed isolating.
  As $D$ lies in $V(\Gstar) - t \subseteq V(G) - T$ and $S$ is disjoint from $T$, also $D$ is disjoint from~$T$.
  Suppose, for sake of contradiction, that there is a closed walk $W$ containing two different vertices $t_1, t_2 \in T$ in $G - S'$.
  Set $S$ was isolating, and therefore intersects $W$ in some vertex $x \in S \setminus S' \subseteq S \cap \Cstar$.
  Moreover, at least one of $t_1$ and~$t_2$ is different from $t$.
  So~$W$ contains a $S \cap \Cstar \to T \setminus \lbrace t \rbrace$-path --- a contradiction to \autoref{claim:no_cstar_s_path_leaving_cstar}.
\end{proof}

%

We will now state the remaining problem that we face:
\begin{center}
  \framebox[1.0\textwidth]{
    \begin{tabularx}{0.97\textwidth}{rXl}
      \multicolumn{3}{X}{{\sc\centering Important Hitting Separator in Strong Digraphs} \hfill  \textit{Parameter:} $k + \ell$.} \\[1em]
      \textit{Input:}      & \multicolumn{2}{l}{A strong directed multigraph $G$, integers $k,\ell \in\mathbb N$,}\\
      								 & \multicolumn{2}{l}{$t \in V(G)$ and sets $Z, V_\textsf{out} \subseteq V(G)$.}\\
      \textit{Properties:}	& \multicolumn{2}{l}{$\cf(G - t) \leq \ell$, $G$ has no medium-length cycles and}\\
      									&  \multicolumn{2}{l}{every arc of $G$ lies on a cycle of length at most $\ell$.}\\
      \textit{Task:}	&  \multicolumn{2}{l}{If $G$ has a non-trivial important hitting $t \to V_\textsf{out}$-separator,}\\
      							&  \multicolumn{2}{>{\raggedright}p{0.70\textwidth}}{find a vertex set $\setHS$ intersecting either
      								\begin{compactitem}
      									\item all important shadowless hitting $t \to V_\textsf{out}$-separators\linebreak with respect to~$Z$ of size $\leq k$ or
      									\item one important hitting $t \to \Vout$-separator of size $\leq k$\linebreak in every range equivalence class.
									\end{compactitem}}\\[-1em]
  \end{tabularx}}
\end{center}

\begin{theorem}
\label{thm:isolating_lchs_to_important_hitting_separator}
  The {\sc Isolating Long Cycle Hitting Set Intersection} problem for an instance $(G, k, \ell, T)$ can be solved by in time $2^{\mathcal{O}(k^2)} \cdot n^{\mathcal{O}(1)}$ by an algorithm which makes $2^{\mathcal{O}(k^2)}|T|\cdot \log^2(n)$ calls to an algorithm $\algorithmHS$ solving the {\sc Important Hitting Separator in Strong Digraphs} problem and returns a set of size $2^{\mathcal{O}(k^2)}|T|\cdot \log^2(n) \cdot \functionHS(k, \ell)$, where $\functionHS$ is a computable function such that $|\setHS| \leq \functionHS(k,\ell)$ and $\functionHS(k,\ell) \geq 1$.
\end{theorem}
\begin{proof}
  The algorithm reads as follows:\\
  \begin{algorithm}[H]
    \SetKwInOut{Input}{Input}\SetKwInOut{Output}{Output}	
	
    \Input{A digraph $G$, integers $k, \ell$ and a vertex set $T \subseteq V(G)$}
    \Output{A vertex set $\mathcal{S} \subseteq V(G)$}
    Let $\mathcal{S} = \emptyset$\;
		Compute sets $Z_1, \hdots, Z_p$ as in \autoref{cor:covering_of_shadows_lchs}.\;
		\ForEach{$i \in \lbrace 1 \hdots p \rbrace$}{
			\ForEach{$t \in T$} {
				Compute $\setkcrit$ for $G, t, Z_i$ with \autoref{lem:intersecting_t_to_vout_paths}.\;
				Compute $\setdisj$ for $G, t, Z_i$ with \autoref{lem:force_isolating_lchs_to_t_vout_separators}\;
				Add $\setkcrit \cup \setdisj$ to $\mathcal{S}$.\;
				Run $\algorithmHS$ on $(\Gstar, k, \ell, t, Z \cap \Cstar, \Voutt)$ and add the result $\setHS$ to $\mathcal{S}$.\;
			}
		}
  \end{algorithm}
	
  Let $(G, k, \ell, T)$ be an \textsc{Isolating Long Cycle Hitting Set Intersection} instance.
  We first make sure that the calls to $\algorithmHS$ are correct.
  As $\Gstar - t$ is a subgraph of $G - T$, we have $\cf(\Gstar - t) \leq \cf(G - T) \leq \ell$.
  As $G$ contained no medium-length cycles, neither does the subgraph~$\Gstar$.
  Last but not least, $\Gstar$ is a strong component of $\Gcirc$ where every arc lies on a short cycle, so also this property is inherited.
	
  Now we argue the correctness of the algorithm:
  If there is no isolating long cycle hitting separator for $(G, k, \ell, T)$ of size at most $k$, we may return any set.
  Otherwise, fix a long cycle hitting separator with $|S| \leq k$.
  By \autoref{cor:covering_of_shadows_lchs} there is some $Z_i$ that is disjoint from $T \cup S$ and covers the shadow with respect to~$Z$.
  Also, there is some $t$ which is a last vertex of $T$ with respect to $S$.
  For some inner loop we made the correct choices of $Z_i$ and $t$.
  If $S$ intersects $\setkcrit \cup \setdisj$ for these choices of $Z_i$ and $t$, we are done.
  Otherwise, $S$ is a hitting $t \to \Voutt$-separator in $\Gstar$ by \autoref{lem:force_isolating_lchs_to_t_vout_separators}.
  
  \begin{claim}
    The set $S$ is a shadowless hitting $t \to \Voutt$-separator in~$\Gstar$.
  \end{claim}
  \begin{claimproof}
    Let $v \in V(\Gstar) \setminus (S \cup Z_i)$.
    As $Z_i$ covers the shadow of $S$, there is a $v \to t'$-path~$P$ in $G - S$ with $t' \in T$.
    By $v$ and $T$ being disjoint from $Z$, $P$ induces a $v \to t'$-path $P'$ in~$\Gtorso$.
		
    Let $x$ be such that $(x, y)$ us the first arc in $\delta_{\Gtorso}(\Cstar \setminus Z_i) \cup \Ulong$ along $P'$, or $x = t'$ if there is none.
	In the latter case, $t'$ has to be in $\Cstar$ as no arc of $\delta_{\Gtorso}(\Cstar \setminus Z_i)$ was used.
    By choice of $x$, we have that $x \in \Voutt \cup \lbrace t \rbrace$ and $P'[v,x]$ is disjoint from $\delta_{\Gtorso}(\Cstar \setminus Z_i) \cup \Ulong$.
    This implies that $P[v,x]$ uses only arcs of~$\Along$, as otherwise an arc of $P'[v,x]$ had to be in $\Along$.
    But the short cycles of the arcs not in $\Along$ prove that every vertex of $P[v,x]$ is in the same strong component of $\Gcirc$ as $v$, namely $\Cstar$.
    Therefore, $P[v,x]$ is an $v \to  \Voutt \cup \lbrace t \rbrace$-path in $\Gstar = \Gcirc[\Cstar]$.
  \end{claimproof}
	
  Now let $D_\textsf{sl}$ be an important shadowless hitting $t \to \Voutt$-separator with $|D_\textsf{sl}| \leq |S \cap \Cstar|$ and $R^-_{\Gstar - D_\textsf{sl}}(\Voutt) \subseteq R^-_{\Gstar - S}(\Voutt)$.
  Furthermore, let $D$ be an important hitting $t \to \Voutt$-separator with $R^-_{\Gstar - D}(\Voutt) \subseteq R^-_{\Gstar - S}(\Voutt)$ and $|D| \leq |S \cap \Cstar|$.
  Both exist by $S \cap \Cstar = S \cap V(\Gstar)$ being a shadowless hitting $t \to \Voutt$-separator in $\Gstar$.
  Then, by \autoref{thm:bcl_pushinglemmanotcompletlyincstar}, both $S'_\textsf{sl} =  (S\setminus\Cstar) \cup D_\textsf{sl}$ and $S' =  (S\setminus\Cstar) \cup D$ are isolating long cycle hitting sets.
  A solution calculated by $\algorithmHS$ intersects either $D_\textsf{sl}$ or $D$, and therefore intersects either $S'_\textsf{sl}$ or~$S'$.
  Also $|S'_\textsf{sl}| = |S| - |S \cap \Cstar| + |D_\textsf{sl}| \leq k$ and $|S'| = |S| - |S \cap \Cstar| + |D| \leq k$.
  Thus our set $\mathcal{S}$ intersects a isolating long cycle hitting set.
  
  The run time and size bounds follow directly from \autoref{cor:covering_of_shadows_lchs}, \autoref{lem:intersecting_t_to_vout_paths} and \autoref{lem:force_isolating_lchs_to_t_vout_separators}.
\end{proof}

\subsection{Portals and Clusters}
\label{sec:findinghittingseparators}
In the previous section we reduced \textsc{Isolating Long Cycle Hitting Set Intersection} to {\sc Important Hitting Separator in Strong Digraphs}.
We now want to simplify this problem further by consideration of the strong components of $G - t$.
The deletion of $t$ reduces the long cycles in $G$ to paths.
We observe that every long path must be traversing a long distance in some strong component of $G - t$.
By restricting the important hitting separators with the help of some set $\setHS{}$, we can assume that there are not many strong components that need handling.
For these remaining strong components, we then solve the problem individually.

Let us start with the structure of $G$ after the deletion of $t$.
Let $\mathcal C$ be the set of strong components of~$G - t$.
For each $C\in\mathcal C$, we identify certain ``portal'' vertices that can be used to enter/leave the component.

\begin{definition}
  Let $G$ be a graph and let $C \subset V(G)$.
  A vertex $v \in C$ is a \emph{portal vertex} of $C$, if $\Delta_{G}(v) > \Delta_{G[C]}(v)$, where~$\Delta_H(v)$ is the number of incident arcs (both in-coming and out-going) of~$v$ in a graph~$H$.
  We denote by $X_C$ the set of all portal vertices of $C$.
\end{definition}

\begin{lemma}
\label{thm:bcl_shortcyclethroughvertex}
	Let $G$ be a digraph where every arc lies on a cycle of length at most $\ell$.
	Then for any $C \in \mathcal{C}$ and any $v \in X_C$ there is a cycle of length at most $\ell$ in $G$ going through~$v$ and $t$.
\end{lemma}
\begin{proof}
  As $\Delta_{G}(v) > \Delta_{G[C]}(v)$ there is an arc $a \in A(G)$ incident to $v$ with its other endpoint $w$ not contained in $C$.
  We know that $a$ lies on a cycle of length at most $\ell$ in $G$.
  As $a \not \in G[C]$ this cycle exists in $G$ but not in $G - t$; thus, the cycle goes through $t$.
\end{proof}

For every $C \in \mathcal{C}$ and $v\in X_C$, fix an arbitrary cycle as in \autoref{thm:bcl_shortcyclethroughvertex}, and let $O_v$ be the vertex set of that cycle.

\begin{lemma}
\label{thm:bcl_atleastoneshortdist}
  For any $v_1,v_2\in X_C$, either $\mathsf{dist}_{G[C]}(v_1,v_2)\leq 2\ell^2$ or $\mathsf{dist}_{G[C]}(v_1,v_2)\geq 2\ell^6 - 2\ell$.
\end{lemma}
\begin{proof}
  Suppose, for sake of contradiction, that $P_1$ is a $v_1\rightarrow v_2$-path of $G[C]$ with $2\ell^2 + 1\leq |P_1|\leq 2\ell^6 - 2\ell - 1$.
  There is a $t\rightarrow v_1$-path using only the vertices of $O_{v_1}$ and hence has length at most~$\ell$.
  Similarly, there is a $v_2\rightarrow t$-path of length at most $\ell$.
  Concatenating these two paths shows that $\mathsf{dist}(v_2,v_1)\leq 2\ell$; let $P_2$ be a $v_2\rightarrow v_1$-path of length at most $2\ell$ in $G$.
  
  Consider the digraph $G'$ induced by the vertices of the $v_1\rightarrow v_2$-path $P_1$ and the $v_2\rightarrow v_1$-path~$P_2$.
  This graph has at most $|P_1| + |P_2| \leq |P_1| + 2\ell < 2\ell^6$ vertices.
  As $G$ contains no medium-length cycles (i.e no cycles with length in $(\ell, 2\ell^6]$\enspace), we get $\mathsf{cf}(G')\leq \ell$.
  Applying \autoref{thm:bcl_distboundedcircum} on $P_1$ and $P_2$ in~$G'$, we get $|P_1|\leq (\mathsf{cf}(G')-1)|P_2|\leq (\ell-1)\cdot 2\ell < 2\ell^2 + 1$, a contradiction.
\end{proof}

Let $C\in\mathcal C$.
We partition $X_C$ into clusters the following way.
Let $\ell_{\max} := 2\ell^2$.
For every $v\in X_C$, let $X_v$ be the subset of $X_C$ that is at distance at most $\ell_{\max}$ from $v$ in $G[C]$ (note that $v\in X_v$).

\begin{lemma}
\label{thm:bcl_disjointorequal}
  For every $C\in\mathcal C$ and $v_1,v_2\in X_C$, the sets $X_{v_1}$ and $X_{v_2}$ are either disjoint or equal.
\end{lemma}
\begin{proof}
  Suppose that $x\in X_{v_1}\cap X_{v_2}$ and, without loss of generality, $y\in X_{v_1}\setminus X_{v_2}$.
  Now
  \begin{eqnarray*}
    \mathsf{dist}_{G[C]}(v_2,y) & \leq & \mathsf{dist}_{G[C]}(v_2,x) + \mathsf{dist}_{G[C]}(x,v_1) + \mathsf{dist}_{G[C]}(v_1,y)\quad(\text{triangle inequality})\\
                                & \leq & \ell_{\max} + (\ell - 1)\mathsf{dist}_{G[C]}(v_1,x) + \ell_{\max}\quad(\text{\autoref{thm:bcl_distboundedcircum}})\\
                                & \leq & \ell_{\max} + (\ell - 1)\cdot\ell_{\max} + \ell_{\max}\\
                                &    = & (\ell + 1)\ell_{\max}\\
                                & \leq & 2\ell^6 - 2\ell - 1,
  \end{eqnarray*}
  hence by \autoref{thm:bcl_atleastoneshortdist}, we actually have $\mathsf{dist}_{G[C]}(v_2,y) \leq 2\ell^2 = \ell_{\max}$, implying $y\in X_{v_2}$.
\end{proof}

Therefore, the sets $X_v$ for $v\in X_C$ define a partition of $X_C$; we call the classes of these partitions the \emph{clusters} of $X_C$.
An example for portals and clusters can be found in \autoref{fig:portalsandclusters}.
\begin{figure}[hb!]
  \center
  \includegraphics[scale=1.5]{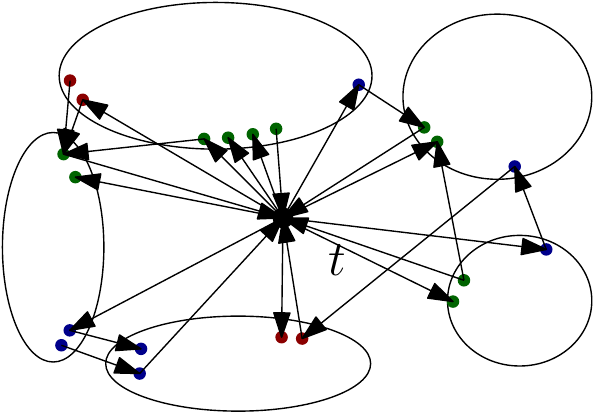}
  \caption{An example for the structure of $G- t$. The large circles form the strong components $C \in \mathcal{C}$. The colored dots represent the portals with their color corresponding to their cluster.\label{fig:portalsandclusters}}
\end{figure}

The huge distance between the clusters allows for the following structural insight:

\begin{lemma}
\label{thm:bcl_longcyclethroughsolution}
  Let $R$ be a cycle of length more than $\ell$ in $G$.
  Then $R$ contains a path between two different clusters of some strong component $C\in\mathcal C$.
\end{lemma}
\begin{proof}  
  As $\cf(G - t) \leq \ell$ we have that $t \in R$.
  Starting from~$t$, let $x_0,\hdots,x_p$ be the vertices of $R$ that are in $\bigcup_{C\in\mathcal C}X_C$.
  By definition, the vertex after~$t$ is in $X_C$ for some $C\in\mathcal C$ and the vertex before $t$ is in $X_{C'}$ for some $C'\in\mathcal C$.
  Thus, $R[x_0,x_p]$ contains every vertex of $R$ except $t$, yielding $|R[x_0,x_p]| = |R| - 2$.
  If an arc $(u,v)$ of $R[x_0,x_p]$ has $u$ and $v$ in different strong components $C_1\in\mathcal C$ and $C_2\in\mathcal C$ respectively, then $u \in X_{C_1}$ and $v\in X_{C_2}$, hence both appear in the sequence $x_0,\hdots,x_p$.
  Therefore, for $i = 0,\hdots,p-1$ the subpath $R[x_i,x_{i+1}]$ is either fully contained in a single component $C\in\mathcal C$ or consists of only one arc.
  If $x_i$ and $x_{i+1}$ are in the same strong component $C\in\mathcal C$ and they are in two different clusters of $C$, then we are done.
  Otherwise, if $x_i$ and $x_{i+1}$ are in the same cluster, then $\mathsf{dist}_{G[C]}(x_i,x_{i+1})\leq \ell_{\max}$ by the definition of the clusters.
  Thus \autoref{thm:bcl_distboundedcircumsquared} implies $|R[x_i,x_{i+1}]|\leq (\mathsf{cf}(G[C])-1)^2\cdot\ell_{\max} < (\ell-1)^2\cdot\ell_{\max}$.
  Therefore, if $p < \ell^2$, we have $|R| = |R[x_0,x_p]|+2\leq p\cdot(\ell-1)^2\cdot\ell_{\max} + 2 < 2\ell^6$, contradicting that $G$ has no medium length cycles.
  Otherwise, consider the vertex $x_{\ell^2}$; we have $\ell^2\leq |R[x_0,x_{\ell^2}]|\leq \ell^2(\ell-1)^2\ell_{\max}$.
  By \autoref{thm:bcl_shortcyclethroughvertex}, there is an $x_i\rightarrow t$-path of length at most $\ell - 1$.
  As $x_0$ is an out-neighbor of~$t$, this means that there is an $x_{\ell^2}\rightarrow x_0$-path $Q$ of length at most~$\ell$ in $G$.
  Let $G'$ be the digraph induced by $R[x_0,x_{\ell^2}]$ and $Q$.
  As $G'$ has at most $|R[x_0,x_{\ell^2}]| + |Q| \leq \ell^2(\ell -1)^2\lmax + \ell < 2\ell^6$ vertices and $G$ contains no medium-length cycles (i.e., no cycles with length in $(\ell, 2\ell^6]$), we have that $\cf(G') \leq \ell$.
  Applying \autoref{thm:bcl_distboundedcircum} on $R[x_0,x_{\ell^2}]$ and $Q$ in $G'$ we get $|R[x_0,x_{\ell^2}]| \leq (\ell - 1)|Q| < \ell^2 \leq |R[x_0,x_{\ell^2}]|$---a contradiction.
\end{proof}

We now focus again on finding the important hitting separators in $G$.
For this we fix an arbitrary important hitting separator $S$.
This separator is not known to the algorithm but helps our analysis.
Our main observation about clusters of $C\in\mathcal C$ is that $S$ has to separate them from each other.

\begin{lemma}
\label{thm:bcl_clusterintersection}
  Let $x_1\in L_1, x_2\in L_2$ for distinct clusters $L_1,L_2$ of some strong component $C \in\mathcal C$.
  For each hitting $t \to \Vout$-separator $S$ disjoint from $O_{x_1} \cup O_{x_2}$, there is no $x_1\rightarrow x_2$-path in $G - S$.
\end{lemma}
\begin{proof}
  Let $S$ be disjoint from $O_{x_1} \cup O_{x_2}$ and suppose, for sake of contradiction, that $P_1$ is an $x_1\rightarrow x_2$-path in $G[C] - S$.
  As~$x_1$ and~$x_2$ are in distinct clusters, we have $|P_1| > \ell_{\max}$.
  There is a $t\rightarrow x_1$-path of $G$ using only the vertices of $O_{v_1}$ and hence has length at most $\ell$.
  Similarly, there is an $x_2\rightarrow t$-path in $G$ using only vertices of $O_{v_2}$ and having length at most $\ell$.
  The concatenation of these two paths gives an $x_2\rightarrow x_1$-walk using only the vertices $O_{v_1}\cup O_{v_2}$ and having length at most $2\ell$.
  This walk contains an $x_2\rightarrow x_1$-path $P_2$ of length at most $2\ell$.
  
  By the assumptions of the lemma, $P_1$ and $P_2$ are disjoint from $S$.
  Applying \autoref{thm:bcl_distboundedcircum} on the $x_1\rightarrow x_2$-path $P_1$ and the $x_2\rightarrow x_1$-path $P_2$ in $G - S$, we get $|P_1|\leq (\mathsf{cf}(G - S) - 1)|P_2|\leq (\ell - 1)\cdot 2\ell < \ell_{\max}$---a contradiction.
\end{proof}

Our next goal is to use \autoref{thm:bcl_clusterintersection} to argue that there cannot be too many clusters in a component $C\in\mathcal C$ and only a few components can contain more than one cluster.
This may in general not be the case, but if there are many clusters we can identify vertices of $S$.

\begin{lemma}
\label{thm:bcl_atmostoneclusterpath}
  Let $C\in\mathcal C$ be a component with distinct clusters $L_1,L_2$ and $Z \subseteq V(G)$.
  Let $x_1\in L_1$, $x_2\in L_2$ and $x_3 \in C\setminus Z$.
  Any hitting $t \to \Vout$-separator $S$ that is disjoint from $O_{x_1} \cup O_{x_2}$ and shadowless with respect to $Z$ cannot have both an $x_1\rightarrow x_3$-path $P_1$ and an $x_2\rightarrow x_3$-path $P_2$ in $G[C] - S$.
\end{lemma}
\begin{proof}
  Let $S$ be as in the statement and suppose, for sake of contradiction, that both $P_1$ and~$P_2$ exist.
  Let~$R_i$ be a $t\rightarrow x_i$-path in $O_{x_i}$ for $i = 1,2$.
  The concatenation of $R_1$ and $P_1$ shows that there is a $t\rightarrow x_3$-path in $G - S$.
  By $x_3 \in V(G) \setminus Z$ and $S$ being shadowless with respect to $Z$, a vertex $v \in V_\textsf{out} \cup \{t\}$ is reachable from $x_3$ in $G - S$.
  If $v\neq t$, then this means that $V_\textsf{out}$ is reachable from $t$ in $G - S$, contradicting that~$S$ is a $t \to V_\textsf{out}$-separator.
  Therefore, there is an $x_3\rightarrow t$-path~$Q$ disjoint from $S$ in $G - S$.
  Let $(a,b)$ be the last arc of~$Q$ that is not entirely in~$C$ (as $t\notin C$, there is such an arc).
  As $Q[x_3, a]$ is a path disjoint from $t$ which starts and ends in~$C$, the path~$Q$ is entirely contained in the strong component $C$ of $G - t$.
  Thus $a$ is in~$X_C$, and hence in some cluster $L$.
  Now for $i = 1,2$ we have that $P_i\circ Q[x_3,a]$ is a walk from~$L_i$ to $L$, fully contained in $G[C] - S$.
  However, $L$ is different from at least one of $L_1$ and $L_2$.
  Without loss of generality, let $L \neq L_1$.
  Then $P_1\circ Q[x_3,a]$ contains an $x_1 \to a$-path of length at least~$\lmax$ by definition of clusters.
  Likewise, $P_1 \circ Q[x_3, t]$ contains an $x_1 \to t$-path $R^\star$ of length at least $\lmax$ (as $Q[a, t]$ is outside of~$C$).
  Consider again the $t \to x_1$-path $R_1$ inside $O_{x_1}$.
  As it lies inside $O_{x_1}$, it is disjoint from~$S$ and has length at most $\ell$.
  If we now compare the length of $R^\star$ and $R_1$ with help of \autoref{thm:bcl_distboundedcircum}, we get $|R^\star| \leq (\cf(G - S) - 1)|R_1| \leq \ell^2 < \lmax \leq |R^\star|$---a contradiction.
\end{proof}

The following lemma suggests a branching step when a vertex is reachable from many clusters on (mostly) disjoint paths.
\begin{lemma}
\label{thm:bcl_vertexreachedbymanyclusters}
  Let $(G, k, \ell, t, Z, \Vout)$ be an {\sc Important Hitting Separator in Strong Digraphs} instance and $\mathcal{C}$ the strong components of $G - t$.
  Let $x_1,\hdots,x_{k+2}$ be vertices in distinct clusters of a component $C\in\mathcal C$ and $v\in C \setminus Z$ be a vertex.
  Furthermore, let $P_1,\hdots,P_{k+2}$ be paths in $G[C]$ such that $P_i$ is an $x_i \rightarrow v$-path and these paths share vertices only in $Z\cup\{v\}$.
  Then every hitting $t \to \Vout$-separator $S$ of size at most $k$ that is shadowless with respect to $Z$ intersects $v \cup \bigcup_{i = 1}^{k + 2} O_{x_i}$.
\end{lemma}
\begin{proof}
  As $|S|\leq k$ and $S$ is disjoint from $Z$, at least two of the $P_i$'s have their internal vertices disjoint from~$S$.
  Assume, without loss of generality, that $S$ contains no internal vertex of $P_1$ and~$P_2$.
  If $S$ is disjoint from $O_{x_1}\cup O_{x_2}\cup\{v\}$, then \autoref{thm:bcl_atmostoneclusterpath} gives a contradiction.
\end{proof}

\begin{lemma}
\label{lem:check_if_vertex_reached_by_many_clusters_is_applicable}
  For an {\sc Important Hitting Separator in Strong Digraphs} instance\linebreak $(G, k, \ell, t, Z, \Vout)$ it can be tested in polynomial time whether \autoref{thm:bcl_vertexreachedbymanyclusters} is applicable.
\end{lemma}
\begin{proof}
  For every $C \in \mathcal{C}$ and $v \in C \setminus Z$, we solve the following vertex-capacitated maximum flow problem: introduce a new source adjacent to each cluster of $C$, set $v$ to be the only sink, vertices in $Z\cup\{v\}$ have infinite capacity, and every other vertex of $C$ has unit capacity.
  An integral flow of value at least $k+2$ corresponds directly to the vertices in the \autoref{thm:bcl_vertexreachedbymanyclusters}.
  As an maximum integral flow can be found in polynomial time and we have at most $|V(G)|$ choices for $v$ and $C$ (choosing $v$ fixes $C$) we can check for this in polynomial time.
\end{proof}

If \autoref{thm:bcl_vertexreachedbymanyclusters} is not applicable and a strong component $C \in \mathcal{C}$ with many clusters exists, we can find a simple set intersecting every shadowless hitting $t \to \Vout$-separator $S$ of size at most~$k$:

\begin{lemma}
\label{thm:bcl_intersectsomecluster}
  Let $(G, k, \ell, t, Z, \Vout)$ be an {\sc Important Hitting Separator in Strong Digraphs} instance and $\mathcal{C}$ the strong components of $G - t$.
  Let $x_0,\hdots,x_{k(k+1)+1}$ be vertices from different clusters of some $C\in\mathcal C$.
  If \autoref{thm:bcl_vertexreachedbymanyclusters} is not applicable then every hitting $t \to \Vout$-separator $S$ of size at most $k$ that is shadowless with respect to $Z$ intersects $\bigcup_{i = 1}^{k(k+1)+1} O_{x_i}$.
\end{lemma}
\begin{proof}
  Suppose that $S$ is disjoint from every $O_{x_i}$.
  For $i = 1,\hdots,k(k+1)+1$ let us fix an $x_i\rightarrow x_0$-path~$P_i$ in $C$.
  By \autoref{thm:bcl_clusterintersection}, $S$ intersects path $P_i$ for every $i = 1,\hdots,k(k+1)+1$.
  For $i = 1,\hdots,k(k+1)+1$, let~$y_i$ be the first vertex of $S$ on $P_i$.
  There has to be a vertex of~$S$ that appears as $y_i$ for at least $k+2$ values of~$i$; assume, without loss of generality, that $y_1 =\hdots= y_{k+2}=:y$.
  If for some $1\leq j_1 < j_2 \leq k+2$, paths $P_{j_1}[x_{j_1},y]$ and $P_{j_2}[x_{j_2},y]$ share a vertex different from $y$, then by \autoref{thm:bcl_atmostoneclusterpath} this vertex has to be in $Z$.
  Therefore, the paths $P_1[x_1,y],\hdots,P_{k+2}[x_{k+2,y}]$ share vertices only in $Z\cup\{y\}$, and hence \autoref{thm:bcl_vertexreachedbymanyclusters} would be applicable, a contradiction.
\end{proof}

%

Next, we find a simple intersection set if many components have more than two clusters.

\begin{lemma}
\label{thm:bcl_fewcomponentswithmanyclusters}
  Let $(G, k, \ell, t, Z, \Vout)$ be an {\sc Important Hitting Separator in Strong Digraphs} instance and $\mathcal{C}$ the strong components of $G - t$.
  If there exist strong components $C_1,\hdots,C_{k+1}\in\mathcal C$, each containing at least two clusters.
  Then for arbitrary vertices $x_i,y_i$ of different clusters of $C_i$, every hitting $t \to \Vout$-separator~$S$ of size at most $k$ intersects $\bigcup_{i=1}^{k+1}(O_{x_i}\cup O_{y_i})$.
\end{lemma}
\begin{proof}
  Suppose that $S$ is disjoint from $\bigcup_{i=1}^{k+1}(O_{x_i}\cup O_{y_i})$.
  Then some $C_i$ is disjoint from~$S$, implying that there is an $x_i\rightarrow y_i$-path in $G[C_i] - S$, contradicting \autoref{thm:bcl_clusterintersection}.
\end{proof}

\begin{corollary}
\label{cor:unbalanced_clusters}
  Let $(G, k, \ell, t, Z, \Vout)$ be an {\sc Important Hitting Separator in Strong Digraphs} instance and $\mathcal{C}$ the strong components of $G - t$.
  If
  \begin{compactitem}
    \item either there is a $C \in \mathcal{C}$ with more than $k(k+1)+1$ clusters,
    \item or there are more than $k$ components in $\mathcal{C}$ with at least two clusters,
  \end{compactitem}
  then there is a set $\setManyClusters \subseteq V(G)$ of size at most $(k^2+ k +1)(\ell - 1)$ that intersects every hitting $t \to \Vout$-separator $S$ of size at most $k$ which is shadowless with respect to $Z$.
  Moreover, the set~$\setManyClusters$ can be found in polynomial time.
\end{corollary}
\begin{proof}
  If there is a $C \in \mathcal{C}$ with more than $k(k+1)+1$ clusters either \autoref{thm:bcl_vertexreachedbymanyclusters} or \autoref{thm:bcl_intersectsomecluster} is applicable.
  If there are more than $k$ components in $\mathcal{C}$ with at least two clusters \autoref{thm:bcl_fewcomponentswithmanyclusters} is applicable.
  We can check whether these conditions are accurate (by possibly using \autoref{lem:check_if_vertex_reached_by_many_clusters_is_applicable}) in polynomial time.
  We can also compute the sets $\mathcal{S}$ they produce in polynomial time and use one of them as set $\setManyClusters$.
  By taking the maximum over their size bounds and using that $|O_x - t| \leq (\ell - 1)$, we obtain the promised size bound.
\end{proof}

Now we need to handle the remaining case.
Namely, that there are only $k$ components in $\mathcal{C}$ with more than two clusters and these have at most $k(k+1) + 1$ clusters.

\begin{lemma}
\label{lem:force_cluster_path_inside_single_component}
  Let $(G, k, \ell, t, Z, \Vout)$ be an {\sc Important Hitting Separator in Strong Digraphs} instance and $\mathcal{C}$ the strong components of $G - t$.
  If \autoref{cor:unbalanced_clusters} is not applicable, then there is a set $\mathcal{S}$ of size at most $2^{\mathcal{O}(k + \log \ell)}$ such that every hitting $t \to \Vout$-separator $S$ of size at most $k$ either intersects $\mathcal{S}$ or for every $C \in \mathcal{C}$ intersects all paths between different clusters of $C$ in $G[C]$.
  Furthermore, the set $\mathcal{S}$ can be found in time $2^{\mathcal{O}(k)} \cdot \polyn$.
\end{lemma}
\begin{proof}
  Construct $\mathcal{S}$ as follows:\\
  \begin{algorithm}[H]
    \SetKwInOut{Input}{Input}\SetKwInOut{Output}{Output}	
	
    \Input{Digraph $G$, integers $k, \ell$ and a vertex $t \in V(G)$}
    \Output{A vertex set $\mathcal{S} \subseteq V(G)$}
	Let $\mathcal{S} = \emptyset$\;
	\ForEach{$C \in \mathcal{C}$ with at least two clusters} {
      Let $L_1,\hdots,L_t$ the clusters of $C$\;
      Apply \autoref{thm:bcl_smallpathwitnesstsets} to $L_1,\hdots,L_t$ to obtain $L_1',\hdots,L_t'$.\;
      \ForEach{$x \in L_1' \cup \hdots \cup L_t'$}{
        Add $O_x$ to $\mathcal{S}$.\;
      }
    }		
  \end{algorithm}
	
  The size bound follows from $|O_x| \leq \ell$, $t \leq k^2 + k +1$ (as \autoref{cor:unbalanced_clusters} is not applicable) and \autoref{thm:bcl_smallpathwitnesstsets} which yields $|\mathcal{S}| \leq \ell \cdot 2(t-1)(k+1)4^{k + 1} = 2^{\mathcal{O}(k + \log \ell)}$.
	
  For correctness, let $S$ be a hitting $t \to \Vout$-separator of size at most $k$ that is disjoint from $\mathcal{S}$.
  Suppose, for sake of contradiction, that there is a path $P$ in $G[C] - S$ between different clusters of $C$ for some $C \in \mathcal{C}$.
  Without loss of generality, let $P$ be a $L_1 \to L_2$-path.
  By \autoref{thm:bcl_smallpathwitnesstsets} there is also an $x \to y$-path $Q$ in $G[C] - S$ with $x \in L_1', y \in L_2'$.
  But $S$ is disjoint from $O_x \cup O_y \subseteq \mathcal{S}$ in contradiction to \autoref{thm:bcl_clusterintersection}.
\end{proof}

\begin{lemma}
\label{lem:force_v_to_vout_path_inside_single_component}
  Let $(G, k, \ell, t, Z, \Vout)$ be an {\sc Important Hitting Separator in Strong Digraphs} instance, and let~$\mathcal{C}$ be the set of strong components of $G - t$.
  There is a set $\mathcal{S}$ of size at most $2^{\mathcal{O}(k + \log \ell)}$ such that for every hitting $t \to \Vout$-separator $S$ of size at most $k$  that is disjoint from $\mathcal{S}$,
  \begin{compactitem}
    \item there is no $L \to \Vout$-path in $G - S$ for any cluster $L$ of some $C \in \mathcal{C}$,
    \item and any $v \rightarrow \Vout$-path $P$ in $G - (S - v)$ for some $v \in C$ is entirely contained in $C$.
  \end{compactitem}
  Furthermore, the set $\mathcal{S}$ can be found in time $2^{\mathcal{O}(k)} \cdot \polyn$.
\end{lemma}
\begin{proof}
  Let $X = \bigcup_{C \in \mathcal{C}}X_C$ be the set of all portal vertices.
  Use \autoref{thm:bcl_smallpathwitnessset} on $X$ and $\Vout$ to obtain a set $X'$ of size $2^{\mathcal{O}(k)}$ such that if there is an $X \to \Vout$-path in $G - S$, there is also a $X' \to \Vout$-path in $G - S$.
  Let $\mathcal{S} = \bigcup_{x \in X'} O_x$.
  This set has size $|\mathcal{S}| \leq |X'| \cdot \ell = 2^{\mathcal{O}(k + \log \ell)}$.
	
  Now let $S$ be a hitting $t \to \Vout$-separator of size at most $k$ disjoint from $\mathcal{S}$.
  To show the first statement, suppose, for sake of contradiction, that there is an $x \to \Vout$-path in $G - S$ with $x \in L$ for some cluster $L$ of some $C \in \mathcal{C}$.
  We have that $x \in X_C \subseteq X$.
  Therefore, by \autoref{thm:bcl_smallpathwitnessset}, there is an $x' \to \Vout$-path $P$ with $x' \in X'$.
  So $\Vout$ is reachable from $x'$ in $G - S$, but also~$x'$ is reachable from~$t$ in $G - S$ as $O_{x'}$ is disjoint from~$S$.
  Thus, $\Vout$ is reachable from $t$---a contradiction to $S$ being a $t \to \Vout$-separator.
	
  Now for the second statement:
  Suppose, for sake of contradiction, that there is a $v \rightarrow V_\textsf{out}$-path~$P$ in $G - (S - v)$ for some $v\in C \in \mathcal{C}$ that is not entirely contained in $C$.
  The path~$P$ cannot contain~$t$ as then $G - S$ would have a $t \rightarrow V_\textsf{out}$ path in contradiction to $S$ being a $t \to \Vout$-separator.
  Thus $P$ visits some component $C' \in \mathcal{C}$ different from $C$.
  It enters this component through some $x \in X_{C'}$.
  Let $P$ end in $z$, then $P[x,z]$ is an $X \to \Vout$-path in $G - S$.
  By choice of $X'$ there is also an $x' \to \Vout$-path $Q$ in $G - S$ with $x' \in X'$.
  As~$S$ is disjoint from $O_{x'} \subseteq \mathcal{S}$, there is a $t \to x'$-path~$R$ disjoint from $S$ inside $O_{x'}$.
  Overall, $R \circ Q$ is an $t \to \Vout$-path disjoint from~$S$---a contradiction to~$S$ being a $t \to \Vout$-separator.
\end{proof}
%

If we assume disjointness of the sets constructed above, we now have that all important paths lie in a single component $C \in \mathcal{C}$.
To emphasize the structure in every strong component $C \in \mathcal{C}$ we introduce the concept of ``cluster separators''.
\begin{definition}
  Let $G$ be a digraph and let $X_1, \hdots, X_t, Y \subseteq V(G)$ be pairwise disjoint vertex sets.
  We call a vertex set $U \subseteq V \setminus (X_1 \cup \hdots \cup X_t \cup Y)$ a \emph{cluster separator} if $G - U$ contains (i) no path from $X_i$ to $X_j$ for $i \neq j$ and (ii) no path from $X_i$ to $Y$ for $i = 1, \hdots, t$. 

  A cluster separator $U$ is \emph{important} if there is no cluster separator ~$U'$ with $|U'| \leq |U|$ and $R^-_{G - U'}(Y) \subsetneq R^-_{G - U}(Y)$.
\end{definition}

With this notion of important cluster separators we can describe the structure of $S$ in every strong component $C \in \mathcal{C}$.

\begin{lemma}
\label{lem:cluster_separators_in_every_component_are_hitting_separators}
  Let $(G, k, \ell, t, Z, \Vout)$ be an \textsc{Important Hitting Separator in Strong Digraphs} instance, and let $\mathcal{C}$ be the set of strong components of $G - t$.
  For each $C \in \mathcal{C}$, let $L^C_1, \hdots, L^C_t$ be the clusters of $C$.
  If a set $S \subseteq V(G)$ is a cluster separator for each component of $G$, i.e., $S \cap C$ is an important cluster separator in $G[C]$ for $L_1^C, \hdots, L_t^C, (V_\textsf{out} \cap C)$ for every $C \in \mathcal{C}$, then $S$ is a hitting $t \to \Vout$-separator.
\end{lemma}
\begin{proof}
  Suppose, for sake of contradiction, that $S$ is not a hitting $t \to \Vout$-separator.
  Then $G - S$ contains either a $t \to \Vout$-path or a cycle of length more than $\ell$.
  If $G - S$ contains a $t \to z$-path~$P$ with $z \in \Vout$, there must be a $C \in \mathcal{C}$ such that $z \in C$.
  As $P$ can enter $C$ only through a portal vertex, there has to be a $y \in L^C_i$ such that $P[y, z]$ is completely contained in $G[C] - (S \cap C)$.
  This is a contradiction to $S \cap C$ being a cluster separator for $L_1^C, \hdots, L_t^C, (V_\textsf{out} \cap C)$ in $G[C]$.
  So $G - S$ must contain a cycle $O$ of length more than $\ell$.
  By \autoref{thm:bcl_longcyclethroughsolution} there must be a subpath~$Q$ connecting two different clusters in some component $C \in \mathcal{C}$.
  This is again a contradiction to $S \cap C$ being a cluster separator for $L_1^C, \hdots, L_t^C, (V_\textsf{out} \cap C)$ in $G[C]$.
\end{proof}

\begin{lemma}
\label{lem:s_is_important_cluster_separator_in_every_strong_component}
  Let $(G, k, \ell, t, Z, \Vout)$ be an {\sc Important Hitting Separator in Strong Digraphs} instance and $\mathcal{C}$ the strong components of $G - t$.
  For $C \in \mathcal{C}$ let $L^C_1, \hdots, L^C_t$ be the clusters of $C$.
  If \autoref{cor:unbalanced_clusters} is not applicable, then there is a set $\setSeparateClusters$ of size at most $2^{\mathcal{O}(k + \log \ell)}$ such that every important hitting $t \to \Vout$-separator $S$ of size at most $k$ either intersects $\mathcal{S}$ or $S \cap C$ is an important cluster separator in $G[C]$ for $L_1^C, \hdots, L_t^C, (V_\textsf{out} \cap C)$ for every $C \in \mathcal{C}$.
  Furthermore, the set $\mathcal{S}$ can be found in time $2^{\mathcal{O}(k)} \cdot \polyn$.
\end{lemma}
\begin{proof}
  We set $\setSeparateClusters$ to be the union of the sets computed in \autoref{lem:force_cluster_path_inside_single_component} and \autoref{lem:force_v_to_vout_path_inside_single_component}.
  Now let~$S$ be an important hitting $t \to \Vout$-separator of size at most $k$ that is disjoint from $\setSeparateClusters$.
	
  First we have to show that $S \cap C$ is an cluster separator at all.
  Assume for contradiction that $S \cap C$ is not a cluster separator.
  Then in $G[C] -S$ there is either a path between two different clusters or a path from one cluster $L_i$ to $\Vout \cap C$.
  \autoref{lem:force_cluster_path_inside_single_component} rules out the first case, while \autoref{lem:force_v_to_vout_path_inside_single_component} rules out the latter --- a contradiction.
	
  Now we can show the importance of $S \cap C$ as an cluster separator in $G[C]$.
  Suppose, for sake of contradiction, that $S \cap C$ is not an important cluster separator.
  Let $U$ be an important cluster separator with respect to the cluster separator $S \cap C$, i.e. $|U| \leq |S \cap C|$, $R^-_{G[C] - U}(\Vout \cap C) \subsetneq R^-_{G[C] - S}(\Vout \cap C)$.
  As $S \cap C$ is not important, we know $U \neq S \cap C$.
  Consider the set $S' = (S \setminus C) \cup U$.

  We now want to show that $S'$ is a hitting $t \to \Vout$-separator that is important with respect to~$S$.
  As $U \neq S \cap C$ we have $S' \neq S$.
  If we can show that $S'$ is important with respect to $S$, then $S$ cannot be an important hitting $t \to V_\textsf{out}$-separator, which yields a contradiction to the choice of~$S$.
  \begin{claim}
    $S'$ is a $t \to \Vout$-separator.
  \end{claim}
  \begin{claimproof}
    Suppose, for sake of contradiction, that there is a $t \to \Vout$-path $P$ in $G^\star - S'$.
    Then path $P$ is intersected by $S$ (as $S$ is an $t \to \Vout$-separator), but as $P$ is disjoint from $S'$ this intersection must lie inside of $C$.
    Let $v$ be the last vertex of $P$ in $S$ and $z$ be the last vertex of $P$.
    Then we know that $P[v, z]$ is a $v \to \Vout$-path in $G - (S - v)$.
    \autoref{lem:force_v_to_vout_path_inside_single_component} implies that $P[v,z]$ lies completely inside $C$.
    Therefore, $v \in V(P[v, z]) \subseteq R^-_{G - U}(\Vout)$.
	But as $v \in S$ we have $v \notin R^-_{G[C] - S}(\Vout)$, a contradiction to $R^-_{G[C] - U}(\Vout) \subsetneq R^-_{G[C] - S}(\Vout)$.
  \end{claimproof}
		
  \begin{claim}
    Set $S'$ is a hitting $t \to \Vout$-separator.
  \end{claim}
  \begin{claimproof}
    In the preceding claim we have shown that $S'$ is a $t \to \Vout$-separator, so we only need to show it is also a hitting $t \to V_\textsf{out}$-separator.
    Suppose, for sake of contradiction, that $S'$ is not hitting, i.e. there is a cycle $O$ of length more than $\ell$ in $G - S'$.
    As $S$ is hitting, we know that~$O$ is intersected by~$S$.
    By choice of $S'$, these intersections lie only in $C$.
    \autoref{thm:bcl_longcyclethroughsolution} tells us that~$O$ has a path $R$ between two different clusters of some component.
    On the other hand \autoref{lem:force_cluster_path_inside_single_component} says that $R$ must be intersected by $S$.
    Therefore,~$R$ must be a path between different clusters in~$C$.
    By choice of $O$, the path $R$ is disjoint from $S' \supseteq U$, showing that $U$ is not a cluster separator---a contradiction.
  \end{claimproof}
		
  It remains to show that $S'$ is important with respect to $S$, i.e. $|S'| \leq |S|$ and $R^-_{G - S'}(\Vout) \subsetneq R^-_{G - S}(\Vout)$.
  For the size bound note that $|S'| = |S| - |S \cap C| + |U| \leq |S|$ by importance of $U$.

  \begin{claim}
    Set $R^-_{G - S'}(\Vout)$ is a proper subset of $R^-_{G - S}(\Vout)$.
  \end{claim}
  \begin{claimproof}
    Let $v \in R^-_{G - S'}(\Vout)$ and $P$ be a $v \to z$-path in $G - S'$ with $z \in \Vout$.
    Suppose, for sake of contradiction, that $v \notin R^-_{G - S}(\Vout)$.
    Then $P$ is intersected by $S$ but not by $S'$, therefore the intersection lies in $C$.
    Let $y$ be the last vertex of $S$ on $P$.
    Then $P[y, z]$ is a $y \to \Vout$-path in $G - (S - y)$.
    By \autoref{lem:force_v_to_vout_path_inside_single_component} we have that $P[y, z]$ lies in $C$.
    Now $y$ lies in $R^-_{G[C] - S'}(\Vout)$ as certified by $P[y, z]$, but not in $R^-_{G[C] - S}(\Vout)$ as $y \in S$.
    This is a contradiction to the choice of $U$.
		
    We still have to show that the inclusion is strict, i.e. $R^-_{G - S'}(\Vout) \neq R^-_{G - S}(\Vout)$.
    There is a $v \in \subsetneq R^-_{G[C] - S}(\Vout \cap C) \setminus R^-_{G[C] - U}(\Vout \cap C)$ by choice of $U$.
    We want to show that $v$ is in $R^-_{G - S}(\Vout)$ but not in $R^-_{G - S'}(\Vout)$.
    Suppose, for sake of contradiction, that there is a $v \to z$-path $P$ in $G - S'$ with $z \in \Vout$.
    If $P$ contains $t$, $P[t, z]$ is an $t \to \Vout$-path in $G - S'$ contradiction to $S'$ being a $t \to \Vout$-separator.
    So $P$ is disjoint from $t$.
    If $z \in C$ then $P$ must be entirely contained in $C$ as it starts and ends there and is disjoint from~$t$.
    But then $P$ is a certificate that $v \in R^-_{G[C] - U}(\Vout \cap C)$---a contradiction to the choice of $v$.
    So $P$ must visit some other $C' \in \mathcal{C}$ different from $C$ with $z \in C'$.
    Let $y$ be the last vertex trough which $P$ enter~$C'$.
    Then $P[y, z]$ lies entirely in $C'$ and $y \in X_{C'}$.
    We have that $S' \cap C' = S \cap C'$ and therefore $P[y, z]$ is a $X_{C'} \to \Vout$-path disjoint from~$S$---a contradiction to \autoref{lem:force_v_to_vout_path_inside_single_component} and the choice of~$S$.
  \end{claimproof}
	
  So we have shown that if $S \cap C$ is not an important cluster separator, $S$ is not an important hitting $t \to V_\textsf{out}$-separator (as witnessed by $S'$).
  This is a contradiction to the choice of~$S$.
\end{proof}

Finding these important cluster separators can now be defined as separate problem, called {\sc Important Cluster Separator in Strong Digraphs of Bounded Circumference}.

\begin{center}
  \framebox[1.0\textwidth]{
    \begin{tabularx}{0.95\textwidth}{rXr}
      \multicolumn{2}{X}{\sc\centering Important Cluster Separator in Strong Digraphs of Bounded Circumference} &\textit{Parameter:} $k + \ell$.\\[2em]
      \textit{Input:}      & \multicolumn{2}{l}{A strong digraph $G$, integers $k,\ell \in\mathbb N$ and sets $X_1, \hdots, X_p, V_\textsf{out} \subseteq V(G)$.}\\
      \textit{Properties:} & $\cf(G) \leq \ell$, $X_i, \Vout \neq \emptyset$, $2 \leq p \leq k(k+1)+1$,\\
      					& $\dist(v,w) \leq 2\ell^2 \quad \forall v,w \in X_i, i \in \lbrace 1, \hdots, p\rbrace$.\\
      \textit{Task:}	& \multicolumn{2}{l}{Find a vertex set $\setCS$ intersecting any important cluster separator}\\
      							& \multicolumn{2}{l}{with respect to $X_1, \hdots, X_\ell, V_\textsf{out}$ of size at most $k$.}
  \end{tabularx}}
\end{center}

\begin{definition}
  Let $(G, k, \ell, t, Z, \Vout)$ be an {\sc Important Hitting Separator in Strong Digraphs} instance.
  A component $C \subseteq V(G)$ is called
  \begin{compactitem}
    \item \emph{trivial} if it has only one cluster and $C \cap \Vout = \emptyset$,
    \item \emph{easily separable} if it has only one cluster and $C \cap \Vout \neq \emptyset$,
    \item \emph{multiway cut separable} if it has more than one cluster and $C \cap \Vout = \emptyset$.
  \end{compactitem}
\end{definition}

\begin{definition}
  Let $G$ be a digraph and let $X_1, \hdots, X_t \subseteq V(G)$.
  A set $S \subset V(G)\setminus (X_1 \cup \hdots \cup X_t)$ is called \emph{$X_1, \hdots, X_t$-multiway cut} if $G - S$ contains no $X_i \to X_j$-path for any $i \neq j$.
	
  The {\sc Directed Multiway Cut} problem asks for a digraph $G$, an integer $p \in \mathbb{Z}_{\geq 0}$ and sets $X_1, \hdots, X_t \subseteq V(G)$ whether there is an $X_1, \hdots, X_t$-multiway cut of size at most $p$.
\end{definition}

\begin{proposition}[\cite{ChitnisEtAl2012}\cite{ChitnisEtAl2015}]
\label{thm:multiway_cut_is_fpt}
  Let $G$ be a digraph, let $p \in \mathbb{N}$ and let $X_1, \hdots, X_t \subseteq V(G)$.
  The {\sc Directed Multiway Cut} problem for $(G, p, X_1, \hdots, X_t)$ can be solved in $2^{\mathcal{O}(p^2)}\cdot \polyn$ time.
  Further, an $X_1, \hdots, X_t$-multiway cut of size at most $p$ can be found in the same time, if it exists.
\end{proposition}

\begin{lemma}
\label{lem:cluster_separator_easy_cases}
  Let $(G, k, \ell, t, Z, \Vout)$ be an {\sc Important Hitting Separator in Strong Digraphs} instance and let $\mathcal{C}$ be the set of strong components of $G - t$.
  For each $C \in \mathcal{C}$ let $L^C_1, \hdots, L^C_t$ be the clusters of $C$.
  The inclusion-wise minimal important cluster separators for $L^C_1, \hdots, L^C_t, \Vout \cap C$ in $G[C]$ are exactly
  \begin{compactitem}
    \item the empty set, if $C$ is trivial,
    \item the important $L^C_1 \to \Vout \cap C$-separators, if $C$ is easily separable,
    \item the $L^C_1, \hdots, L^C_t$-multiway cuts of minimal size, if $C$ is multiway cut separable.
  \end{compactitem}
\end{lemma}
\begin{proof}
  If $C$ is trivial, we have that $C \cap \Vout = \emptyset$ and therefore $R^-_{G[C] - S}(C \cap \Vout ) = \emptyset$ for any $S \subseteq V(G[C])$.
  Thus any cluster separator of minimal size is important.
  But as there are no clusters to separate, the empty set is the only cluster separator of minimal size.
	
  If $C$ is easily separable the only thing an cluster separator has to do, is to hit all $L^C_1 \to \Vout \cap C$-paths.
  Let $S$ be an inclusion-wise minimal important cluster separator.
  Then $S$ is a $L^C_1 \to \Vout \cap C$-separator.
  Furthermore, for every $v \in S$ there is a $L^C_1 \to \Vout \cap C$-path $P_v$ that intersects $S$ only in $v$ (as otherwise~$S$ would not be inclusion-wise minimal).
  This shows that $V(G[C]) = R^+_{G - S}(L^C_1) \cup S \cup R^-_{G[C] - S}(C \cap \Vout )$.
  Translated into importance that means that inclusion-wise minimal cluster separator $S$ for $C$ is important if and only if it is important as $L^C_1 \to \Vout \cap C$-separator.
	
  If $C$ is multiway cut separable, we again have $R^-_{G[C] - S}(C \cap \Vout ) = \emptyset$ for any $S \subseteq V(G[C])$.
  Thus all cluster separators of minimum size are important.
  Furthermore, there are no $L^C_i \to \Vout \cap C$-paths to be cut.
  Thus, any cluster separator is a $L^C_1, \hdots, L^C_t$-multiway cut.
  So the important cluster separators are exactly the $L^C_1, \hdots, L^C_t$-multiway cuts.
\end{proof}

\begin{theorem}
\label{thm:hitting_separator_to_cluster_separator}
  There is an algorithm that solves instances $(G, k, \ell, t, Z, \Vout)$ of {\sc Important Hitting Separator in Strong Digraphs} in time $2^{\mathcal{O}(k^2)} \cdot \polyn$ by making at most one call to an algorithm $\algorithmCS$ for the {\sc Important Cluster Separator in Strong Digraphs of Bounded Circumference} problem.
  Further, the set returned has size at most $2^{\mathcal{O}(k + \log \ell)} + \functionCS(k,\ell)$, where~$\functionCS(k,\ell)$ is a bound on the size of the output of $\algorithmCS$.
\end{theorem}
\begin{proof}
  Let us first describe our algorithm:\\
  \begin{algorithm}[H]
    \SetKwInOut{Input}{Input}\SetKwInOut{Output}{Output}	
	
    \Input{Digraph $G$, integers $k,\ell$, a vertex $t \in V(G)$ and vertex sets $Z, \Vout \subseteq V(G)$}
    \Output{A vertex set $\mathcal{S} \subseteq V(G)$}
		
		Compute the strong components $C \in \mathcal{C}$ of $G - t$ together with their clusters $L^C_i$.\;
		\If{\autoref{cor:unbalanced_clusters} can be applied to $(G, k, \ell, t, Z, \Vout)$}{
			\Return $\setManyClusters$ as in \autoref{cor:unbalanced_clusters}.\;
		}
		Compute $\setSeparateClusters$ as in \autoref{lem:s_is_important_cluster_separator_in_every_strong_component}.\;
		\If{there is a multiway cut separable component $C \in \mathcal{C}$}{
			Solve the \textsc{Directed Multiway Cut} on $(G[C], p, L^C_1, \hdots, L^C_t)$ for all $1 \leq p \leq k$ with \autoref{thm:multiway_cut_is_fpt}.\;
			Let $\setMultiwayCut$ the smallest solution found (or $\setMultiwayCut = \emptyset$ if none existed).\;
			\Return $\setSeparateClusters \cup \setMultiwayCut$.\;
		}
		\If{there is an easily separable $C \in \mathcal{C}$}{
			Compute all important $L^C_1 \to \Vout \cap C$ separators in $G[C]$ with \autoref{thm:fptenumerationofimportantseperators}.\;
			Let $\setImportantSeparator$ the union of all these separators.\;
			\Return $\setSeparateClusters \cup \setImportantSeparator$.\;
		}
		\If{there is an $C \in \mathcal{C}$ which is not trivial}{
			Let $\setCS$ the output of $\algorithmCS$ on $G([C], k, \ell, L^C_1, \hdots, L^C_t, \Vout \cap C)$.\;
			\Return $\setSeparateClusters \cup \setCS$.\;
		}
		\Return $\setSeparateClusters$\;
	\end{algorithm}
	
  Now we argue for correctness:
  Let $(G, k, \ell, t, Z, \Vout)$ an \textsc{Important Hitting Separator in Strong Digraphs} instance.
  If \autoref{cor:unbalanced_clusters} can be applied to this instance, the set $\setManyClusters$ intersects all important hitting $t \to \Vout$-separators that are shadowless with respect to $Z$.
  So we return a correct solution.
  Otherwise, we compute $\setSeparateClusters$ as in \autoref{lem:s_is_important_cluster_separator_in_every_strong_component}.
  Then any important hitting $t \to \Vout$-separator $S$ disjoint from $\setSeparateClusters$ is an important cluster separators for $G[C]$, $C \in \mathcal{C}$.
	
  By \autoref{lem:cluster_separator_easy_cases}, if there is a multiway cut separable component, we know that $S$ is an multiway cut inside~$G[C]$.
  Let now $\setMultiwayCut$ the multiway cut for $G[C]$ computed by our algorithm.
  We know that $|\setMultiwayCut| \leq |S \cap C|$.
  Consider now $S' = (S \setminus C) \cap \setMultiwayCut$.
  We know $|S'| = |S| - |S \cap C| + |\setMultiwayCut| \leq |S|$.
  As we replaced one cluster separator by another, \autoref{lem:cluster_separators_in_every_component_are_hitting_separators} tells us that $S'$ is a hitting $t \to \Vout$-separator.
	
  \begin{claim}
    Set $S'$ is an important hitting $t \to \Vout$-separator range equivalent to $S$.
  \end{claim}
  \begin{claimproof}	
    Consider first the range $R^-_{G - S'}(\Vout)$.
    Now, we want to show that it equals the range~$R^-_{G - S}(\Vout)$.
    Suppose, for sake of contradiction, that there is a $v \in R^-_{G - S'}(\Vout) \setminus R^-_{G - S}(\Vout)$.
    Let $P$ be a $v \to z$-path with $z \in \Vout$ disjoint from $S'$.
    As $\Vout \cap C = \emptyset$ there is a different component $C' \in \mathcal{C}$ with $z \in C'$.
    As $P$ has to enter~$C'$ through a portal vertex $y$, there is a $y \to z$-path from a cluster of $C'$ to $\Vout \cap C'$ completely in $C' \setminus S'$.
    This is a contradiction to $S' \cap C' = S \cap C'$ being a cluster separator for $G[C']$.
    So $R^-_{G - S'}(\Vout) \subseteq R^-_{G - S}(\Vout)$.
    But $S$ is an important hitting $t \to \Vout$-separator, $S'$ is a hitting $t \to \Vout$-separator and $|S'| \leq |S|$.
    Therefore, the ranges $R^-_{G - S'}(\Vout)$ and $R^-_{G - S}(\Vout)$ have to be equal and $S'$ must be important too.
  \end{claimproof}
	
  So $\setMultiwayCut$ intersects an important hitting $t \to \Vout$-separator out of the range equivalence class of $S$.
  As~$S$ was arbitrary, $\setSeparateClusters \cup \setMultiwayCut$ intersects an important hitting $t \to \Vout$-separator out of every range equivalence class.
	
  If there is an easy separable component $C$, we know by \autoref{lem:cluster_separator_easy_cases}, that $S \cap C$ is an important $L^C_1 \to \Vout \cap C$ separator in $G[C]$.
  As we computed all of them, $\setSeparateClusters \cup \setImportantSeparator$ intersects all important hitting $t \to \Vout$-separators (in especially one of every range equivalence class).
	
  If there is a component $C$ that is neither multiway cut separable, easy separable nor trivial, we know the following properties of~$C$:
  \begin{compactitem}
    \item $\cf(G[C]) \leq \ell$ as $G[C]$ is a subgraph of $G - t$.
    \item $\Vout \cap C \neq \emptyset$, as $C$ would be trivial or multiway cut separable otherwise.
    \item $C$ has at least two cluster, as it would be trivial or easy separable otherwise.
    \item $C$ has at most $k(k+1) + 1$ cluster, as we could have applied \autoref{cor:unbalanced_clusters} otherwise.
    \item for $v, w \in L^C_i$ we have $\dist_{G[C]}(v,w) \leq \ell^2$ by definition of clusters.
  \end{compactitem}
  So $(G[C], k, \ell, L^C_1, \hdots, L^C_t, \Vout \cap C)$ is an instance of \textsc{ Important Cluster Separator in Strong Digraphs of Bounded Circumference}.
  By \autoref{lem:cluster_separator_easy_cases}, $S \cap C$ is an important cluster separator in this instance and therefore $\setCS$ intersects it.
  Therefore $\setSeparateClusters \cup \setCS$ intersects all important hitting $t \to \Vout$-separator.
	
  If none of the cases before happened, we know that every component is trivial.
  \autoref{lem:cluster_separator_easy_cases} tells us that~$S$ is the empty set in every component, so $S = \emptyset$.
  But if an important hitting separator $S$ is trivial, we have that $R^-_{G -S}(\Vout)$ is maximal and $|S|$ is minimal, therefore any other set cannot be an important hitting separator.
  Thus all important hitting separators are trivial.
  So we may return any set in this case.
  The set $\setSeparateClusters$ covers the case that $S$ was not disjoint from it.
  Therefore the algorithm is correct.
	
  The run time and size bound follow by combining \autoref{cor:unbalanced_clusters}, \autoref{lem:s_is_important_cluster_separator_in_every_strong_component}, \autoref{thm:multiway_cut_is_fpt} and \autoref{thm:fptenumerationofimportantseperators}.
\end{proof}


\subsection{Finding Important Cluster Separators}
\label{sec:finding_important_cluster_separators}
In this section we want to solve the {\sc Important Cluster Separator in Strong Digraphs of Bounded Circumference} problem.
Let $X_1, \hdots, X_p, \Vout$ be as in the definition of cluster separators.
For ease of notation, let $\mathcal{X} = \bigcup_{i=1}^t X_i$.
We know that every important cluster separator is an $\mathcal{X} \to \Vout$-separator.
Unfortunately, not every important cluster separator intersects an important $\mathcal{X} \to \Vout$-separator.
The goal of this section is to identify vertices on paths between different clusters ($X_i$'s) that are also separated by the $\mathcal{X} \to \Vout$ separator part of the cluster separator.
As in the sections before we fix some arbitrary cluster separator $S$.

By definition of cluster separators, every $X_i \to X_j$-path contains an vertex of $S$.
To guide or search we fix for every ordered pair $(i, j) \in \lbrace 1, \hdots, p\rbrace, i \neq j$ an $X_i \to X_j$-path $P_{i,j}$.
Let $\mathcal{P}$ the set of all these paths~$P_{i, j}$.
As $p \leq k(k + 1) +1$ we have that $|\mathcal{P}| \in \mathcal{O}(k^4)$.

To identify interesting vertices on the $P_{i,j}$'s guiding our search, we introduce the notion of ``outlets''.

\begin{definition}
  Let $v$ be a vertex on an $x\rightarrow y$-path $P$ and $\alpha\in\mathbb N$ be some integer.
  The \emph{$\alpha$-neighborhood} of $v$ on $P$ denoted by $P^\alpha(v)$ is the subpath of $P$ that contains the subpath of $\alpha$ arcs before and after $v$ (or all arcs until the end of the path, if $v$ is closer to an endpoint).
  Formally, $P^\alpha(v) = P[x',y']$, where $|P[x',v]| = \min\{\alpha,|P[x,v]|\}$ and $|P[v,y']| = \min\{\alpha,|P[v,y]|\}$.

  Let $P$ be a path and let $\alpha, \beta \in \mathbb{N}$.
  A vertex $v\in P$ is an $(\alpha, \beta)$-\emph{outlet} of $P$ (with respect to some $\Vout \subseteq V(G)$ if there is a $v\rightarrow \Vout$-path $R$ in $G$ that is at distance at least $\beta$ from every vertex of~$P$ \emph{not} in~$P^{\alpha}(v)$, that is, $\mathsf{dist}_{G^\star[C]}(V(P)\setminus V(P^{\alpha}(v)),R)\geq \beta$.
  
  An outlet is \emph{open} (with respect to $S \subseteq V(G)$) if there is a path $R$ as above such that $R - v$ is disjoint from $S$; otherwise, the outlet is \emph{closed}.
\end{definition}

First we show how to efficiently find outlets on a path $P$:

\begin{lemma}
\label{thm:compute_outlets_in_cubic_time}
  Given a digraph $G$, a set $\Vout \subseteq V(G)$, integers $\alpha, \beta \in \mathbb{N}$ and a path $P$ in $G$.
  Then the set~$\Omega(P)$ of outlets on $P$ with respect to $\Vout$ can be found in $\mathcal{O}(n^3)$ time.
\end{lemma}
\begin{proof}
  In $\mathcal{O}(n^3)$ time we can calculate the distances of every vertex pair in $G$ (for example by the Moore-Bellman-Ford-algorithm).
  Then for every $v \in V(P)$ do the following:
  Iterate over all vertices in $w \in V(P)$.
  If $\min\lbrace \dist_G(v, w), \dist_G(w,v)\rbrace > \alpha$ we mark all vertices $z \in V(G)$ with $\dist(w, z) \leq \beta$.
  This marking can be done in time $\mathcal{O}(n^2)$ for a single vertex~$v$.
  After we marked all vertices for a single~$v$, we try to find a $v \to \Vout$-path using only unmarked vertices by a DFS in time $\mathcal{O}(n + m) = \mathcal{O}(n^2)$.
  The vertex $v$ is an outlet of $P$ if and only if such an path exists.
  By checking this for every vertex $v \in V(P)$ we can find all outlets of $P$ in time $\mathcal{O}(n^3 + |V(P)| \cdot n^2) = \mathcal{O}(n^3)$.
\end{proof}

We now consider outlets on paths in $\mathcal{P}$.
Most useful for our purpose are outlets which are open:
\begin{lemma}
\label{obs:bcl_nosopenoutletpath}
  Let $v$ be an open outlet of a path $P$ with respect to $S$.
  Then there is no $\mathcal{X} \to v$-path in $G - S$.
\end{lemma}
\begin{proof}
  Suppose, for sake of contradiction, that there is an $\mathcal{X} \to v$-path $Q$ in $G - S$.
  In particular we have that $v \notin S$.
  By definition of an open outlet, there is also a $v \rightarrow V_\textsf{out}$-path $R$ in $G - (S \setminus \lbrace v \rbrace) = G - S$.
  This means that $Q \circ R$ is an $\mathcal{X} \rightarrow V_\textsf{out}$-walk in $G - S$, contradicting that~$S$ is a cluster separator.
\end{proof}

Unfortunately, not every path in $\mathcal{P}$ has an open outlet.
But every path in $\mathcal{P}$ is intersected by $S$.
We introduce the notion of ``frontier'' to denominate the vertices of $S$ which would be open outlets if they were outlets in the first place.
\begin{definition}
  Let $S$ be a cluster separator in a digraph $G$ with respect to $X_1, \hdots, X_t, \Vout$.
  The \emph{frontier} $F$ of $S$ is the set of vertices $v \in S$ such that there is an $v \to \Vout$-path in $G - (S \setminus \lbrace v \rbrace)$.
\end{definition}

Let $F$ denote the frontier of $S$.
Our next lemma shows that if an $X_i \to X_j$-path $Q \in \mathcal{P}$ is intersected by $F$, then every $X_i \to X_j$-path either has an open outlet, or at least a closed outlet with some vertex of $S$ nearby.
\begin{lemma}
\label{thm:bcl_openornearclosedoutlets}
  Let $(G, k, \ell, X_1, \hdots, X_p, \Vout)$ be an instance of the {\sc Important Separator in Strong Graphs of Bounded Circumference} problem and let $S$ be any cluster separator.
  Then for $\beta \geq 2 \ell^3$ and $\alpha \geq \ell^3 \beta$, the following holds:
  If there is an $X_i \to X_j$-path $Q$ for some $i \neq j$ that is intersected by the frontier $F$ of $S$ then every $X_i \to X_j$-path $P$ contains
  \begin{compactitem}
    \item either an open $(\alpha, \beta)$-outlet,
	\item or a closed $(\alpha, \beta)$-outlet $\omega$ with a $\omega \rightarrow \Vout$-path $R_\omega$ in $G^\star$ such that $\emptyset \subsetneq R_\omega \cap S \subseteq R^+_\beta(\omega)$.
  \end{compactitem}
\end{lemma}
\begin{proof}
  Let $Q$ be an $x_1 \to y_1$-path and $P$ be an $x_2 \to y_2$-path with $x_1, x_2 \in X_i$ and $y_1, y_2 \in X_j$.
  By assumption there is a vertex $w \in V(Q) \cap F$.
  As $w \in F$ we get a $w \rightarrow u$-path $W$, with $u \in \Vout$, that intersects $S$ only in $w$.
  By \autoref{thm:bcl_pathdistance} we have that $\dist(P, w) \leq (\cf(G) - 1) \cdot \max\lbrace \dist(x_1, x_2), \dist(y_1, y_2)\rbrace + 2(\cf(G) - 2) \leq (\ell - 1)2\ell^2 + 2(\ell - 2) \leq \beta$.
  Let $x$ be the last vertex on $W$ with $\mathsf{dist}(P,x) \leq \beta$; as $\mathsf{dist}(P, W) \leq \dist(P, w) \leq \beta$, there is such a vertex.
  Let $\omega$ be a vertex on $P$ minimizing $\mathsf{dist}_{G}(\omega,x)$, and let $R$ be a shortest $\omega \rightarrow x$-path.
  \begin{claim}
    $\omega$ is an $(\alpha, \beta)$-outlet witnessed by $R \circ W[x,u]$.
  \end{claim}
  \begin{claimproof}
    By definition of $x$, $\omega$ and $R$ we have that $R \circ W[x,u]$ is a $\omega \rightarrow V_\textsf{out}$-path (and not only a walk).
    We will now show that every vertex $z$ at distance more than $\alpha$ from~$\omega$ on $P$ has distance at least $\beta$ from $R \circ W[x,u]$.
    Every vertex on $W[x, u]$---except for $x$---has distance at least $\beta + 1$ from every vertex of~$P$ by definition of $x$.
    Therefore, every vertex within distance~$\beta$ from $P$ lies on $R$.
    Assume there is a vertex $r \in R$ with $\textsf{dist}(P \setminus P^\alpha(\omega), r) \leq \beta$.
    Let $p$ be a vertex on $P \setminus P^\alpha(\omega)$ with $\textsf{dist}(p, r) = \textsf{dist}(P \setminus P^\alpha(\omega), r)$.
    
	If $p$ appears before $\omega$ on $P$, we obtain
    \begin{eqnarray*}
	  \mathsf{dist}_{G}(p,\omega) & \leq & \mathsf{dist}_{G}(p,r) + \mathsf{dist}_{G}(r, \omega)\\
	                                       & \leq & \mathsf{dist}_{G}(p,r) + (\mathsf{cf}(G)-1)\mathsf{dist}_{G}(\omega,r)\\
										   & \leq & \ell\beta,
    \end{eqnarray*}
	using \autoref{thm:bcl_distboundedcircum}.
	Equivalently, if $\omega$ appears before $p$ on $P$ we obtain
	\begin{eqnarray*}
	  \mathsf{dist}_{G}(\omega,p) & \leq & \textsf{dist}_{G}(\omega,r) + \textsf{dist}_{G}(r, p)\\
                                           & \leq & \mathsf{dist}_{G}(\omega,r) + (\mathsf{cf}(G)-1)\mathsf{dist}_{G}(p, r)\\
                                           & \leq & \ell\beta \enspace .
	\end{eqnarray*}
	In either case, we get that the distance between $p$ and $\omega$ is bounded by $\ell\beta \leq \alpha / (\textsf{cf}(G)-1)^2$.
	Note that the segment of $P$ between $p$ and $\omega$ has length at least $\alpha$.
	Therefore $|P[p, \omega]| \geq \alpha = \ell^2 \cdot \ell\beta > (\textsf{cf}(G)-1)^2 \mathsf{dist}_{G}(p,\omega)$ respectively $|P[\omega,p]| > (\textsf{cf}(G)-1)^2 \mathsf{dist}_{G}(\omega, p)$ which is a contradiction to \autoref{thm:bcl_distboundedcircumsquared}.
  \end{claimproof}

  If $\omega$ is an open $(\alpha, \beta)$-outlet, we are done.
  Otherwise, $R_\omega = R \circ W[x,t]$ is intersected by~$S$.
  As~$W$ is disjoint from $S$ except maybe for its first vertex, we have that $S \cap R_\omega = S \cap R$.
  The path~$R$ has length at most~$\beta$ and starts at $\omega$, implying $\emptyset \subsetneq R_\omega \cap S \subseteq R^+_\beta(\omega)$.
\end{proof}

Our next goal is to ensure that every path intersecting $F$ contains at least one open outlet.
For this we want to guess the paths with closed outlets near to a vertex $v \in F$ and find this $v$.
The problem with guessing the closed outlets is that the number of outlets on these paths may not be bounded in $k + \ell$.
Fortunately, paths with many outlets contain always an open outlet:
\begin{lemma}
\label{thm:bcl_firstgammaoutletssuffice}
  Let $G$ be a strong digraph with $\cf(G) \leq \ell$, $\Vout \subseteq V(G)$ and $S \subseteq V(G)$ with $|S| \leq k$.
  Then for $\beta \geq 3 \ell$, the following holds:
  If a path $P$ has at least $\gamma = k\cdot(2\alpha+2)+1$ outlets, one of the first $\gamma$ many $(\alpha, \beta)$-outlets is an open $(\alpha, \beta)$-outlet with respect to $S$.
\end{lemma}
\begin{proof}
  As there are at least $\gamma = k\cdot(2\alpha+2)+1$ many $(\alpha, \beta)$-outlets on $P$, we can choose outlets $\omega_1,\hdots,\omega_{k+1}$ among the first $\gamma$ outlets of $P$ such that $P[\omega_i, \omega_j] \geq 2\alpha+2$ for $1 \leq i < j \leq k+1$.
  For $i = 1,\hdots,k+1$ let us fix an $\omega_i \rightarrow V_\textsf{out}$-path $R_{\omega_i}$ in $G$ that is at distance at least $\beta$ from $P\setminus P^{\alpha}(\omega_i)$.
  
  \begin{claim}
  	The paths $R_{\omega_i}$ are pairwise vertex-disjoint.
  \end{claim}
  \begin{claimproof}
    Suppose, for sake of contradiction, that there is a vertex $v$ that appears on both~$R_{\omega_i}$ and~$R_{\omega_j}$ for $i < j$.
    As $|P[\omega_i,\omega_j]|\geq 2\alpha+2$ by assumption, we can choose a vertex $p$ of $P[\omega_i,\omega_j]$ with $|P[\omega_i,p]| > \alpha$ and $|P[p,\omega_j]| > \alpha$.
    That is, $p$ is not in $P^{\alpha}_{\omega_i}\cup P^{\alpha}_{\omega_j}$, hence every vertex of $V(R_{\omega_i}) \cup V(R_{\omega_j})$ is at distance at least $\beta$ from $p$.
    As every arc of $R_{\omega_j}$ is in a cycle of length at most~$\ell$ (since $G$ is strong and $\cf(G) \leq \ell$), we can use the arcs of these cycles to create a $v\rightarrow \omega_j$-path~$Q$.
    For every vertex $q$ of $Q$, there is some vertex of $R_{\omega_j}$ at distance at most $\ell-1$ from $q$, hence $\mathsf{dist}_{G}(p,Q) \geq \beta - (\ell-1)$.
    Therefore, concatenating~$R_{\omega_i}$ and $Q$, we can obtain an $\omega_i\rightarrow \omega_j$-path whose vertices are at distance at least $\beta-(\ell -1) > 2\textsf{cf}(G^\star)$ from $p$ in~$G^\star$, contradicting \autoref{thm:bcl_distsecondpath}.
 \end{claimproof}  
  
  Since the paths $R_{\omega_i}$ are pairwise vertex-disjoint, there is some $j\in\{1,\hdots,k+1\}$ such that~$R_{\omega_j}$ is disjoint from $S$ and thus $\omega_j$ is an open outlet.
\end{proof}

We now have all restrictions on $\alpha$ and $\beta$ and choose $\beta = 3\ell^3$ and $\alpha = \ell^3\beta = 3\ell^6$.
This sets~$\gamma := k(6\ell^6+2)+1$.
\autoref{thm:bcl_firstgammaoutletssuffice} defines the separation of $\mathcal{P}$ in the disjoint union $\mathcal{P} = \mathcal{P}_\textsf{long} \uplus \mathcal{P}_\textsf{short}$, where $\mathcal{P}_\textsf{long}$ contains all paths of $\mathcal{P}$ with at least $\gamma$ outlets.

For the paths in $\mathcal{P}_\textsf{short}$ we want to eliminate closed vertices by guessing nearby vertices of~$S$:
\begin{lemma}
\label{thm:near_closed_outlet_guessing_set}
  Let $G$ be a strong digraph with $\cf(G) \leq \ell$ and $\Vout \subseteq V(G)$.
  Moreover, let $\omega$ an $(\alpha, \beta)$-outlet on a path $P$ with respect to $\Vout$.
  Then in time $2^{\mathcal{O}(k^2 \log k \log \ell)} \cdot \polyn$ we can compute a set $\mathcal{S}_\omega \subseteq V(G)$ of size at most $2^{\mathcal{O}(k^2\log k \log l)}\beta \cdot \log n$ such that for every set $S \subseteq V(G)$ with
  \begin{compactitem}
    \item $|S| \leq k$,
    \item $\omega$ is closed $(\alpha, \beta)$-outlet with respect to $S$ and
    \item there is an $\omega \to \Vout$-path $R_\omega$ with $\emptyset \subsetneq V(R_\omega) \cap S \subseteq R^+_\beta(\omega)$
  \end{compactitem}		
  we have that $\mathcal{S}_\omega \cap S \neq \emptyset$.
\end{lemma}
\begin{proof}
  Our algorithm works as follows:\\
  \begin{algorithm}[H]
    \SetKwInOut{Input}{Input}\SetKwInOut{Output}{Output}	
	
    \Input{Digraph $G$, integers $k, \ell$, a vertex $\omega \in V(G)$ and a vertex set $\Vout \subseteq V(G)$}
    \Output{A vertex set $\mathcal{S}_\omega \subseteq V(G)$}
    Use \autoref{thm:bcl_smallpathwitnessset} on $\omega$ and $\Vout$ to obtain a set $\Vout' \subseteq \Vout$.\;
		Let $\mathcal{S}_\omega = \emptyset$.\;
		\ForEach{$v \in \Vout'$} {
			Compute a set $\mathcal{R}_v$ of $k$-representative $\omega \to v$-paths by \autoref{thm:bcl_computerepresenetativeforboundedcf}.\;
			\ForEach{$R \in \mathcal{R}_v$}{
				Add $V(R) \cap R^+_\beta(\omega)$ to $\mathcal{S}_\omega$.\;
			}
		}		
	\end{algorithm}
	
  We first care about the correctness of our algorithm.
  Let $S$ be as in the statement of the lemma.
  Then we know that there is a $\omega \to \Vout$-path $R_\omega$ with $\emptyset \subsetneq V(R_\omega) \cap S \subseteq R^+_\beta(\omega)$.
  Let $S_\textsf{near} = V(R_\omega) \cap S$ and $S_\textsf{far} = S \setminus S_\textsf{near}$.
  Then $G - S_\textsf{far}$ contains an $\omega	 \to \Vout$-path, namely $R_\omega$.
  We have $|S_\textsf{far}| \leq |S| \leq k$.
  By \autoref{thm:bcl_smallpathwitnessset} there is a $v \in \Vout'$ such that an $\omega \to v$-path exists in $G - S_\textsf{far}$.
  The algorithm computed a set $\mathcal{R}_v$ of $k$-representative $\omega \to v$-paths.
  By \autoref{thm:bcl_computerepresenetativeforboundedcf} this set contains a $\omega \to v$-path $R$ in $G - S_\textsf{far}$.
  We want to show that $V(R) \cap R^+_\beta(\omega)$ contains an element of $S$.
  As $\omega$ is a closed $(\alpha, \beta)$-outlet with respect to $S$ we know that there is no $\omega \to v$-path in $G - S$ by $v \in \Vout' \subseteq \Vout$.
  So $S$ has to intersect $R$.
  But $R$ is disjoint from~$S_\textsf{far}$, therefore~$R$ intersects $S_\textsf{near}$.
  This yields $\left(V(R) \cap R^+_\beta(\omega)\right) \cap S = V(R) \cap \left(R^+_\beta(\omega) \cap S\right) = R \cap S_\textsf{near} \neq \emptyset$.
  So we introduced an element of $S$ to $\mathcal{S}_\omega$.
	
  Now for the size bound on $\mathcal{S}_\omega$.
  By our algorithm we have 
  \begin{displaymath}		
    |\mathcal{S}_\omega| \leq \sum_{v \in \Vout'} \sum_{R \in \mathcal{R}_v} |V(R) \cap R^+_\beta(\omega)|  \enspace .
  \end{displaymath}
	
  \begin{claim}
  \label{claim:bcl_pathintersectionwithforwardconeisbounded}
    For every path $R$ starting in $\omega$ we have $|R \cap R^+_\beta(\omega)| \leq (\ell - 1)^2\beta+1$.
  \end{claim}
  \begin{proof}
    Suppose, for sake of contradiction, that $|R \cap R^+_\beta(\omega)| >(\ell - 1)^2\beta + 1$.
    Let $v$ be the last vertex on~$R$ in $R^+_\beta(\omega^\star)$ when traversed from $\omega$, yielding $|R[\omega,v]| > (\ell - 1)^2\beta$.
    But by \autoref{thm:bcl_distboundedcircumsquared} we get $|R[\omega,v]| \leq (\mathsf{cf}(G) - 1)^2 \mathsf{dist}_G(\omega,v) \leq (\ell - 1)^2\beta < |R[\omega,v]|$, a contradiction.
  \end{proof}
	
  Plugging in the size bounds of \autoref{claim:bcl_pathintersectionwithforwardconeisbounded}, \autoref{thm:bcl_smallpathwitnessset} as well as \autoref{thm:bcl_computerepresenetativeforboundedcf}, we get:
  \begin{align*}		
    |\mathcal{S}_\omega| &\leq \sum_{v \in \Vout'} \sum_{R \in \mathcal{R}_v} |V(R) \cap R^+_\beta(\omega)|\\
                         &\leq (k+1)4^{k+1} \cdot \ell^{\mathcal{O}(k^2 \log k)}\log n \cdot \left[(\ell - 1)^2\beta + 1\right]\\
                         &\leq 2^{\mathcal{O}(\log k + k + k^2 \log k \log \ell + \log \ell)}\beta \cdot \log n = 2^{\mathcal{O}(k^2\log k \log l)}\beta \cdot \log n.
  \end{align*}
  For the run time we first compute the distances from $\omega$ to all other vertices in $G$ by a simple BFS.
  Then we compute $\Vout'$ in time $2^{\mathcal{O}(k)} \cdot \polyn$.
  For at most $|\Vout| \leq (k+1)4^{k+1}$ choices of $v$ we compute~$\mathcal{R}_v$.
  By \autoref{thm:bcl_computerepresenetativeforboundedcf}, this can be done in time $2^{\mathcal{O}(k)} \cdot \ell^{\mathcal{O}(k^2 \log k)} \cdot \polyn = 2^{\mathcal{O}(k^2 \log k \log \ell)} \cdot \polyn$.
  Also checking whether the vertices of the $(k+1)4^{k+1} \cdot \ell^{\mathcal{O}(k^2 \log k)}\log n = 2^{\mathcal{O}(k^2 \log k \log \ell)} \cdot \log n$ paths lie in $R^+_\beta(\omega)$ can be done in the stated run time, using the precomputed distances.
\end{proof}

We get a simple corollary about open outlets in cluster separators disjoint from $S_\omega$.
\begin{corollary}
\label{cor:solutions_disjoint_of_s_omega}
  Let $(G, k, \ell, X_1, \hdots, X_p, \Vout)$ be an instance of the {\sc Important Cluster Separator in Strong Digraphs of Bounded Circumference} problem.
  Set $\alpha = 3\ell^6$, $\beta=3\ell^3$ and $\gamma = k(6\ell^6 + 2) +1$.
  Let $P$ be an $X_i \to X_j$-path in $G$ for some $i \neq j$ and denote by $\Omega_\gamma(P)$ the first $\gamma$ many $(\alpha, \beta)$-outlets on $P$.
  Then any cluster separator $S$ of size at most $k$ which
  \begin{compactitem}
    \item is disjoint from $\bigcup_{\omega \in \Omega_\gamma(P)} \mathcal{S}_\omega$ (for $\mathcal{S}_\omega$ as in \autoref{thm:near_closed_outlet_guessing_set}) and
    \item has its frontier $F$ intersect some $X_i \to X_j$-path $Q$
  \end{compactitem}
  has an open $(\alpha, \beta)$-outlet in $\Omega_\gamma(P)$.
\end{corollary}
\begin{proof}
  By \autoref{thm:bcl_firstgammaoutletssuffice} we know that if $P$ has at least $\gamma$ many $(\alpha, \beta)$-outlets, one of the first $\gamma$ must be an open one, i.e. lie in $\Omega_\gamma(P)$.
  So we can restrict ourselves to paths with at most $\gamma$ many $(\alpha, \beta)$-outlets i.e. all outlets lie in $\Omega_\gamma(P)$.
	
  As some $X_i \to X_j$-path $Q$ is intersected by the frontier of $S$, \autoref{thm:bcl_openornearclosedoutlets} states that $P$ either has an open $(\alpha, \beta)$-outlet with respect to $S$ or has a closed $(\alpha, \beta)$-outlet $\omega^\star$ with a $\omega^\star \rightarrow \Vout$-path~$R_\omega$ in $G^\star$ such that $\emptyset \subsetneq R_\omega \cap S \subseteq R^+_\beta(\omega^\star)$.
  Again, if $P$ has an open outlet, it is in $\Omega_\gamma(P)$ and we are done.
	
  Now, suppose for sake of contradiction, that this is not the case and that there is a closed $(\alpha, \beta)$-outlet~$\omega^\star$ as above.
  By \autoref{thm:near_closed_outlet_guessing_set} we have that $\mathcal{S}_{\omega^\star}$ intersects $S$.
  Thus $S$ intersects $\mathcal{S}_{\omega^\star} \subseteq \bigcup_{\omega \in \Omega_\gamma(P)} \mathcal{S}_\omega$ --- a contradiction to the choice of $S$.
\end{proof}

\begin{lemma}
\label{thm:inclusion_wise_minimal_cluster_separators_have_all_vertices_reachable_from_x}
  Let $(G, k, \ell, X_1, \hdots, X_p, \Vout)$ be an {\sc Important Cluster Separator in Strong Digraphs of Bounded Circumference} instance.
  Then in any inclusion-wise minimal cluster separator $S$ every vertex $v \in S$ is reachable from $\mathcal{X}$ in $G - (S \setminus \lbrace v \rbrace)$.
\end{lemma}
\begin{proof}
  Suppose, for sake of contradiction, that there is a vertex $v \in S$ that is not reachable from~$\mathcal{X}$ in $G - (S \setminus \lbrace v \rbrace)$.
  Consider $S' = S \setminus \lbrace v \rbrace$.
  As $S$ was an inclusion-wise minimal cluster separator, $G - S'$ must contain an $X_i \to X_j \cup \Vout$-path $P$ for some $i \neq j$.
  This path $P$ does not exists in $G - S$ and therefore has to be intersected by $S \setminus S' = \lbrace v \rbrace$.
  Let $x$ be the start vertex of $P$.
  Then $P[x,v]$ as a subpath of $P$ exists in $G - S' = G - (S \setminus \lbrace v \rbrace)$ and thus is a certificate that $v$ is reachable from an $x \in \mathcal{X}$ in $G - (S \setminus \lbrace v \rbrace)$ --- a contradiction to the choice of $v$.
\end{proof}

Now we want to use \autoref{thm:bypassing_through_nearby_paths}.
For this we need the following definition:
\begin{definition}
  Given an integer $t \in \mathbb{N}$, a path $P$ and a vertex $v$ on $P$, the \emph{landing strip} $L^t_P(v)$ is the vertex~$v$ and its $t$ predecessors on $P$ or all predecessors if there are less than $t$.
\end{definition}

\begin{lemma}
\label{thm:important_cluster_separator_contains_important_separator}
  Let $(G, k, \ell, X_1, \hdots, X_p, \Vout)$ be an \textsc{Important Cluster Separator in Strong Digraphs of Bounded Circumference} instance.
  Let $\mathcal{P}$ be a set of paths that contains an $X_i \to X_j$-path $P_{i,j}$ for every ordered pair $(i,j) \in \lbrace 1, \hdots, p\rbrace^2, i \neq j$.
  Set $\alpha = 3\ell^6$, $\beta=3\ell^3$ and $\gamma = k(6\ell^6 + 2) +1$.
  Denote by $\Omega_\gamma(P)$ the set of the first $\gamma$ many $(\alpha, \beta)$-outlets on a path $P$.
  Then every important cluster separator $S$ of size at most $k$ that is disjoint from $\bigcup_{P \in \mathcal{P}} \bigcup_{\omega \in \Omega_\gamma(P)} \mathcal{S}_\omega \cup L^{3\ell^7k}_P(\omega)$ (with $\mathcal{S}_\omega$ as in \autoref{thm:near_closed_outlet_guessing_set}) contains an important $\mathcal{X} \to \Vout \cup V_\Omega$-separator for some $V_\Omega$ that contains for each path $P$ at most one landing strip $L^{3\ell^7k}_P(\omega)$ with $\omega \in \Omega_\gamma(P)$.
\end{lemma}
\begin{proof}
  Let $S$ be an important cluster separator as in the statement of the lemma.
  We can assume~$S$ to be inclusion-wise minimal, by the following argument:
  if any important cluster separator $S' \subsetneq S$ contains an important separator with properties as in the theorem statement, then also $S$ includes this important separator.
  Note that two important cluster separators with $S' \subsetneq S$ can exists, for example for $S = S' \cup \lbrace v \rbrace$ with $v \in R^-_{G - S'}(\Vout)$.
	
  Define $V_\Omega$ to contain for every path $P \in \mathcal{P}$ the landing strip $L^{3\ell^7k}_P(\omega)$ of an arbitrary open $(\alpha, \beta)$-outlet $\omega \in \Omega_\gamma(P)$ if such an outlet exists on $P$.
  We will show that in fact the frontier~$F$ of~$S$ is an important $\mathcal{X} \to \Vout \cup V_\Omega$-separator.
  This is done in two steps.
  First, we show that~$F$ is an $\mathcal{X} \to \Vout \cup V_\Omega$-separator.
  Then we assume that $F$ is not important, replace it by an important separator and get a contradiction by showing that $S$ was not important.
	
  \begin{claim}
    The frontier $F$ of $S$ is an $\mathcal{X} \to \Vout \cup V_\Omega$-separator.
  \end{claim}
  \begin{proof}
    Suppose, for sake of contradiction, that there is an $\mathcal{X} \to v$-path $P$ in $G - F$ with $v \in \Vout \cup V_\Omega$.
	If $v \in V_\Omega$ we can assume that $v$ is an open $(\alpha, \beta)$-outlet by prolonging the path by vertices of the landing strip till the end.
	None of the vertices of the landing strip is contained in $S \supseteq F$, by assumption.
	As $v$ is open, there is a $v \to \Vout$-path $Q$ disjoint from $S \supseteq F$.
	So either $P$ (if $v \in \Vout$) or $P \circ Q$ (if $v \in V_\Omega$) is an $\mathcal{X} \to \Vout$-path.
	As $S$ is an $\mathcal{X} \to \Vout$-separator and $Q$ is disjoint from $S$ we have that $P$ must be intersected by $S$.
	Let $s$ be the last vertex of $S$ on $P$, i.e. $P[s, v]$ intersects~$S$ only in $s$.
	Then $P[s,v]$ or $P[s,v] \circ Q$ is a certificate that $s$ should be in the frontier $F$, but $F$ was disjoint from $P$ and $Q$---a contradiction.
  \end{proof}
	
  If $F$ is an important $\mathcal{X} \to \Vout \cup V_\Omega$-separator we are done as $S$ contains $F$.
  So assume that~$F$ is not an important $\mathcal{X} \to \Vout \cup V_\Omega$-separator.
  Let $F'$ be an important $\mathcal{X} \to \Vout \cup V_\Omega$-separator with $R^+_{G - F}(\mathcal{X}) \subseteq R^+_{G - F'}(\mathcal{X})$ and $|F'| \leq |F|$.
  We now consider the set $S' = (S \setminus F) \cup F'$.
  Obviously, $|S'| \leq |S| - |F| + |F'| \leq |S|$.
	
  \begin{claim}
    The set $S'$ is a cluster separator.
  \end{claim}
  \begin{proof}
    As $S'$ contains the $\mathcal{X} \to \Vout \cup V_\Omega$ separator $F'$, there is no $\mathcal{X} \to \Vout$-path in $G - S'$.
    Now suppose, for sake contradiction, that there is an $X_i \to X_j$-path $Q$ in $G - S'$ for some $i \neq j$.
    By choice of $\mathcal{P}$ we have a $P \in \mathcal{P}$ that is also an $X_i \to X_j$-path in $G - S'$.
		
    The path $Q$ was intersected by the original set $S$, as $S$ was a cluster separator, but is not intersected by~$S'$.
    Therefore, there is a vertex in $Q \cap (S' \setminus S) \subseteq F$.
    \autoref{cor:solutions_disjoint_of_s_omega} then tells us that~$P$ has an open $(\alpha, \beta)$-outlet with respect to $S$ among the first $\gamma$ ones.
    So there is an $\omega \in V_\Omega \cap \Omega_\gamma(P)$.
		
    Let $P$ be an $x_1 \to y_1$-path and let $Q$ be an $x_2 \to y_2$-path.
    Then \mbox{$\dist(x_1, x_2), \dist(y_1, y_2) \leq 2\ell^2$}.
    Further, $Q$ is disjoint from $F' \subset S'$.
    Also $\omega$ and its landing strip $L^{3\ell^7 k}_P(\omega)$ are in $V_\Omega$ and therefore disjoint from $F'$ (as it is a $\mathcal{X} \to V_\Omega$-separator).
    The landing strip is a subpath of~$P$, has length
    \begin{equation*}
      3\ell^7 k \geq \ell^5k \cdot (2\ell^2 + 2)
                \geq \ell^5k \cdot \left(  \max\lbrace \dist(x_1, x_2), \dist(y_1, y_2)\rbrace + 2\right),
    \end{equation*}
    and ends in $\omega$.
    Together with $|F'| \leq |S'| \leq |S| \leq k$, \autoref{thm:bypassing_through_nearby_paths} guarantees us the existence of an $x_2 \to \omega$-path disjoint from $F'$, in contradiction to $F'$ being an $\mathcal{X} \to V_\Omega$-separator.
  \end{proof}
  Now all that remains to show is $R^-_{G - S'}(\Vout) \subsetneq R^-_{G - S}(\Vout)$ to derive a contradiction to $S$ being an important cluster separator.
	
  We first show that $R^-_{G - S'}(\Vout) \subseteq R^-_{G - S}(\Vout)$, and then prove that they are not equal.
  Suppose, for sake of contradiction, that there is a vertex $v \in R^-_{G - S'}(\Vout) \setminus R^-_{G - S}(\Vout)$.
  Then there is a $v \to z$-path $P$ for some $z \in \Vout$ that is disjoint from $S'$ but not of $S$.
  Let $w$ be the last vertex of $P$ which is in $S$.
  Then $P[w, z]$ is a certificate that $w$ lies in the frontier $F$ of $S$.
  As we chose $S$ to be inclusion-wise minimal, \autoref{thm:inclusion_wise_minimal_cluster_separators_have_all_vertices_reachable_from_x} tells us that there is a $x \to w$-path $R$ in $G - (S \setminus \lbrace w \rbrace)$ for some $x \in \mathcal{X}$.
  This path $R$ lies in $R^+_{G - S}(\mathcal{X}) \cup \lbrace w \rbrace \subseteq R^+_{G - F}(\mathcal{X}) \cup \lbrace w \rbrace \subseteq R^+_{G - F'}{\mathcal{X}} \cup \lbrace w \rbrace$.
  So $R$ is disjoint from $F'$ except for maybe $w$.
  Also $P[w, z]$ as subpath of $P$ is disjoint from $F' \subseteq S'$, this time including $w$.
  So $R \circ P[w,z]$ is a $\mathcal{X} \to \Vout$-path in $G - F$ --- a contradiction to $F$ being an $\mathcal{X} \to \Vout$-separator.
	
  Now suppose for sake of contradiction, that $R^-_{G - S'}(\Vout) = R^-_{G - S}(\Vout)$.
  Let $Z$ be the set of vertices $z$ that have an arc $(z, r)$ with $r \in R^-_{G - S}(\Vout)$.
  \begin{claim}
    The sets $F$ and $Z$ are equal.
  \end{claim}
  \begin{proof}
    As $R^-_{G - S}(\Vout)$ is disjoint from $S$ every $z \in Z$ has an $z \to \Vout$-path that is disjoint from~$S$ except for maybe $z$.
    But as the $z$'s do not lie in $R^-_{G - S}(\Vout)$ they have to lie in $S$.
    Therefore, $Z \subseteq F$.

    Suppose, for sake of contradiction, that there is a $v \in F \setminus Z$.
    Then, by $F$ being the frontier of~$S$, there is a $v \to \Vout$-path $P$ disjoint from $S - v$.
    The path $P$ has to enter $R^-_{G - S}(\Vout)$ at some point.
    Let $w$ be the vertex on $P$ before entering $R^-_{G - S}(\Vout)$.
    This vertex is in $Z$, and therefore in~$S$.
    But $v \notin Z$, and thus $P$ is intersected by $S$, contradicting the choice of~$P$.
  \end{proof}
  As $R^-_{G - S'}(\Vout) = R^-_{G - S}(\Vout)$, we have that $Z$ also must lie in $S'$.
  But $S'$ contains no vertices of $F$ if they are not in $F'$, so $F = Z \subseteq F'$.
  Now $F'$ was chosen as an important separator, and therefore $|F'| \leq |F|$ which implies $F = F'$.
  This is a contradiction to the assumption of $F$ being not important.
\end{proof}

\begin{theorem}
\label{thm:solving_important_cluster_separator_in_strong_graphs_of_bounded_circumference}
  The {\sc Important Cluster Separator in Strong Digraphs of Bounded Circumference} problem can be solved in time $2^{\mathcal{O}(k^4\log k \log \ell)}\cdot \polyn$ by a set $\setCS$ of size at most $2^{\mathcal{O}(k^2\log k\log\ell)}\log n + 2^{\mathcal{O}(k^4\log k \log \ell)}$.
\end{theorem}
\begin{proof}
  Let $(G, k, \ell, X_1, \hdots, X_p, \Vout)$ be an instance of the \textsc{Important Cluster Separator in Strong Digraphs of Bounded Circumference} problem.
  Set $\alpha = 3\ell^6$, $\beta=3\ell^3$ and $\gamma = k(6\ell^6 + 2) +1$.
  Then our algorithm works as follows:\\
	
  \begin{algorithm}[H]
    \SetKwInOut{Input}{Input}\SetKwInOut{Output}{Output}	
	
    \Input{A digraph $G$, integers $k, \ell$ and vertex sets $X_1, \hdots, X_p, \Vout \subseteq V(G)$}
    \Output{A vertex set $\setCS \subseteq V(G)$}
    Generate a set $\mathcal{P}$ of paths that contains an arbitrary $X_i \to X_j$-path for every ordered pair $(i,j) \in \lbrace 1, \hdots, p\rbrace^2, i \neq j$\;
    Let $\setCS = \emptyset$\;
    \ForEach{$P \in \mathcal{P}$} {
      Compute the set $\Omega(P)$ of $(\alpha, \beta)$-outlets on $P$ by \autoref{thm:compute_outlets_in_cubic_time}\;
      Let $\Omega_\gamma(P)$ be the set of the first $\gamma$ many $(\alpha, \beta)$-outlets on $P$\;
      \ForEach{$\omega \in \Omega_\gamma(P)$}{
        Compute $S_\omega$ by \autoref{thm:near_closed_outlet_guessing_set} and add it to $\setCS$\;
        Compute the landing strip $L^{3\ell^7 k}_P(\omega)$ and add it to $\setCS$\;
      }
      Let $\mathcal{L}_P = \lbrace L^{3\ell^7 k}_P(\omega) | \omega \in \Omega_\gamma(P)\rbrace$\;
    }
    \ForEach{$(Z_P)_{P \in \mathcal{P}} \in \bigtimes_{P \in \mathcal{P}} \left(\mathcal{L}_P \cup \lbrace \emptyset\rbrace \right)$}{
      Let $V_\Omega = \bigcup_{P \in \mathcal{P}} Z_P$\;
      \ForEach{important $\mathcal{X} \to \Vout \cup V_\Omega$-separator $F$} {
        Add $V(F)$ to $\setCS$\;
      }
    }
  \end{algorithm}

  For correctness we consider \autoref{thm:important_cluster_separator_contains_important_separator}.
  We computed a set $\mathcal{P}$ as in the lemma.
  Every important cluster separator that intersects some $S_\omega$ or $L^{3\ell^7 k}_P(\omega)$ for an $\omega \in \omega_\gamma(P)$ intersects our $\setCS$.
  So we only have to argue about those important cluster separators disjoint from $\bigcup_{P \in \mathcal{P}}\bigcup_{\omega \in \Omega_\gamma(P)} S_\omega$.
  By \autoref{thm:important_cluster_separator_contains_important_separator} these important cluster separators contain an important $\mathcal{X} \to \Vout \cup V_\Omega$-separator for some $V_\Omega$ that contains for every $P \in \mathcal{P}$ at most one landing strip $L^{3\ell^7 k}_P(\omega)$ for some $\omega \in \Omega_\gamma(P)$.
  Our algorithm iterates over all possible choices for $V_\Omega$ and the important separator.
  Therefore, the resulting set $\setCS$ contains this separator and thus also intersects the important cluster separators disjoint from $\bigcup_{P \in \mathcal{P}}\bigcup_{\omega \in \Omega_\gamma(P)} S_\omega$.
	
  The size and run time bound follow from \autoref{thm:near_closed_outlet_guessing_set}, $|\mathcal{P}| = p(p - 1) \leq 6k^4$, and that there are at most~$4^k$ important $X \to Y$-separators for fixed $X$ and $Y$.
\end{proof}

\subsection{Putting Everything Together}
This section combines the previous sections to an overall algorithm solving \DLCHS{}.
For the analysis we need to bound expressions of type $\log^{f(k)} n$ by some function $g(k) \cdot n^{\mathcal{O}(1)}$.
This we do by the following lemma:

\begin{lemma}
\label{thm:boundingpowersoflogarithm}
  For $n \geq 4$ and $f(k) \geq 0$ we have \mbox{$(\log n)^{f(k)} \leq f(k)^{2f(k)} + \frac{n}{2^{f(k)}} \in 2^{\mathcal{O}(f(k)\log f(k))} + n$}.
\end{lemma}
\begin{proof}
  We distinguish two cases, and add the upper bounds for $(\log n)^{f(k)}$ from both cases.

  If $f(k) \leq \frac{\log n}{1 + \log\log n}$ then we have
  $n \geq 2^{f(k)^2(1 + \log \log n)} = (2\log n)^{f(k)}$,
  which is equivalent to $(\log n)^{f(k)} \leq \frac{n}{2^{f(k)}}$.

  Otherwise, we have $f(k) > \frac{\log n}{1 + \log\log n}$.
  For $n \geq 4$ it then holds
  \begin{displaymath}
    \frac{f(k)^2}{\log n} > \frac{\log n}{(1 + \log \log n)^2} \geq 1.
  \end{displaymath}
  This is equivalent to $\log n \leq f(k)^2$ which implies $(\log n)^{f(k)} \leq f(k)^{2f(k)}$.
	
  Adding both cases we get $(\log n)^{f(k)} \leq f(k)^{2f(k)} + \frac{n}{2^{f(k)}}$ which by $f(k) = 2^{\log f(k)}$ lies in $2^{\mathcal{O}(f(k)\log f(k))} + n$.
\end{proof}

Now, by combining \autoref{thm:original_to_isolating_lchs}, \autoref{thm:isolating_lchs_to_important_hitting_separator}, \autoref{thm:hitting_separator_to_cluster_separator} and \autoref{thm:solving_important_cluster_separator_in_strong_graphs_of_bounded_circumference}, we get the following:
\begin{theorem}
  There is an algorithm solving \DLCHS{} with run time \overallruntime{}.
\end{theorem}
\begin{proof}
  By using \autoref{thm:solving_important_cluster_separator_in_strong_graphs_of_bounded_circumference} as an oracle for \autoref{thm:hitting_separator_to_cluster_separator}, we get an algorithm solving \textsc{Important Hitting Separator in Strong Graphs} in time $\left(2^{\mathcal{O}(k^2)} + 2^{\mathcal{O}(k^4\log k \log \ell)}\right) \cdot \polyn$ producing a set of size at most $2^{\mathcal{O}(k + \log \ell)} + 2^{\mathcal{O}(k^2\log k\log\ell)}\log n + 2^{\mathcal{O}(k^4\log k \log \ell)} =  2^{\mathcal{O}(k^4\log k \log \ell)} \log n $.
	
  Using this algorithm in \autoref{thm:isolating_lchs_to_important_hitting_separator}, we obtain an algorithm for \textsc{Isolating Long Cycle Hitting Set} with run time
  \begin{displaymath}
    2^{\mathcal{O}(k^2)}|T| \cdot \log^2(n) \cdot \left(2^{\mathcal{O}(k^2)} + 2^{\mathcal{O}(k^4\log k \log \ell)}\right) \cdot \polyn = 2^{\mathcal{O}(k^4\log k \log \ell)} \cdot \polyn
  \end{displaymath}
  producing a set of size at most
  \begin{displaymath}
    2^{\mathcal{O}(k^2)}|T| \cdot \log^2(n) \cdot 2^{\mathcal{O}(k^4\log k \log \ell)} \log n = |T| \cdot 2^{\mathcal{O}(k^4\log k \log \ell)} \log^3(n)\enspace.
  \end{displaymath}

  If we plug this into \autoref{thm:original_to_isolating_lchs}, we may assume that $|T| \leq k$.
  Using $\log^{3k}(n) \leq 2^{\mathcal{O}(k \log k)} + n$ (by \autoref{thm:boundingpowersoflogarithm}), we obtain our final algorithm for \DLCHS{} with run time
  \begin{align*}
    &2^{\mathcal{O}(\ell k^3\log k)} \cdot \left( |T| \cdot 2^{\mathcal{O}(k^4\log k \log \ell)} \log^3(n) \right)^k\cdot 2^{\mathcal{O}(k^4\log k \log \ell)} \cdot \polyn\\
    \leq & 2^{\mathcal{O}(\ell k^3\log k + k^5\log k\log \ell)} \cdot \log^{3k}(n) \cdot \polyn\\
    \leq & 2^{\mathcal{O}(\ell k^3\log k + k^5\log k\log \ell)} \cdot \polyn \enspace . \qedhere
  \end{align*}		
\end{proof}

\section{Technical Tools Proofs}
\label{sec:technical_tools_proofs}
This section contains the proofs of theorems found in \autoref{sec:technical_tools}.

We will show that separators of bounded size are also defined only by a set of vertices of bounded size (\autoref{sec:bcl_impseps}).
Then we consider graphs with bounded circumference.
For these we can show some length bounds between paths (\autoref{sec:bcl_propertiesofdirectedgraphswithboundedcircumference}) and ways to bypass small deletion sets (\autoref{sec:bypassing}).
We are also able to obtain $k$-representative sets of paths of bounded size for strongly connected digraphs with bounded circumference, independent of the structure of the paths (\autoref{sec:bcl_representativesetsofpaths}).

\subsection{Important Separators and Consequences}
\label{sec:bcl_impseps}
Given a vertex $x\in V(G)$ and a large vertex set $\{y_1,\hdots,y_r\}\subseteq V(G)$, it is certainly possible that for every~$y_i$ there is a small set $S_i$ of vertices that separates a single $y_i$ from $x$, but does not separate~$y_j$ from $x$ for any $j\not=i$.
An example of such a situation is depicted in \autoref{fig:bcl_onewaysep}.
\begin{figure}[htpb]
  \centering
  \includegraphics{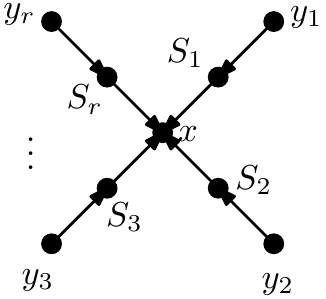}
  \caption{A digraph $G$ in which for every $t_i$ there is a singleton $S_i$ that separates $t_i$ from $s$ but does not separate $t_j$ from $s$ for any $j\neq i$.}
\label{fig:bcl_onewaysep}
\end{figure}

The following simple statement shows that the opposite is not possible: if $r$ is large, then it cannot happen that for every $y_i$ there is a small separator $S_i$ that separates every vertex of $\{y_1,\hdots,y_r\}\setminus\{y_i\}$ but does \emph{not} separate $y_i$ from $x$.
\begin{lemma}
\label{thm:bcl_exceptsep_small}
  Let $G$ be a digraph and let $x,y_1,\hdots,y_r$ be vertices of $G$.
  Let $S_1,\hdots,S_r$ be sets of vertices of size at most $k$ each, such that the following holds for each $i = 1,\hdots,r$:
  \begin{compactitem}
    \item $y_i$ is reachable from $x$ in $H - S_i$, but
    \item for each $j\in\{1,\hdots,r\}\setminus\{i\}$, vertex $y_j$ is not reachable from $x$ in $H - S_i$ (potentially $y_j\in S_i$).
  \end{compactitem}
  Then $r\leq (k+1)4^{k+1}$.
\end{lemma}
\begin{proof}
  Create a graph $G'$ from $G$ by adding a new vertex $y^\star$ together with the arcs $(y_i,y^\star)$, for each $i = 1,\hdots,r$.
  Observe that each vertex $y_i$ is part of an $x  \rightarrow y^\star$-separator $S_i' = S_i\cup\{y_i\}$ of size $k+1$ and moreover $R_{H'\setminus S_i'}^+(x)$ contains some vertex $v_i$ such that $(v_i,y_i)$ is an arc of $G$.
  Therefore, there exists an important $x \rightarrow y^\star$-separator $S_i''$ such that $R^+_{H'\setminus S_i'}(x)\subseteq R^+_{H'\setminus S_i''}(x)$, which implies that $v_i\in R^+_{H'\setminus S_i''}(s)$ and $y_i\in S_i''$.
  Consequently, each vertex $y_i$ belongs to some important $x  \rightarrow y^\star$-separator of size at most $k+1$, and since by \autoref{thm:fptenumerationofimportantseperators} applied with $p = k + 1$ there are at most $(k+1)4^{k+1}$ such vertices, the lemma follows.
\end{proof}

This means that for the set $Y = \{y_1,\hdots,y_r\}$ at most $(k+1)4^{k+1}$ vertices of $Y$ define what $x \rightarrow Y$-separators look like.
We will now show how to construct for $Y$ a small ``witness'' set~$Y'$ of size at most $(k+1)4^{k+1}$ such that all $x \rightarrow Y'$-separators of size at most $k$ are also $x \rightarrow Y$-separators.

\restatesmallpathwitness*
\begin{proof}
  Initially, we start with $Y' = Y$, which certainly satisfies property~\eqref{eqn:bcl_impsepproperty}.
  For every $v\in Y'$ we check whether $Y'\setminus \{v\}$ also satisfies property~\eqref{eqn:bcl_impsepproperty}.
  For this purpose, we need to check whether there is a set $S$ of at most $k$ vertices such that some vertex of $Y$ is reachable from $x$ in $G - S$, but no vertex of $Y'\setminus\{v\}$ is reachable.
  As $Y'$ satisfies the assumptions of the lemma, if $Y$ is reachable, then some vertex of~$Y'$ is reachable.
  Therefore, what we need is a set $S$ such that $v$ is reachable from~$x$ in $G - S$, but no vertex of $Y'\setminus\{v\}$ is reachable.

  Let us introduce a new vertex $y^\star$ into $G$ and add an arc from every vertex of $Y'\setminus\{v\}$ to~$y^\star$.
  Observe that $S$ is an $x \rightarrow y^\star$-separator (clearly, we have $x\notin S$).
  We claim that if there is an $x \rightarrow y^\star$-separator $S$ of size at most $k$ such that $v$ is reachable from $x$ in $G - S$, there is such an important separator~$S'$.
  Indeed, if $S'$ is an important separator with $|S'|\leq |S|$ and $R^+_{G - S}(\{x\})\subseteq R^+_{G - S'}(\{x\})$, then $v$ is reachable from $x$ also in $G - S'$.
  Therefore, we can test existence of the required separator~$S$ by testing every important $s \rightarrow y^\star$-separator of size at most~$k$.
  If none of them satisfies the requirements, then we can conclude that $Y'\setminus\{v\}$ also satisfies property~\eqref{eqn:bcl_impsepproperty} and we can continue the process with the smaller set $Y'\setminus\{v\}$.
  
  Suppose now that for every $v\in Y'$, we have found a set $S_v$ of at most $k$ vertices such that $v$ is reachable from $x$ in $G - S_v$, but $Y' - \{v\}$ is not.
  Then \autoref{thm:bcl_exceptsep_small} implies that $|Y'|\leq (k+1)4^{k+1}$.
\end{proof}

Next, we prove a ``set extension'' of the previous lemma, in which the vertex $x$ is enlarged to a set $X$.
\restatesmallpathwitnessset*
\begin{proof}
  Let us introduce a new vertex $x$ into $G$ and add an arc from $x$ to every vertex of~$X$.
  Let us use the algorithm of \autoref{thm:bcl_smallpathwitness} to find a set $Y'\subseteq Y$ of size at most $(k+1)4^{k+1}$.
  Let~$\overleftarrow{G}$ be the digraph obtained from~$G$ by reversing the orientation of all arcs.
  Add a vertex~$\overleftarrow{x}$ to~$\overleftarrow{G}$ and add an arc $(\overleftarrow{x}, v)$ for every vertex $v\in Y'$.
  Apply the algorithm of \autoref{thm:bcl_smallpathwitness} on $\overleftarrow{G}$ with $\overleftarrow{x}$ playing the role of $x$ and $X$ playing the role of $Y$; let $X'$ be the set returned by the algorithm.
  
  We claim that $X'$ and $Y'$ satisfy the requirements of the lemma.
  Suppose that there is an $X \rightarrow Y$-path~$P$ in $G - S$.
  By the way we obtained $Y'$, we may assume that $P$ ends in~$Y'$.
  Then the reverse of $P$ is a $Y' \rightarrow X$-path in $\overleftarrow{G} - S$.
  Therefore, by the way we obtained~$X'$ there is a path~$Q$ in $\overleftarrow{G} - S$ from $Y'$ to $X'$.
  Now the reverse of $Q$ is an $X' \rightarrow Y'$-path in $G - S$, as required.
\end{proof}

We can further extend above set version to multiple sets $X_i$.

\restatesmallpathwitnesstsets*
\begin{proof}
  For every ordered pair $(i,j)$ apply \autoref{thm:bcl_smallpathwitnessset} to $X_i$ and $X_j$ to obtain sets $X_i^{(i,j)}$ and~$X_j^{(i,j)}$.
  Let
  \begin{displaymath}
    X'_i = \bigcup_{\substack{j=1\\ j \not = i}}^t \left( X_i^{(i,j)} \cup X_i^{(j,i)} \right) \enspace .
  \end{displaymath}
  These have the desired properties, as for a $X_i \rightarrow X_j$-path in $G - S$ for $i \not = j$ there is by construction a $X_i^{(i,j)} \subseteq X'_i \rightarrow X_j^{(i,j)} \subseteq X'_j$-path in $G - S$.
  The size bound follows directly.
\end{proof}

\subsection{Properties of Digraphs with Bounded Circumference}
\label{sec:bcl_propertiesofdirectedgraphswithboundedcircumference}
We now establish some properties of digraphs with bounded ``circumference''.
Recall that for a digraph $G$, its \emph{circumference} $\mathsf{cf}(G)$ is defined as the maximum length of any of its directed cycles or~$0$ if it is acyclic.
Further, recall that the \emph{distance} $\mathsf{dist}_G(x,y)$ between any two vertices $x,y\in V(G)$ in~$G$ the minimum length of a directed path from $x$ to $y$ in $G$.
For vertex sets $X,Y\subseteq V(G)$, their distance $\mathsf{dist}_G(X,Y)$ is defined as the minimum of $\mathsf{dist}_G(x,y)$ over all $x\in X,y\in Y$.

\restatedistboundedcircum*
\begin{proof}
  Note that the statement trivially holds if $x = y$: then every simple path between the two vertices has length 0.
  Let $x',y'$ be any two distinct vertices of $P_1$ such that no internal vertex $P[x',y']$ is on $P_2$ (see \autoref{fig:intersectingpaths}).
  \begin{figure}[htpb]
    \centering
    \includegraphics{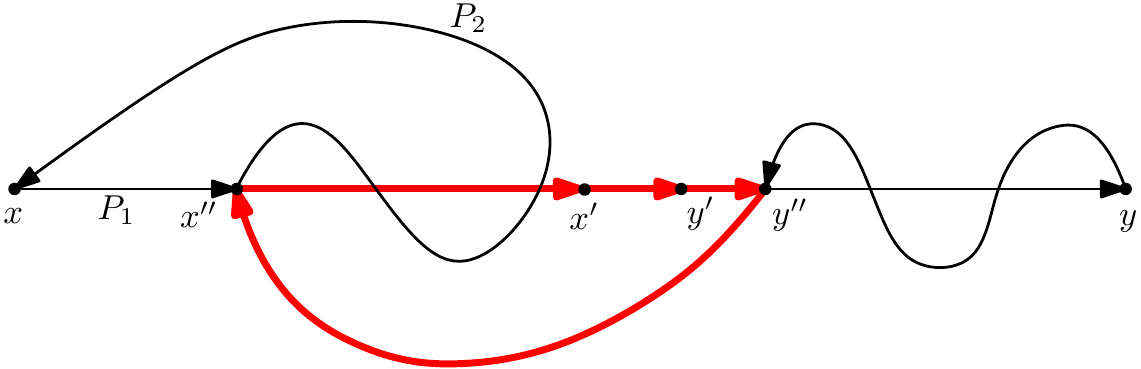}
    \caption{Proof of \autoref{thm:bcl_distboundedcircum}.}
    \label{fig:intersectingpaths}
  \end{figure}
  Going from $y$ to $x$ on $P_2$, let $x''$ be the first vertex of~$P_2$ that is in $P_1[x,x']$ (possibly $x''$ is equal to~$x$ or $x'$) and let $y''$ be the last vertex of $P_2$ before~$x''$ that is on $P_1$.
  Note that $y''$ has to be between $y'$ and $y$ (possibly $y''$ is $y'$ or $y$).
  As no internal vertex of $P_2[y'',x'']$ is on $P_1$, concatenating $P_2[y'',x'']$ and $P_1[x'',y'']$ gives a simple cycle; note that $P_1[x'',x']$ and $P_1'[y',y'']$ may contain vertices of $P_2$ outside $P_2[y'',x'']$.
  The length of this cycle is at most $\mathsf{cf}(G)$, hence $|P_1[x',y']| \leq |P_1[x'',y'']| \leq \mathsf{cf}(G)-1$.
  It follows that for any $\mathsf{cf}(G)-1$ consecutive vertices of $P_1$, at least one of these vertices is used by~$P_2$.
  As the first and last vertices of $P_1$ (that is, $x$ and $y$) are in $P_2$, it is easy to see that if~$n_i$ denotes the number of vertices of $P_i$, then $n_1\leq (n_2 - 1)(\mathsf{cf}(G) - 1) + 1$.
  In other words, $|P_1|\leq |P_2|(\mathsf{cf}(G) - 1)$, what we had to show.
\end{proof}

Note that the ratio $\mathsf{cf}(G) - 1$ in \autoref{thm:bcl_distboundedcircum} is tight; see the blue and green paths in \autoref{fig:cf4} for an example.
\begin{figure}[htpb]
  \centering
  \includegraphics{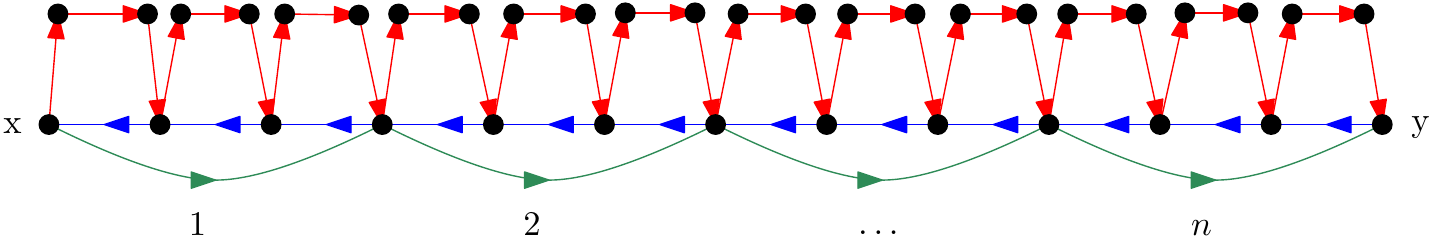}
  \caption{A strong digraph $G$ with circumference $\mathsf{cf}(G) = 4$.
    There is an $x\rightarrow y$-path of length~$n$ (green), an $x\rightarrow y$-path of length $(\mathsf{cf}(G)-1)^2n = 9n$ (red), and a $y\rightarrow x$-path of length $(\mathsf{cf}(G) - 1)n = 3n$ (blue).
    Examples for arbitrary $\mathsf{cf}(G) = \ell$ can be constructed in a similar way.}
  \label{fig:cf4}
\end{figure}

\restatedistboundedcircumsquared*
\begin{proof}
  If $x = y$, then every simple $x\rightarrow y$-path has length 0, and the statement holds trivially.
  Otherwise, let $P^\star$ be a shortest $ x\rightarrow y$-path.
  As $G$ is strong, every arc of~$P^\star$ is in a cycle of length at most $\mathsf{cf}(G)$.
  Thus, for every arc $(u,v)$ of $P^\star$, there is a $v\rightarrow u$-path of length at most $\mathsf{cf}(G) - 1$.
  Concatenating these paths for every arc of~$P^\star$, we obtain a $y\rightarrow x$ walk of length at most $(\mathsf{cf}(G) - 1)|P^\star| = (\mathsf{cf}(G) - 1)\mathsf{dist}_G(x,y)$, and hence there is a $y\rightarrow x$-path $P_2$ of at most this length.
  By \autoref{thm:bcl_distboundedcircum}, we have $|P_1|\leq (\mathsf{cf}(G) - 1)|P_2| \leq (\mathsf{cf}(G) - 1)^2|P^\star| = (\mathsf{cf}(G) - 1)^2\mathsf{dist}_G(x,y)$.
\end{proof}
Again, the ratio $(\mathsf{cf}(G) -1)^2$ in \autoref{thm:bcl_distboundedcircumsquared} is tight, see the red and green paths in \autoref{fig:cf4}.

\restatedistsecondpath*
\begin{proof}
  The claim is true if $x = y$, as then both $P_1$ and $P_2$ have length 0, and the statement holds trivially.
  Otherwise, as $G$ is strong, every arc of $P_2$ is in a cycle of length at most~$\mathsf{cf}(G)$.
  Thus, for every arc $(u,v)$ of $P$, there is a $v\rightarrow u$-path of length at most $\mathsf{cf}(G) - 1$.
  Concatenating these paths for every arc of $P_2$, we obtain a walk from $y$ to $x$ where every vertex is at distance at most $\mathsf{cf}(G) - 2$ from $P_2$.
  This implies that there is an $y\rightarrow x$-path $P_3$ where every vertex is at distance at most $\mathsf{cf}(G) - 2$ from $P_2$.
  
  Observe that if $(u,v)$ is an arc of $P_1$, then $\mathsf{dist}_G(P_2,v) \leq \mathsf{dist}_G(P_2,u) + 1$.
  Therefore, if~$P_1$ has a vertex $v$ with $\mathsf{dist}_G(P_2,v) > 2(\mathsf{cf}(G) - 2)$, then there is a subpath $P_1[v',v'']$ with $\mathsf{dist}_G(P_2,v') = \mathsf{cf}(G) - 2$, $\mathsf{dist}_G(P_2,v'') = 2\mathsf{cf}(G) - 3$, and every internal vertex of $P_1[v',v'']$ is at distance more than $\mathsf{cf}(G) - 2$ from~$P_2$.
  This means that $P_3$ does not contain any internal vertex of $P_1[v',v'']$, since every vertex of $P_3$ is at distance at most $\mathsf{cf}(G) - 2$ from~$P_2$.
  Now $P_1[v'',y]\circ P_3\circ P_1[x,v']$ is a $v''\rightarrow v'$ walk that does not contain any internal vertex of $P_1[v',v'']$ and hence there is a simple cycle containing $P[v',v'']$.
  Note that the length of any $v''\rightarrow v$-path is at least 2: $P_1$ has no arc $(v'',v')$ and such an arc cannot appear in $P_3$ either, as $\mathsf{dist}_G(P_2,v) > \mathsf{cf}(G) - 2$.
  Therefore, the length of this cycle is at least $|P[v',v'']| + 2 (\mathsf{cf}(G) - 1) + 2 > \mathsf{cf}(G)$, a contradiction.
  This proves $\mathsf{dist}_G(P_2,v) \leq 2(\mathsf{cf}(G) - 2)$.
  To prove the second bound $\mathsf{dist}_G(v,P_2) \leq 2(\mathsf{cf}(G) - 2)$, let us reverse the arcs of the graph and apply the first bound on the two $y\rightarrow x$-paths corresponding to $P_1$ and~$P_2$.
\end{proof}
The bound $2(\mathsf{cf}(G)-2)$ in \autoref{thm:bcl_distsecondpath} is tight: in \autoref{fig:cf4}, the red $x\rightarrow y$-path has vertices at distance exactly $2(\mathsf{cf}(G) - 2) = 4$ from the green path (the example can be generalized to larger~$\mathsf{cf}(G)$).

Next, we generalize \autoref{thm:bcl_distsecondpath} to the case when the start/end vertices of the two paths are not necessarily the same, but they are close to each other.

\restatepathdistance*
\begin{proof}
  Let $Q_x$ be a shortest $x_1\rightarrow x_2$-path (which has length $\mathsf{dist}_G(x_1,x_2)\leq t$ and let $Q_y$ be a shortest $y_2\rightarrow y_1$-path (which has length $\mathsf{dist}_G(y_2,y_1)\leq (\mathsf{cf}(G) - 1)\mathsf{dist}_G(y_1,y_2)\leq (\mathsf{cf}(G) - 1)t$ by \autoref{thm:bcl_distboundedcircum}).
  The concatenation $Q_x\circ P_2\circ Q_y$ is an $x_1\rightarrow y_1$ walk; let $R$ be an $x_1\rightarrow y_1$-path using a subset of arcs of this walk.
  By \autoref{thm:bcl_distsecondpath}, every vertex $v$ of $P_1$ is at distance at most $2(\mathsf{cf}(G) - 2)$ from $R$; let $u$ be a vertex of $R$ with $\mathsf{dist}_G(u,v)\leq 2(\mathsf{cf}(G) - 2)$.
  We consider three cases depending on the location of $u$:
  \begin{compactitem}
    \item $u\in P_2$: Then we are done, since vertex $v$ is at distance at most $2(\mathsf{cf}(G) - 1)$ from $P_2$.
    \item $u\in Q_x$: Then
      \begin{eqnarray*}
        \mathsf{dist}_G(x_2,v) & \leq & \mathsf{dist}_G(x_2,u) + \mathsf{dist}_G(u,v)\\
                               & \leq & (\mathsf{cf}(G) - 1) \mathsf{dist}_G(u,x_2) + \mathsf{dist}_G(u,v)\\
                               & \leq & (\mathsf{cf}(G) - 1)|Q_x| + 2(\mathsf{cf}(G) - 2)\\
                               & \leq & (\mathsf{cf}(G) - 1)t + 2(\mathsf{cf}(G) -2),
      \end{eqnarray*}
      by the triangle inequality and \autoref{thm:bcl_distboundedcircum}; hence, we are done.
    \item $u\in Q_y$: Then
      \begin{eqnarray*}
        \mathsf{dist}_G(y_2,v) & \leq & \mathsf{dist}_G(y_2,u)+\mathsf{dist}_G(u,v)\\
                               & \leq & |Q_y|+2(\mathsf{cf}(G)-2)\\
                               & \leq & (\mathsf{cf}(G)-1)t + 2(\mathsf{cf}(G)-1),
      \end{eqnarray*}
      and we are done again.\qedhere
  \end{compactitem}
\end{proof}

\subsection{Bypassing}
\label{sec:bypassing}
In this subsection, we prove a result exploiting that any two $x\rightarrow y$-paths in a strong digraph of bounded circumference are ``close'' to each other, hence (if the paths are sufficiently long), there are many disjoint paths connecting them.
Therefore, if we delete a set $S$ of at most $k$ vertices, then we can use these connecting paths to switch from one path to the other, avoiding the vertices of $S$.
We use this result in \autoref{sec:finding_important_cluster_separators}.

\restatebypassingthroughnearbypaths*
\begin{proof}
  As $|P_1[a, b]| \geq \cf{G}^5 \cdot (t + 2)k$, it is possible to select vertices $v_1, \hdots, v_{k+1}$ on $P_1[a, b]$ such that $|P_1[v_i, v_j]| \geq \cf{G}^3 (t + 2)$ for $i, j \in \lbrace 1, \hdots, k\rbrace, i < j$.
  By \autoref{thm:bcl_pathdistance} every $v_i$ is at distance at most $(\cf{G} - 1)t + 2(\cf{G} - 2)$ from $P_2$.
  Let $Q_i$ be a $P_2 \to v_i$-path for every $i \in \lbrace 1, \hdots, k+1\rbrace$.

  \begin{claim}
    The paths $Q_1, \hdots, Q_{k +1}$ are pairwise vertex-disjoint.
  \end{claim}
  \begin{claimproof}
    Suppose, for sake of contradiction, that two paths $Q_i$ and $Q_j$, $i < j$ intersect in a vertex~$q$.
    Using the triangle inequality and \autoref{thm:bcl_distboundedcircum}, we have
    \begin{eqnarray*}
      \mathsf{dist}_G(v_j,v_i) & \leq & \mathsf{dist}_G(v_j,q) + \mathsf{dist}_G(q,v_i)\\
                               & \leq & (\cf(G) - 1)\mathsf{dist}_G(q,v_j) + \mathsf{dist}_G(q,v_i)\\
                               & \leq & (\cf(G) - 1)|Q_{i_2}| + |Q_{i_1}|\\
                               & \leq & \cf(G) \cdot \left[(\cf(G) - 1)t + 2(\cf(G) - 2)\right]\\
                               &   <  & \cf(G)^2 (t + 2) \enspace .
    \end{eqnarray*}
    Another usage of \autoref{thm:bcl_distboundedcircum} yields then $|P[v_i, v_j]| \leq (\cf(G) - 1) \dist_G(v_j, v_i) < \cf(G)^3 (t + 2)$ --- a contradiction to the choice of $v_i$ and $v_j$ as $|P_1[v_i, v_j]| \geq \cf{G}^3 (t + 2)$.
	\end{claimproof}
	
  As the $Q_1, \hdots, Q_{k+1}$ are pairwise vertex-disjoint and $|S| \leq k$, there is a $Q_i$ that is disjoint from~$S$.
  Let $w_i$ be the first vertex of $Q_i$.
  Consider the $x_2 \to b$-walk $W = P_2[x_2, w_i] \circ Q_i \circ P_1[v_i, b]$.
  Note that it is only a walk (and not necessarily a path) as $Q_i$ may use arcs of $P_[v_i, b]$.
  As $v_i \in V(P_1[a,b])$ we know that $P_1[v_i, b]$ is disjoint from $S$.
  Also, every subpath of $P_2$ is disjoint from $S$ as well as~$Q_i$.
  Therefore, $W$ is a $x_2 \to b$-walk in $G - S$ and as such contains a $x_2 \to b$-path in $G - S$.
\end{proof}

\subsection{Representative Sets of Paths}
\label{sec:bcl_representativesetsofpaths}
Let $G$ be a digraph and let $x,y\in V(G)$.
We say that a set $\mathcal P$ of $x\rightarrow y$-paths is \emph{$k$-representative} if whenever $S\subseteq V(G)$ is a set of at most $k$ vertices such that $G - S$ has an $x\rightarrow y$-path, then~$\mathcal P$ contains an $x\rightarrow y$-path disjoint from $S$.
Representative sets of paths (and also of other objects) are important tools in the design of parameterized algorithms~\cite{Monien1985,FominEtAl2014,FominEtAl2014b,ShachnaiZehavi2016,Marx2009}.

The algorithm of Bonsma and Lokshtanov~\cite{BonsmaLokshtanov2011} for the case $\mathsf{cf}(G)\leq 2$ uses the following observation in an essential way.
Let $G$ be a strong digraph and let $\langle G\rangle$ denote the underlying undirected graph of $G$.
If $\mathsf{cf}(G)\leq 2$ then $\langle G\rangle$ is a tree (with bidirected arcs in $G$), and hence there is a unique $x\rightarrow y$-path $P$ for any pair $x,y$ of distinct vertices in $G$.
This means that for any set $S\subseteq V(G)$, either $P$ is an $x\rightarrow y$-path in $G - S$ or there is no $x\rightarrow y$-path in $G - S$ at all.
In other words, the set $\{P\}$ is a $k$-representative family for every~$k$.

The situation is significantly different even for $\mathsf{cf}(G) = 3$.
Consider the strong digraph in \autoref{fig:cf3}.
\begin{figure}[htpb]
  \centering
  \includegraphics{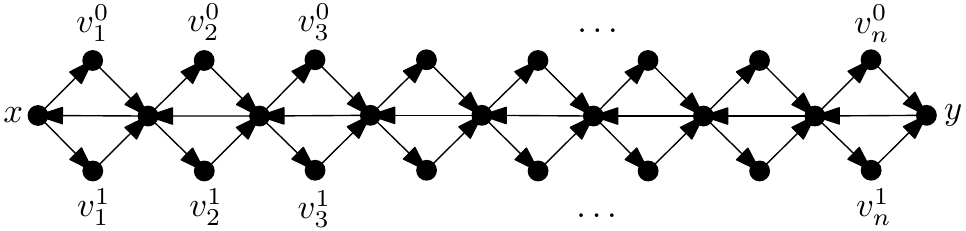}
  \caption{A digraph $G$ with $\mathsf{cf}(G) = 3$ where every $k$-representative set of $x\rightarrow y$-paths has size $2^{\Omega(k)}\log n$.}
  \label{fig:cf3}
\end{figure}
There are exactly $2^n$ different $x\rightarrow y$-paths in $G$; each such path corresponds to a $0$-$1$ vector of length $n$.
Thus, if we remove vertex $v_i^0$ (resp., $v_i^1$), then only those paths survive that have 1 (resp.,~0) at the $i$-th coordinate.
Therefore, a collection of paths in this graph is $k$-representative only if no matter how we fix the values of $k$ arbitrary coordinates, there is a vector in the collection satisfying these constraints.
Kleitman and Spencer~\cite{KleitmanSpencer1973} proved that every collection of vectors of length $n$ satisfying this property has size $2^{\Omega(k)}\cdot\log n$ (more precisely, they gave a lower bound on the dual question of $k$-independent families, but it can be easily rephrased into this lower bound).
The main result of this subsection is that in a digraph of bounded circumference, we can construct a $k$-representative family of paths whose size is somewhat worse than this lower bound: assuming that the circumference is bounded by a constant, there is such a family of size $2^{\mathcal{O}(k^2\log k)}\cdot\log n$ (\autoref{thm:bcl_computerepresenetativeforboundedcf}).

If the paths we are considering have bounded length, then the results of Monien~\cite{Monien1985} give a representative set of bounded size:
\begin{proposition}[\cite{Monien1985}]
\label{thm:bcl_smallrepresentativeset}
  Let $G$ be a digraph, let $x,y\in V(G)$ and let $k\in\mathbb N$.
  If every $x\rightarrow y$-path in~$G$ has length at most $\ell$, then a $k$-representative set containing at most $\ell^k$ many $x\rightarrow y$-paths can be found in time $\ell^{\mathcal{O}(k)}\cdot n^{\mathcal{O}(1)}$.
\end{proposition}

Recently, Fomin et al.~\cite{FominEtAl2014} improved the computation of representative sets of paths, both in terms of the size of the set and the run time, but \autoref{thm:bcl_smallrepresentativeset} will be sufficient for our purposes.

We will show that in strong digraphs of bounded circumference, a $k$-representative set of bounded size can be found even if there is no bound on the length of the $x\rightarrow y$-paths.
The proof uses a certain family of hash functions.
Let $\mathcal F$ be a family of functions $f:U\rightarrow\{1,\hdots,k\}$ on the universe~$U$.
We say that $\mathcal F$ is a $k$-perfect family of hash functions if for every $X\subseteq U$ of size at most $k$, there is an $f\in\mathcal F$ that is injective on~$X$, that is, $f(x)\not=f(x')$ for any two distinct $x,x'\in X$.
Alon et al.~\cite{AlonEtAl1995} showed that a $k$-perfect family $\mathcal F$ of size $2^{\mathcal{O}(k)}\log|U|$ exists, and can be constructed in time $2^{\mathcal{O}(k)}|U|^{\mathcal{O}(1)}$.

Before presenting the construction of representative sets for strong digraphs of bounded circumference, let us explain how $k$-perfect families of hash functions can be used for the construction in the case of the graph of \autoref{fig:cf3}.
Let $\mathcal F$ be a $k$-perfect family of hash functions over the universe $U = \{1,\hdots,n\}$.
For every $f\in \mathcal F$ and every function $h:\{1,\hdots,k\}\rightarrow\{0,1\}$, we add to the set the path that used vertex $v_i^{h(f(i))}$ for every $i\in\{1,\hdots,n\}$.
Let $S$ be a set of vertices that contains, for some $X\subseteq U$ of size $k$ and function $g:X\rightarrow\{0,1\}$, the vertices~$v_i^{g(i)}$.
As $\mathcal F$ is a $k$-perfect family, there is an $f\in\mathcal F$ that is injective on $X$.
For every $i\in X$, let us define $h(f(i)) = 1 - g(i)$; as~$f$ is injective on $X$, this is well-defined and gives a function $h:\{1,\hdots,k\}\rightarrow\{0,1\}$.
We claim that the path $P$ introduced for this choice of $f$ and $h$ is disjoint from $S$.
For $i\notin X$, it does not matter if~$P$ uses $v_i^0$ or $v_i^1$.
For $i\in X$, set~$S$ contains~$v_i^{f(i)}$.
By our definition of $h$, we have $h(f(i)) = 1 - g(i)$, hence $P$ uses $v_i^{1-g(i)}$, avoiding $S$.
Thus $P$ is indeed disjoint from $S$.

The following proof generalizes this construction to arbitrary strong digraphs of bounded circumference: we construct the path by concatenating a series of fairly independent ``short jumps.''
The short jumps are taken from a representative set of short paths; the choice of which short path to select is determined by $k$-perfect hash function, similarly to the argument in the previous paragraph.

\restatecomputerepresentativesetofpathsforboundedcircum*
\begin{proof}
  Let us fix an arbitrary $x\rightarrow y$-path $R$ (which exists as $G$ is strong) to guide our construction.
  Denote by $r$ the length of $R$ and by $v_0 = x, v_1,\hdots,v_{r-1}, v_r=y$ its vertices.
  We only consider a subset of vertices $z_i$ at distance $d = 2\mathsf{cf}(G)^4$ from each other or more formally $z_i = v_{i\cdot d}$.
  These~$z_i$ will be the anchor vertices for our short jumps.
  We divide then~$z_i$ further into $k+1$ subsets~$Z^o$ by taking every $(k+1)$st vertex starting at offset $o$.
  Formally we define $z_i^o = z_{i(k+1) + o}$ and $Z^o = \{z^o_i\}$.
  These subsets have the advantage that one of these is far away from a deletion set $S$ of size at most~$k$.
  For this we fix a set $S$ of size at most $k$ such that a $x \rightarrow y$-path in $G - S$ exists.
  
  \begin{claim}
  \label{thm:bcl_shortjumps_claim1}
    There is some $o_S \in\{0,\hdots,k\}$ such that
    \begin{compactitem}
    \item $\mathsf{dist}_G(Z^{o_S},S) > 2(\mathsf{cf}(G)-2)$ and
    \item $\mathsf{dist}_G(S,Z^{o_S}) > 2(\mathsf{cf}(G) - 2)$.
    \end{compactitem}
  \end{claim}
  \begin{claimproof}
    We claim that for every $s\in S$ there is at most one value $o\in\{0,\hdots,k\}$ such that $\mathsf{dist}_G(Z^o,s)\leq 2\mathsf{cf}(G)^2$.
    Suppose that $\mathsf{dist}_G(w_1,s),\mathsf{dist}(w_2,s)\leq 2\mathsf{cf}(G)^2$ for some $w_1\in Z^{o_1}$ and $w_2\in Z^{o_2}$ with $o_1\not=o_2$.
    Assume, without loss of generality, that $w_1$ is before~$w_2$ on $R$; then $R[w_1,w_2]$ has length at least~$d$ (as different $z_i$ have distance at least~$d$).
    By \autoref{thm:bcl_distboundedcircum}, we have
    \begin{equation*}
      \mathsf{dist}_G(s,w_1)\leq (\mathsf{cf}(G)-1)\mathsf{dist}_G(w_1,s)\leq (\mathsf{cf}(G)-1)\cdot 2\mathsf{cf}(G)^2,
    \end{equation*}
    thus $\mathsf{dist}_G(w_2,w_1)\leq \mathsf{dist}_G(w_2,s) + \mathsf{dist}_G(s,w_1)\leq 2\mathsf{cf}(G)^3$.
    Again by \autoref{thm:bcl_distboundedcircum}, we have\linebreak $d\leq |R[w_1,w_2]|\leq (\mathsf{cf}(G)-1)\mathsf{dist}_G(w_2,w_1) < 2\mathsf{cf}(G)^4$, a contradiction.
    Thus, we have proved that for each of the $k$ vertices $s\in S$ there is at most one value $o\in\{0,\hdots,k\}$ such that~$s$ is at distance at most $2\mathsf{cf}(G)^2$ from~$Z^o$.
    Therefore, by the pigeon-hole principle there is an $o_S\in\{0,\hdots,k\}$ such that $\mathsf{dist}_G(Z^{o_S},S) > 2\mathsf{cf}(G)^2$.
    By \autoref{thm:bcl_distboundedcircum} this also implies $\mathsf{dist}_G(S,Z^{o_S}) > 2\mathsf{cf}(G)^2/(\mathsf{cf}(G)-1) > 2(\mathsf{cf}(G)-2)$.
    This completes the proof of \autoref{thm:bcl_shortjumps_claim1}.
  \end{claimproof}

  Thus we know that a small surrounding of one of the $Z^o$'s will be disjoint from $S$.
  Furthermore, \autoref{thm:bcl_distboundedcircumsquared} gives a bound on the length of a path $P$ between two consecutive vertices $z^o_i$ and~$z^o_{i+1}$ of $Z^o$, by $|P| \leq (\mathsf{cf}(G) - 1)^2 |R[z^o_i,z^o_{i+1}]| = \mathcal{O}(\mathsf{cf}(G)^7k)$.
  This allows us to introduce sets~$\mathcal{P}^o_i$ of $k$-representative $z^o_i \rightarrow z^o_{i+1}$-paths using the algorithm of \autoref{thm:bcl_smallrepresentativeset} and have their size bounded by some $B = \mathcal{O}(\mathsf{cf}(G)^7k)^k = \mathsf{cf}(G)^{\mathcal{O}(k\log k)}$ (using $k = 2^{\log k}$ and $\mathsf{cf}(G) \geq 2$).
  By duplicating paths as necessary we can assume that every $P^o_i$ has size exactly~$B$.
  
  To make sure that our path collections with offset are connected to $x$ and $y$ we construct additional sets $\mathcal{P}^o_x$ and $\mathcal{P}^o_y$ as follows:
  Let $z^o_x$ be the first vertex in $Z^o$ \emph{after} $x$ and $z^o_y$ the last vertex \emph{before} $y$.
  Then compute, using the algorithm of \autoref{thm:bcl_smallrepresentativeset},  $\mathcal{P}^o_x$ as a $k$-representative set of $x \rightarrow z^o_x$-paths and $\mathcal{P}^o_y$ as a $k$-representative set of $z^o_y \rightarrow y$-paths.
  As the distances between these pairs of vertices are bounded by the distance of neighboring vertices in $Z^o$ we can analogously get a size bound of $B$ for $\mathcal{P}^o_x$ and $\mathcal{P}^o_y$.
  Note that for some offsets $o$ either $\mathcal{P}^o_x$ or $\mathcal{P}^o_y$ may align with some~$\mathcal{P}^o_i$; then we leave out this $\mathcal{P}^o_i$ as we do not need it anymore.
  We define the set of these relevant sets as $\mathcal{P}^o := \{\mathcal{P}^o_x$, $\mathcal{P}^o_y\} \cup \{\mathcal{P}^o_i\}_i$ for each $o$.
  
  \begin{claim}
    \label{thm:bcl_shortjumps_claim2}
   	Every $\mathcal{P}^{o_S}_T \in \mathcal{P}^{o_S}$ contains a path disjoint from $S$.
  \end{claim}
  \begin{claimproof}
	Consider a set $\mathcal{P}^{o_S}_T$ with $T \in \{x, y, i\}$ such that the paths in $\mathcal{P}^{o_S}_T$ are $x_T \rightarrow y_T$-paths.
	As above sets are $k$-representative sets of paths, we must only show that there is any $x_T \rightarrow y_T$-path in $G - S$.
	By assumption there is a $x \rightarrow y$~path $Q$ in $G -S$.
	By \autoref{thm:bcl_distsecondpath} we can find a $q_x \in V(Q)$ such that $\mathsf{dist}(x_T, q_x) \leq 2(\mathsf{cf}(G) - 2)$ and a $x_T \rightarrow q_x$-path $Q_x$ in~$G$ achieving this distance.
	By \autoref{thm:bcl_shortjumps_claim1} we know that $Q_x$ is disjoint from~$S$ and therefore, $Q_x \circ Q[q_x,y]$ is a $q_x \rightarrow y$ walk disjoint from $S$.
	Let $\hat{Q}_x$ be a $q_x \rightarrow y$-path contained in this walk.
	Another application of \autoref{thm:bcl_distsecondpath} yields a vertex $q_y \in V(\hat{Q})$ with $\mathsf{dist}(q_y,y_T) \leq 2(\mathsf{cf}(G) - 2)$ and a $q_y \rightarrow y_T$-path~$Q_y$ in $G$ achieving this distance.
	Again, by \autoref{thm:bcl_shortjumps_claim1}, $Q_y$ is disjoint from~$S$.
	Then $\hat{Q}_x[x_T, q_y] \circ Q_y$ contains a $x_T \rightarrow y_T$-path as proposed.
	This completes the proof of \autoref{thm:bcl_shortjumps_claim2}.
  \end{claimproof}
	
  Of course, enumerating all possible tuples of paths would construct to many candidates, as the size of~$\mathcal{P}^{o_S}$ can be $\Omega(m)$.
  Therefore, we want to use a $f(k)$-perfect family of hash functions.
  This is possible if we can bound the number of intersections with the sets $\mathcal{P}^{o_S}$ by~$f(k)$.
  \begin{claim}
	\label{thm:bcl_shortjumps_claim3}
	The set $S$ intersects at most $2k$ sets of $\mathcal{P}^{o_S}$.
  \end{claim}
  \begin{claimproof}
    We show that $s\in S$ can intersect for at most two sets that share an endpoint, thus achieving the claimed size bound.
   	Suppose, for sake of contradiction, that $s$ intersects two paths~$Q_1$ and $Q_2$ out of sets in $\mathcal{P}^{o_S}$ that do not share an endpoint.
   	Let each $Q_i$ be an $x_i \rightarrow y_i$-path.
   	Assume, without loss of generality, that the order in which the endpoints appear on $R$ is $x_1, y_1, x_2, y_2$, and that $|R[y_1,x_2] \geq 2\textsf{cf}(G)^5$ (by the distance of the $z_i$.
   	At the same time, $R[x_i, y_i]$ and $Q_i$ connect the same endpoints, hence \autoref{thm:bcl_distsecondpath} implies that there is a $t_1 \in V(R[x_1, y_1])$ with $\mathsf{dist}(t_1,s) \leq 2(\mathsf{cf}(G) - 2)$ and a $t_2 \in V(R[x_2, y_2])$ with $\mathsf{dist}(s, t_2) \leq 2(\mathsf{cf}(G) - 2)$ as $s \in Q_1 \cap Q_2$.
   	This implies that $\mathsf{dist}(t_1, t_2) \leq \mathsf{dist}(t_1,s) + \mathsf{dist}(s, t_2) \leq 4(\mathsf{cf}(G) - 2)$.
	If we now consider $R[t_1,t_2]$, we get $|R[t_1,t_2]| \geq |R[y_1,x_2] \geq 2\mathsf{cf}(G)^5  > (\mathsf{cf}(G) - 1)^2 \cdot \mathsf{dist}(t_1, t_2)$ in contradiction to \autoref{thm:bcl_distboundedcircumsquared}.
	This completes the proof of \autoref{thm:bcl_shortjumps_claim3}.
  \end{claimproof}

  We can now construct a $2k$-perfect family $\Psi^o$ of hash functions over the universe $\mathcal{P}^o$ for each~$o$.
  For $o_S$ this family contains an element $\psi$ which gives all set of $\mathcal{P}^{o_S}$ which are intersected by $S$ a different number in $\{1, \hdots, 2k\}$ (by \autoref{thm:bcl_shortjumps_claim3}).
  Further, there is a map~$\pi_\textsf{free}$ that maps the numbers of $\{1,\hdots,2k\}$ to a number of $\{1,\hdots,B\}$, such that for every $\mathcal{P} \in \mathcal{P}^{o_S}$ which has a path intersected by $S$, we have that the $\psi \circ \pi_\textsf{free}(\mathcal{P})$th path of $\mathcal{P}$ is not intersected by $S$.
  There is such a path by \autoref{thm:bcl_shortjumps_claim2}.
  Denote by $Q_{\psi,\pi_\textsf{free}}(\mathcal{P})$ this path.
  As we cannot know~$\pi_\textsf{free}$ in advance we create a set $\Pi$ of all possible functions from $\{1,\hdots,2k\}$ to~$\{1,\hdots,B\}$.
	
  We know that for the specific choices of $o_S$, $\psi$ and $\pi_\textsf{free}$ we get a that the union of paths in $\{Q_{\psi,\pi_\textsf{free}}(\mathcal{P}) | \mathcal{P} \in \mathcal{P}^{o_S}\}$ forms a $x \rightarrow y$ walk $W$ in $G - S$.
  Every $x \rightarrow y$-path within $W$ is also disjoint from $S$.
  Therefore, the set~$\mathcal P_{x,y,k}$ created as follows contains a path disjoint from~$S$:
  For every $o \in \{1,\hdots,k+1\}$, every $\psi \in \Psi$ and every $\pi \in \Pi$ consider the $x \rightarrow y$-walk $\Cup_{\mathcal{P} \in \mathcal{P}^{o}} Q_{\psi,\pi}(\mathcal{P})$ and introduce an arbitrary $x \rightarrow y$-path in it into $\mathcal P_{x,y,k}$.
	
  The size bound on $\mathcal P_{x,y,k}$ is proven by multiplying the possibilities for each choice:
  \begin{equation*}
    \underbrace{(k+1)}_{\textnormal{choice of $o$}} \cdot \underbrace{2^{\mathcal{O}(k)}\log m}_{|\Psi|}\cdot \underbrace{B^{2k}}_{|\Pi|} = \mathsf{cf}(G)^{\mathcal{O}(k^2\log k)}\log n.
  \end{equation*}
  The run time follows similarly.
\end{proof}


The previous lemma is very useful if we have a strong digraph of bounded circumference.
However, if we have a graph $G$ and a subset of vertices $T$ such that $\mathsf{cf}(G - T) \leq \ell$ it is not clear how to get a $k$-representative set of paths.
Instead we give a much weaker result which suffices for our algorithm.
We restrict our deletion sets from arbitrary sets $S$ of size at most $k$ to sets which additionally fulfill $\mathsf{cf}(G - S) \leq \ell$.
Additionally, instead of walks we consider closed walks connecting at least two vertices of $T$ after the deletion of $S$.
\begin{lemma}
\label{lem:prepostpath}
  Let $G$ be a digraph, let $s,t\in V(G)$, and let $k,d\in\mathbb Z_{\geq 0}$.
  In time $2^{\mathcal O(kd)}\cdot n^{\mathcal O(1)}$ we can construct collections $\mathcal{R}_{\leq 2d}$ and $\mathcal{R}_{> 2d}$, each of size $2^{\mathcal O(kd)}$, such that
  \begin{compactitem}
    \item $\mathcal{R}_{\leq 2d}$ contains only paths of length at most $2d$,
    \item $\mathcal{R}_{> 2d}$ contains pairs $(P_s,P_t)$, where each $P_s$ is a path of length $d$ starting at $s$, each $P_t$ is a path of length $d$ ending at $t$.
  \end{compactitem}	  
  Then, for every set $S$ of size at most $k$, if there is an $s \to t$-path $P$ disjoint from $S$ then there is a path $P'$ disjoint from $S$ with
  \begin{compactitem}
    \item  $P' \in  \mathcal{R}_{\leq 2d}$ if $|P| \leq 2d$ or
    \item $(P_s', P_t') \in \mathcal{R}_{> 2d}$, where $P'_s$ and $P'_t$ are the disjoint subpaths of $P'$ containing the first and the last~$d$ arcs of~$P'$ respectively, if $|P| > 2d$.
  \end{compactitem}	
\end{lemma}
\begin{proof}
  The proof is by induction on $d$.
  For $d=0$, the construction is easy:
  If $s = t$ introduce the zero length path $\lbrace s \rbrace$ into $ \mathcal{R}_{\leq 2d}$.
  Otherwise set $\mathcal{R}_{\leq 2d} = \emptyset$.
  In any case let~$\mathcal{R}_{> 2d}$ contain a single pair $(P_s,P_t)$ with $P_s$ and $P_t$ being zero length paths containing only vertices $s$ and~$t$ respectively.
  This construction is correct since every $s \to t$-path starts in $s$ and ends in $t$ and a zero length path does only exist (and is unique) if $s = t$.

  Suppose that $d>0$, and that the statement of the lemma holds for $d - 1$.
  We start with  $\mathcal{R}_{\leq 2d} = \mathcal{R}_{> 2d} = \emptyset$.
  If $s=t$ or $s$ and $t$ are adjacent, then we introduce to $\mathcal{R}_{\leq 2d}$ the path of length~$0$ or~$1$, respectively.
  Afterwards, let us invoke the algorithm of \autoref{thm:bcl_smallpathwitness} on $X=N^+(s)$ and $Y=N^-(t)$.
  For every pair $(s',t')$ with $s'\in X'$ and $t'\in Y'$, let us use the induction hypothesis and invoke our algorithm on the digraph $G-\{s,t\}$, vertices $s',t'$, and integers $k$ and $d-1$ to enumerate the collections~$\mathcal{R}'_{\leq 2(d-1)}$ and $\mathcal{R}'_{> 2(d-1)}$.
  We add to~$\mathcal{R}_ {\leq 2d}$ all paths obtained by extending a path $P\in \mathcal{R}'_{\leq 2(d-1)}$ into $sPt$.
  Moreover, we add to~$\mathcal{R}_{> 2d}$ all pairs obtained by extending a pair $(P'_s,P'_t)\in\mathcal{R}'_{> 2(d-1)}$ to $(sP'_s,P'_tt)$.

  To prove that the resulting collections $\mathcal{R}_{\leq 2d}$ and $\mathcal{R}_{> 2d}$ satisfy the requirements, consider a path~$P$ disjoint from an arbitrary set $S \subset V(G)$ of size at most $k$.
  If $P$ has length $0$ or $1$, we introduced this path into $\mathcal{R}_{\leq 2d}$ and are done.
  Otherwise, let $s'$ and~$t'$ be the neighbors of $s$ and $t$ on~$P$, respectively.
  As $|P| \geq 2$, $P[s', t']$ is a subpath of $P$ and, therefore disjoint from~$S$.
  Moreover, it is disjoint from $s$ and $t$ by definition.
  By choice of $X'$ and $Y'$ (see \autoref{thm:bcl_smallpathwitness}), there is an $x \to y$-path~$Q$ in $G-(S\cup \{s,t\})$ with $x\in X'$ and $<\in Y'$.
  Therefore, by the induction hypothesis, there is a path $Q'$ in $G-(S\cup \{s,t\})$ such that either $Q'\in \mathcal{R}'_{\leq 2(d-1)}$ or  $(Q'_s,Q'_t)\in \mathcal{R}'_{> 2(d-1)}$, where $Q'_s$ and~$Q'_t$ contain the first and last $d-1$ arcs of $Q'$, respectively.
  In the first case, $sQ't$ is an $s\to t$-path in $G-S$ which we introduced to $\mathcal{R}_{\leq 2d}$.
  In the second case, $(sQ'_s, Q'_tt)$ will appear in $\mathcal{R}_{> 2d}$ and satisfy the requirements.

  By induction, $\mathcal{R}_1$ and $\mathcal{R}_2$ have size $2^{\mathcal O(kd)}$.
  The time for their construction is $2^{\mathcal O(kd)}\cdot n^{\mathcal O(1)}$.
\end{proof}

\begin{lemma}
\label{lem:stongplus2}
  Let $G$ be a digraph with two vertices $s,t$ and $k$ be an integer and suppose that $\cf(G-\{s,t\})\le \ell$.
  Then in time $2^{\mathcal O(k\ell + k^2\log k)}\cdot n^{\mathcal O(1)}$, we can compute a collection $\mathcal{Q}$ of $2^{\mathcal O(k\ell + k^2\log k)}\log^2 n$ closed walks in $G$, each containing both $s$ and $t$, such that the following holds: if $S\subseteq V(G)$ is a set of at most $k$ vertices in $G$ such that $\cf(G-S)\le\ell$ and $G-S$ has a closed walk containing both $s$ and $t$, then there is a closed walk in $\mathcal{Q}$ disjoint from $S$.
\end{lemma}
\begin{proof}
  We first compute a collection $\mathcal{P}_{s,t}$ of $s\to t$-paths.
  Let us use the algorithm of \autoref{lem:prepostpath} with digraph $G$ and $d=\ell$ to compute the collection $\mathcal{R}_{\leq 2 \ell}$ and  $\mathcal{R}_{> 2\ell}$.
  Let us introduce every path in $\mathcal{R}_{\leq 2\ell}$ into $\mathcal{P}_{s,t}$.
  We will introduce further paths into $\mathcal{P}_{s,t}$ based on~$\mathcal{R}_{> 2 \ell}$ the following way.
  For every $(P_s,P_t)\in \mathcal{R}_{> 2\ell}$, $x\in V(P_s)$, and $y\in V(P_t)$, if $x$ and $y$ are in the same strong component $C$ of $G-\{s,t\}$, then let us invoke the algorithm of \autoref{thm:bcl_computerepresenetativeforboundedcf} to obtain a collection~$\mathcal{P}_{x,y,k}$.
  Then for each $Z\in\mathcal{P}_{x,y,k}$, we extend $P$ to an $s\to t$-walk $Z^*$ by adding the prefix of~$P_s$ ending at $x$ and the suffix of $P_t$ starting at $y$, and we introduce $Z^*$ (or an $s\to t$-path using only the vertex set of $Z^*$) into $\mathcal{P}_{s,t}$.
  Observe that the size of $\mathcal{P}_{s,t}$ can be bounded by $2^{\mathcal O(k\ell + k^2\log k)}\log n$.

  We repeat a similar construction step with the roles of $s$ and $t$ reversed, to obtain a collection~$\mathcal{P}_{t,s}$ of $t\to s$-paths.
  Then for every choice of $P_{s,t}\in \mathcal{P}_{s,t}$ and $P_{t,s}\in\mathcal{P}_{t,s}$, we introduce the concatenation of
  $P_{s,t}$ and $P_{t,s}$ into $\mathcal{Q}$.
  Clearly, every member of $\mathcal{Q}$ is a closed walk containing both $s$ and $t$, and the size of $\mathcal{Q}$ is $2^{\mathcal O(k\ell + k^2\log k)}\log^2 n$.

  To prove the correctness of the construction, suppose that $S$ is a set of at most $k$ vertices such that $\cf(G-S) \le\ell$ and $G-S$
  has a closed walk containing both $s$ and $t$.
  This means that there is an $s\to t$-path $P_{s,t}$ and a $t\to s$-path $P_{t,s}$, both disjoint from $S$.
  We claim that both~$\mathcal{P}_{s,t}$ and~$\mathcal{P}_{t,s}$ contain paths disjoint from $S$.
  If this is true, then it follows by construction that~$\mathcal{Q}$ contains a closed walk disjoint from $S$.

  Let us prove that $\mathcal{P}_{s,t}$ contains a path disjoint from $S$ (the statement for $\mathcal{P}_{t,s}$ follows symmetrically).
  Assume, as a first case, that $P_{s,t}$ has length $\leq 2\ell$.
  Then by \autoref{lem:prepostpath}, we have that~$\mathcal{R}_{\leq 2\ell}$ contains an $s \to t$-path disjoint of $S$.
  This path also appears in $\mathcal{P}_{s,t} \supseteq  \mathcal{R}_{\leq 2\ell}$.
  
  Suppose now that $P_{s,t}$ has length $>2\ell$.
  Then there is an $s \to t$-path $Q$ disjoint of $S$ with $(Q_s, Q_t) \in \mathcal{R}_{> 2\ell}$ being the subpaths of its first and last $\ell$ arcs respectively.
  Let $x$ be the last vertex of $Q_s$ and $y$ be the first vertex of $Q_t$.
  Then $Q[x,y]$ is a certificate that there is an $x \to y$-path in $(G - \lbrace s,t \rbrace) - S$.
  
  We argue that $x$ and $y$ are in the same strong component of $G - \lbrace s,t \rbrace$.
  Consider the path~$P_{t,s}$.
  As both $Q$ and $P_{t,s}$ exist in $G - S$ the closed walk $W$ they form must contain no cycle of length greater than $\ell$.
  Hence, the path $Q_s$ must be intersected by $P_{t,s}$ outside of~$s$, as otherwise the cycle in $W$ containing the segment $Q_s$ has length greater than $\ell$.
  Let $x'$ be the last vertex of $P_{t,s} - s$ that intersects $Q_s$.
  By the same argument, $Q_t$ must be intersected by $P_{t,s} -t$.
  Let $y'$ be the first vertex of $P_{t,s} -t$ that intersects $Q_t$.
  Then $Q[x', y'] \circ P_{t,s}[y', x']$ is a closed walk in $G - \lbrace s,t \rbrace$ containing $x$ and $y$.
  Thus, $x$ and $y$ are in the same strong component of $G - \lbrace s,t \rbrace$.
  
  Then, by choice of $\mathcal{P}_{x,y,k}$, there is a path $Z\in \mathcal{P}_{x,y,k}$ that is disjoint from $S$ and we have extended~$Z$ to~$Z^*$ by adding $P_s$ and $P_t$ to it and then introduced it into $\mathcal{P}_{s,t}$.
  As $Z$, $P_s$, and $P_t$ are all disjoint from~$S$, it follows that~$\mathcal{P}_{s,t}$ contains a path disjoint from $S$.
\end{proof}

\restatestongplusW*
\begin{proof}
  We construct $\mathcal{Q}$ the following way.
  For every pair $s,t$ of vertices in $W$, we invoke the algorithm on \autoref{lem:stongplus2} in $G-(W\setminus\{s,t\})$ and vertices $s,t$.
  The collection $\mathcal{Q}$ will be the union of the $\binom{|W|}{2}$ collections obtained this way.

  To prove the correctness, suppose that $G-S$ has a strong component $C$ containing at least two vertices of $W$.
  This means that there is a closed walk $R$ containing at least two vertices of~$W$; let us choose $R$ such that $|W\cap V(R)|$ is minimum possible (but at least two), and subject to that,~$R$ is of minimum length.
  If $R$ contains exactly two vertices $s,t$ of $W$, then \autoref{lem:stongplus2} guarantees that a member of $\mathcal{Q}$ is disjoint from $S$.
  Suppose that $R$ contains a set $W_0$ of at least three vertices of $W$.
  If $R$ is a simple cycle, then it has length at most $\ell$ (as $\cf(G-S)\le \ell$) and we are done.
  Otherwise, there is a vertex $v\in V(R)$ that is visited at least twice during the walk, meaning that the walk can be split into two closed walks $R_1$ and $R_2$, meeting at $x$ (this is true even if $R$ visits $x$ more than twice).
  As $|W|\ge 3$, we can assume without loss of generality that $R_1$ visits at least two vertices of $W$. Note that $R_1$ cannot visit all vertices of $W$, as this would contradict the minimal choice of $R$.
  This means that $R_1$ visits at least two vertices of~$W$, but strictly fewer than~$C$, contradicting the minimal choice of $R$.
\end{proof}

\section{Reductions for Directed Long Cycle Vertex Deletion}
\label{sec:boundedcyclelength}
In this section we deal with reductions arising in the context of {\sc Directed Long Cycle Hitting Set}.
At first we will show that the arc deletion version can be reduced to an instance of the vertex deletion version of the same size.
This is simply done by taking the directed line graph.

\begin{theorem}
\label{thm:bcl_reductionarctovertex}
  There exists a polynomial parameter transformation from instances $(G,k,\ell)$ of the arc-deletion variant of {\sc Directed Long Cycle Hitting Set} to an instance $(G', k, \ell)$ of the vertex-deletion variant of {\sc Directed Long Cycle Hitting Set}.
\end{theorem}
\begin{proof}
  Given an instance $(G,k,\ell)$ of the arc-deletion variant of {\sc Directed Long Cycle Hitting Set}, create a digraph~$G'$ from $G$ by letting $G'$ be the directed line graph of $G$.
  In other words, $V(G') = A(G)$ and $A(G) = \{((v_1,v_2),(v_3,v_4))~|~(v_1,v_2),(v_3,v_3)\in A(G),v_2 = v_3\}$.
  We further set $k' = k$ and $\ell' = \ell$; then $(G',k',\ell')$ is an instance of the vertex-deletion variant of {\sc Directed Long Cycle Hitting Set}.

  In the forward direction, let $S$ be a set of at most $k$ arcs of $G$ such that in $G - S$ every simple cycle has length at most $\ell$.
  Then, since every cycle of $G$ is mapped to a cycle of $G'$ with the same length, the set $S'$ of vertices to which the arcs in $S$ get mapped to in $G'$ is such that $G' - S'$ does not have any simple cycles of length strictly more than $\ell$.

  In the backward direction, let $S'$ be a set of at most $k$ vertices of $G'$ such that every simple cycle of $G' - S'$ has length at most $\ell$.
  Then, since every cycle of $G'$ is mapped to a cycle of $G$ with the same length, the set $S$ of arcs to which the vertices in $S'$ get mapped to in $G$ is such that $G - S$ does not have any simple cycles of length strictly more than $\ell$.
\end{proof}
An alternative reduction works as follows: We create a digraph $G'$ by subdividing each arc $a\in A(G)$ of the original digraph $G$ by introducing a new vertex $v_a$ and then blowing up each original vertex $v\in V(G)$ into a biclique with $(|V(G)^2|+1$ vertices in each of its two partite classes.
Then for any $\ell\geq 3$ each directed cycle in $G$ of length $\ell$ becomes a cycle of length~$3\ell$ in $G'$, and vice versa.
Further, any arc-hitting set of the original digraph for cycles of length $\ell$ or more, gives a vertex hitting set of the new graph for cycles of length $3\ell$ or more, and vice versa.
This reduction was suggested to us by an anonymous reviewer of an earlier version of this paper.

\medskip
There is also a reduction in the other direction by splitting each vertex $v$ into $v^-$ and~$v^+$ and adding an arc from $v^-$ to $v^+$.
Arcs $(u,v)$ are replaced by $(u^+, v^-)$ and are made undeletable by taking enough copies.
\begin{theorem}
  There exists a polynomial parameter transformation from instances $(G,k,\ell)$ of the vertex-deletion variant of {\sc Directed Long Cycle Hitting Set} to an instance $(G', k,2 \ell)$ of the arc-deletion variant of {\sc Directed Long Cycle Hitting Set}.
\end{theorem}
\begin{proof}
  Given an instance $(G,k,\ell)$ of the vertex-deletion variant of \DLCHS{}, create a digraph $G'$ from $G$ by splitting each vertex $v\in V(G)$ into two vertices $v^+,v^-$, adding an arc from $v^-$ to $v^+$, and connecting all in-neighbours $u$ of $v$ in $G$ by $k + 1$ parallel arcs from~$u^+$ to~$v^-$ in $G'$, and all out-neighbours $u$ of $v$ in $G$ by an arc from $v^+$ to $u^-$ in $G'$.
  In other words, $V(G') = \{v^+,v^-~|~v\in V(G)\}$ and $A(G') = \{(v^-,v^+)~|~v\in V(G)\}\cup \{(u^+,v^-)^k~|~u\in N_G^-(v),v\in V(G)\}\cup\{(v^+,u^-)^k~|~u\in N_G^+(v),v\in V(G)\}$.
  We further set $k' = k$ and $\ell' = 2\ell$; then $(G',k',\ell')$ is an instance of the arc-deletion variant of \DLCHS{}.

  In the forward direction, let $S$ be a set of at most $k$ vertices such that any simple cycle of $G - S$ has length at most $\ell$.
  We let $S' = \{(v^-,v^+)~|~v\in S\}$, it follows that $S$ is a set of at most~$k$ arcs such that any simple cycle of $G' - S'$ has length at most $\ell' = 2\ell$.

  In the backward direction, let $S'$ be a set of at most $k'$ arcs such that $G' - S'$ does not have any simple cycles of length strictly more than $\ell' = 2\ell$.
  We may assume that $S'$ only contains arcs of the from $(v^-,v^+)$ for some vertex $v\in V(G)$, as $S'$ contains at most $k$ arcs and there are $k + 1$ parallel arcs between any two vertices of $G'$ that correspond to distinct vertices of $G$.
  Therefore, the set $S = \{v~|~(v^-,v^+)\in S'\}$ is a set of at most $k$ vertices in $G$ such that $G - S$ does not have any simple cycles of length more than $\ell$.
\end{proof}

It is clear that the \DLCHS{} problem generalizes the DFVS problem for parameter $\ell = 0$.
We now show that this problem also generalizes the {\sc Feedback Vertex Set in Mixed Graphs} problem, but this time for the parameter $\ell = 2$.
A mixed graph $G = (V, A, E)$ is a graph on a vertex set $V$ that has a set of directed arcs $A$, as well as a set of undirected edges $E$.
Recall that in the FVS problem in mixed graphs, we are given as input a mixed graph $G = (V,A,E)$, where each arc in $A$ can be traversed only along its direction and each edge in $E$ can be traversed in both directions, together with an integer $k$, and we are seeking a set $S$ of at most $k$ vertices such that $G - S$ does not contain any cycles.

\begin{theorem}
  There is a polynomial parameter transformation from instances $(G = (V,A,E), k)$ of {\sc Feedback Vertex Set in Mixed Graphs} to an instance $(G', k,2)$ of {\sc Directed Long Cycle Hitting Set}.
\end{theorem}
\begin{proof}
  Let $(G = (V,A,E), k)$ be an instance of the {\sc Feedback Vertex Set in Mixed Graphs} problem.
  We can assume that $G$ is loop-free, as vertices with loops need to be removed in any solution.
  Now, we will create a digraph $G'$ such that $(G', k, 2)$ as instance of \DLCHS{} has a solution if and only if $(G, k)$ has one.
  For this we replace every arc $a \in A$ by a path $P_a$ of length two in the same direction.
  Afterwards we replace all edges depending on the existence of other arcs/edges between it's endpoints:
  If for an edge $e = \{v, w\} \in E$ the only arc/edge between $v$ and $w$ in $G$ is $e$ (i.e. $G[\{v,w\}]$ contains a cycle) we replace $e$ by arcs $\overrightarrow{e} = (v,w)$ and $\overleftarrow{e} = (w,v)$ in both directions.
  Otherwise we replace $e = \{v, w\} \in E$ by two paths $\overrightarrow{P_e}$ and $\overleftarrow{P_e}$ of length two in both directions.
  The resulting graph is~$G'$.

  Let now $S$ be a solution to $(G, k)$.
  Then the only cycles in $G' - S$ must be those formed by replacing an edge with forward and backward paths/arcs.
  Only the edges replaced by paths can form cycles of length longer than two.
  But those edges had another arc/edge between their endpoints, thus forming a cycle in $G$.
  As $S$ intersects this cycle, only cycles of length two survive in $G' - S$.

  For the opposite direction, let $S'$ be a solution to $(G', k ,2)$.
  We can assume that $S' \subseteq V(G)$ as all other vertices lie in the middle of paths (i.e. have degree two) and we could include an endpoint of the path instead.
  As cycles in $G$ get replaced by longer cycles in $G'$, $G-S'$ contains only cycles of length two which don't get longer when transforming to $G'$.
  These cycles can only contain two edges between the same vertices (as arcs get longer).
  But these get replaced by paths so the cycles would have length at least four in $G'$ and would be deleted by $S'$.
  Thus, $G - S'$  contains no cycles.
\end{proof}

\section{Discussion}
\label{sec:discussion}
In this paper we have settled the parameterized complexity of hitting long cycles in directed graphs.
Our main result is a single-exponential fixed-parameter algorithm for this problem, which generalizes the breakthrough result by Chen et al.~\cite{ChenEtAl2008} for the setting of hitting all cycles in digraphs.
The algorithm also generalizes the fixed-parameter tractability result~\cite{BonsmaLokshtanov2011} by Bonsma and Lokshtanov for hitting cycles in mixed graphs.

Along the way, we showed how to compute a representative set of $x\to y$-paths, that is, a collection of paths such that if an (unknown) set $S$ of at most $k$ vertices does not disconnect $y$ from $x$, then there is at least one $x\to y$-path disjoint from~$S$ in our collection.
The collection has size $\ell^{\mathcal{O}(k^2\log k)}\cdot \log n$ on directed graphs without cycles of length greater than $\ell$.
We believe this result can find applications beyond the problem discussed here.

It would be interesting if the run time of algorithm can be improved.
Precisely, can we find a set hitting all cycles of length $\ell$ in time $2^{\mathcal O(\ell + k\log k)}\cdot n^{\mathcal{O}(1)}$ to match the best known run times for the cases of DFVS (where $\ell = 0$) and {\sc Long Directed Cycle} (where $k = 0$)?

\bibliographystyle{abbrvnat}
\bibliography{long_cycle_hitting_set}

%
%

\end{document}